\documentclass[12pt]{article}

\usepackage{fullpage,amsfonts,amsmath,amsthm,amssymb,mathrsfs,graphicx}

\usepackage[usenames]{color}

\usepackage[margin=1in]{geometry}

\begin{document}

\newcommand{\mm}[1]{{\color{blue}{#1}}}

\definecolor{pink}{rgb}{1,0.08,0.45}
\newcommand{\refc}[1]{{\color{pink}{#1}}}

\def\dd{r}
\def\ld{{\widehat L}}
\def\Gcal{\mathcal G}
\def\eps{\epsilon}
\def\eqn{\eqref}
\def\wtau{\widehat \tau}
\def\ttau{\widetilde \tau}
\def\Q{\mathcal Q}
\def\hG{\widehat H}
\def\C{\mathcal C}
\def\ds{{\mathcal D}}
\def\br{br}
\def\M{{\cal M}}
\def\H{{\cal H}}
\def\bell{\ensuremath{\boldsymbol\ell}}
\def\Bin{{\bf Bin}}

\newcommand{\be}{\begin{equation}}
\newcommand{\ee}{\end{equation}}
\newcommand{\bea}{\begin{eqnarray}}
\newcommand{\eea}{\end{eqnarray}}
\newcommand{\bean}{\begin{eqnarray*}}
\newcommand{\eean}{\end{eqnarray*}}
\newcommand{\non}{\nonumber}
\newcommand{\no}{\noindent}
\newcommand\floor[1]{{\lfloor #1 \rfloor}}
\newcommand\ceil[1]{{\lceil #1 \rceil}}
\newcommand{\jc}[1]{{{\bf [Jane:}\ \color{red} #1} {\bf ]}}
\newcommand{\jt}[1]{{\color{red}\ #1 }}
\newcommand{\jn}[1]{{\color{pink}\ #1 }}
\newcommand{\remove}[1]{}
\newcommand{\lab}[1]{\label{#1}\ }

\def\a{\alpha}
\def\b{\beta}
\def\c{\chi}
\def\d{\delta}
\def\D{\Delta}
\def\e{\epsilon}
\def\f{\phi}
\def\F{\Phi}
\def\g{\gamma}
\def\G{\Gamma}
\def\k{\kappa}
\def\K{\Kappa}
\def\z{\zeta}
\def\th{\theta}
\def\Th{\Theta}
\def\l{\lambda}
\def\la{\lambda}
\def\La{\Lambda}
\def\m{\mu}
\def\n{\nu}
\def\p{\pi}
\def\P{\Pi}
\def\r{\rho}
\def\R{\Rho}
\def\s{\sigma}
\def\S{\Sigma}
\def\t{\tau}
\def\om{\omega}
\def\Om{\Omega}
\def\smallo{{\rm o}}
\def\bigo{{\rm O}}
\def\to{\rightarrow}
\def\E{{\bf Exp}}
\def\ex{{\mathbb E}}
\def\cd{{\cal D}}
\def\rme{{\rm e}}
\def\hf{{1\over2}}
\def\R{{\bf  R}}
\def\cala{{\cal A}}
\def\cale{{\cal E}}
\def\call{{\cal L}}
\def\cald{{\cal D}}
\def\calz{{\cal Z}}
\def\calf{{\cal F}}
\def\Fscr{{\cal F}}
\def\cc{{\cal C}}
\def\calc{{\cal C}}
\def\calh{{\cal H}}
\def\calk{{\cal K}}
\def\cals{{\cal S}}
\def\calr{{\cal R}}
\def\calt{{\cal T}}
\def\msq{{\mathscr Q}}
\def\bk{\backslash}

\def\out{{\rm Out}}
\def\temp{{\rm Temp}}
\def\overused{{\rm Overused}}
\def\big{{\rm Big}}
\def\moderate{{\rm Moderate}}
\def\swappable{{\rm Swappable}}
\def\candidate{{\rm Candidate}}
\def\bad{{\rm Bad}}
\def\crit{{\rm Crit}}
\def\col{{\rm Col}}
\def\dist{{\rm dist}}
\def\poly{{\rm poly}}

\def\tdT{{\widetilde\Theta}}

\newcommand{\Exp}{\mbox{\bf Exp}}
\newcommand{\var}{\mbox{\bf Var}}
\newcommand{\pr}{\mbox{\bf Pr}}

\newtheorem{lemma}{Lemma}
\newtheorem{theorem}{Theorem}
\newtheorem{corollary}[lemma]{Corollary}
\newtheorem{claim}[lemma]{Claim}
\newtheorem{remark}[lemma]{Remark}
\newtheorem{proposition}[lemma]{Proposition}
\newtheorem{observation}[lemma]{Observation}
\theoremstyle{definition}
\newtheorem{definition}[lemma]{Definition}

\newcommand{\limninf}{\lim_{n \rightarrow \infty}}
\newcommand{\proofstart}{{\bf Proof\hspace{2em}}}
\newcommand{\tset}{\mbox{$\cal T$}}
\newcommand{\proofend}{\hspace*{\fill}\mbox{$\Box$}}
\newcommand{\bfm}[1]{\mbox{\boldmath $#1$}}
\newcommand{\reals}{\mbox{\bfm{R}}}
\newcommand{\expect}{\mbox{\bf Exp}}
\newcommand{\he}{\hat{\e}}
\newcommand{\card}[1]{\mbox{$|#1|$}}
\newcommand{\rup}[1]{\mbox{$\lceil{ #1}\rceil$}}
\newcommand{\rdn}[1]{\mbox{$\lfloor{ #1}\rfloor$}}
\newcommand{\ov}[1]{\mbox{$\overline{ #1}$}}
\newcommand{\inv}[1]{\frac{1}{#1} }
\newcommand{\imax}{I_{\rm max}}

\newcommand{\whp}{w.h.p.}
\newcommand{\aas}{a.a.s.\ }

\date{\empty}

\title{Inside the clustering threshold for random linear equations}

\author{Pu Gao\footnote{Research supported by an NSERC Postdoctoral Fellowship.}
and Michael Molloy\footnote{Research supported by an NSERC Discovery Grant and Accelerator Supplement.} \\
Department of Computer Science, University of Toronto\\
10 King's College Road, Toronto, ON\\
pu.gao@utoronto.ca \hspace{5ex} molloy@cs.toronto.edu}

\maketitle

 \thispagestyle{empty}

\begin{abstract}  We study a random system of  $cn$ linear equations over $n$ variables in GF(2), where each equation contains exactly $r$ variables; this is equivalent to $r$-XORSAT. \cite{ikkm,amxor} determined the clustering threshold, $c^*_r$: if $c=c^*_r+\e$ for any constant $\e>0$, then \aas the solutions partition into well-connected, well-separated {\em clusters}  (with probability tending to 1 as $n\rightarrow\infty$).  This is part of a general clustering phenomenon which  is hypothesized to arise in  most of the commonly studied models of random constraint satisfaction problems, via sophisticated but mostly non-rigorous techniques from statistical physics.
We extend that study to the range $c=c^*_r+o(1)$, showing that if $c=c^*_r+n^{-\d}, \d>0$, then the connectivity parameter of each $r$-XORSAT cluster is $n^{\Theta(\d)}$, as compared to $O(\log n)$ when $c=c^*_r+\e$. This means that one can move between any two solutions in the same cluster via a sequence of solutions where consecutive solutions differ on at most $n^{\Theta(\d)}$ variables; this is tight up to the implicit constant.  In contrast, moving to a solution in another cluster requires that some pair of consecutive solutions differ in at least $n^{1-O(\d)}$ variables.

Along the way, we prove that in  a random $r$-uniform hypergraph with edge-density $n^{-\d}$ above the $k$-core threshold, \aas every vertex not in the $k$-core can be removed by a sequence of $n^{\Theta(\d)}$ vertex-deletions in which the deleted vertex has degree less than $k$; again, this is tight up to the implicit constant.
\end{abstract}

\newpage

\setcounter{page}{1}

\section{Introduction} The study of random constraint satisfaction problems (CSP's) has been revolutionized by a collection of hypotheses, arising from statistical physics, concerning {\em clusters} of solutions.  Roughly speaking: when the density (the ratio of the number of constraints to $n$, the number of variables) exceeds a specific constant (the {\em clustering threshold}), the set of solutions \aas\footnote{ A property holds \aas (asymptotically almost surely)
if it holds with probability tending to 1 as $n\rightarrow\infty$.} partitions into clusters.  One can move throughout any cluster by making small {\em local} changes; i.e. changing the values of $o(n)$ variables in each step.  But solutions in two different clusters must differ {\em globally}; i.e. they differ on a linear number of variables.

While these hypotheses are, for the most part, not rigorously proven, they come from some substantial mathematical analysis.  They explain many phenomena, most notably why these CSP's are algorithmically very challenging (eg.~\cite{aco,dg}).  Intuition gained from these hypotheses  has led to  some very impressive heuristics (eg. Survey Propogation\cite{sp,mz,bmz}), the best of the random $r$-SAT algorithms whose performance has been rigorously proven~\cite{coalg}, and some remarkably tight bounds on various satisfiability thresholds (the density above which \aas there are no solutions)~\cite{cop,cop2,cov}.  It is clear that, in order to approach most of the outstanding challenges around random CSP's, we need to  understand clustering.

Amongst the commonly studied random CSP's, $r$-XORSAT (a.k.a.\ linear equations over GF(2)) is the one for which the clustering picture is, rigorously, the most well-established. The exact satisfiability threshold was established for $r=3$ in \cite{dub}, and then for $r\geq 4$ in \cite{cuc,ps}.  The exact clustering threshold, $c^*_r$, was established independently in \cite{ikkm} and in \cite{amxor}.
Those two papers also provide a very thorough description of the clusters for any constant density $c>c_r^*$: roughly speaking, each cluster consists of a solution to the 2-core  along with all satisfying extensions of that solution to the remainder of the variables; see Section~\ref{ssc} for a more thorough (and correct) description.  It is possible to move within each cluster changing only $O(\log n)$ variables at each step, but moving from one cluster to another requires changing a linear number of variables.

{The cluster structure of $r$-XORSAT is simpler than that of most other models (see Section~\ref{sirw}), and this  has enabled researchers to prove challenging results for it which
we do not appear close to proving for other models.  The structure of the clusters in the other models are hypothesized to be a generalization of the simpler structure in $r$-XORSAT.  After the {\em freezing threshold}, most clusters consist of a solution $\s$ to a subset of the variables, called {\em frozen variables}, along with all extensions of $\s$ to the rest of the variables. One (of many) differences from $r$-XORSAT is that, in other models, the set of frozen variables can differ amongst the clusters.  Insights from the work on $r$-XORSAT have been valuable when studying more complicated models; eg. ideas from~\cite{amxor} led to~\cite{mmfreeze,mres}.

In this paper we look inside the  $r$-XORSAT clustering threshold, by analyzing the clusters when $c=c_r^*+o(1)$.  Our goal is to gain a better understanding of the way that clusters are {\em born}; i.e. to understand the first clusters to arise as the constraint density is increased.
Looking inside thresholds for such purposes is a common goal; see e.g.\  the extensive work on the birth of the giant component\cite{bb3,tlcomp,lpw,jlkp}, and the look inside the 2-SAT threshold\cite{bbckw}.

Specifically, we look at $c=c^*_r+n^{-\d}$ for  constant $\d>0$.  The most notable difference we find is that moving within a cluster requires changing
$n^{\Theta(\d)}$ variables during at least one step.
Most of the technical work is in analyzing  the 2-core of a random hypergraph, which is where we begin.

\subsection{Hypergraph $k$-cores}

We define $\calh_{\dd}(n,p)$ to be an $r$-uniform hypergraph on $n$ vertices where each of the $n\choose r$ potential hyperedges is present independently with probability $p$.  This random hypergraph underlies the most commonly studied random CSP models, and the typical range of focus is $p=\Theta(n^{1-r})$, where \aas there is a linear number of hyperedges.

\begin{definition}
\normalfont The {\em $k$-core} of a hypergraph $H$, $\calc_k(H)$, is the {maximum} subgraph of $H$ in which every vertex has degree at least $k$.
\end{definition}

The $k$-core was first studied by Bollob\'as\cite{bbcore} and has since become a major focus in random graph theory.  Its many applications include erasure codes\cite{lmss,lmss2}, colouring\cite{mr2}, hashing\cite{hmwc}, and orientability\cite{pw,csw,fkp,fern}.
The threshold for the appearance of a non-empty $k$-core in $\calh_{\dd}(n,p=c/n^{\dd-1})$ is determined\cite{psw,mmcore,jhk} to be:
\begin{equation}\lab{krthreshold}
c_{\dd,k}=\inf_{\mu > 0}
 \frac{(\dd-1)!\mu}{\left[e^{-\mu}\sum_{i = {k}}^{\infty} \mu^i/i!\right]^{\dd-1}}
 \enspace .
 \end{equation}
The $r$-XORSAT clustering threshold is equal to $c_{r,2}$; so we use that notation rather than $c^*_r$ in the remainder.

One way to find the $k$-core is to repeatedly remove vertices of degree less than $k$ until no such vertices remain.
It is not hard to see that the order in which vertices are removed does not affect the subgraph produced.
 It is often useful to consider removing the vertices in several {parallel} rounds:

\begin{definition} The {\em parallel $k$-stripping process}, applied to a hypergraph $H$, consists of iteratively removing {\em all} vertices of degree less than $k$ at once along with any hyperedges containing any of those vertices, until no vertices of degree less than $k$ remain; i.e. until we are left with the $k$-core of $H$.
We use $S_i$ to denote the vertices that are removed during iteration $i$. We use $\widehat{H}_i$ to denote the hypergraph remaining after $i-1$ iterations, i.e. after removing $S_1,...,S_{i-1}$.
\end{definition}

We will analyze the number of rounds that this process takes:

\begin{definition}  The {\em $k$-stripping number of $G$}, denoted $s_k(G)$, is the number of iterations in the  parallel $k$-stripping process applied to $G$.  We often drop the ``$k$'' and speak of the stripping number and $s(G)$.
\end{definition}

We are  interested in the number of deletions that are required to remove a particular vertex:

\begin{definition} A {\em $k$-stripping sequence} is a sequence of vertices that can be deleted from a hypergraph, one-at-a-time, {along with their incident hyperedges} such that at the time of deletion each vertex has degree less than $k$.
For any vertex $v$ not in the $k$-core, the {\em depth} of $v$ is the length of a shortest $k$-stripping sequence ending with $v$.
\end{definition}

For the purposes of studying $r$-XORSAT clusters, we only need to analyze these parameters for the case $k=2$.  But  the stripping number and depth are of more general interest, so we analyze them for all $k$.
(The case $k=r=2$, i.e. the 2-core of $G_{n,p}$ is very well understood, so we do not cover it here; see Appendix~\ref{sa1}.)
In order to prove that one can move within clusters, changing $O(\log n)$ variables at a time,
\cite{amxor} proved that for $c>c_{k,r}$, \aas every non-$k$-core vertex has depth $O(\log n)$.
(\cite{ikkm} implicitly proves that for $k=2$ the maximum depth is poly$(\log n)$.) The main technical work of this paper is to prove that inside the $k$-core threshold, the maximum depth rises to $n^{\Theta(1)}$:

\begin{theorem}\lab{mt}
For  $\dd,k\geq2, (\dd,k)\neq (2,2)$, consider any $0<\d<\hf$.
\begin{enumerate}
\item[(a)]  For $|c-c_{\dd,k}| \leq n^{-\d}$, \aas $s(\calh_{\dd}(n,p=c/n^{\dd-1}))\geq\Omega(n^{\d/2})$.
\item[(b)]  For $c=c_{\dd,k}+n^{-\d}$,  \aas
$ s(\calh_{\dd}(n,p=c/n^{\dd-1}))= O(n^{\d/2}\log n)$.
\end{enumerate}
\end{theorem}

\begin{theorem}\lab{mt2}
For  $\dd,k\geq2, (\dd,k)\neq (2,2)$, consider any $0<\d<\hf$.
For $c=c_{\dd,k}+n^{-\d}$, \aas
the maximum depth in $\calh_{\dd}(n,p=c/n^{\dd-1})$ is between $\Omega(n^{\d/2})$ and $n^{O(\d)}$.
\end{theorem}

Our requirement $\d<\hf$ comes from the similar restriction on the actual appearance of the $k$-core, as shown by Kim\cite{jhk} (see also~\cite{dmcore}):

\begin{theorem}\lab{tkim}\cite{jhk} For  $\dd,k\geq2, (\dd,k)\neq (2,2)$, consider any $0<\d<\hf$.
\begin{enumerate}
\item[(a)] For $c=c_{\dd,k}-n^{-\d}$, \aas the $k$-core  of $\calh_{\dd}(n,p=c/n^{\dd-1})$ is empty.
\item[(b)] For $c=c_{\dd,k}+n^{-\d}$, \aas $\calh_{\dd}(n,p=c/n^{\dd-1})$ has a $k$-core with $\Theta(n)$ vertices.
\end{enumerate}
\end{theorem}

The most challenging difference between the setting of this paper and that of~\cite{amxor} is:
When $c=c_{r,k}+\Theta(1)$, the number of low degree vertices remaining after each round of the parallel stripping process drops geometrically (see the comment following Lemma~\ref{llt1}).
That easily implies that \aas the stripping number is $O(\log n)$.  Furthermore, it is the key property in the analysis of~\cite{amxor}  proving that the depth is $O(\log n)$.  For $c=c_{r,k}+n^{-\d}$, the number of such vertices drops much more slowly.  This requires us to use a very different, and much more intricate, analysis.

\subsection{Solution clusters}\label{ssc} Our model for a random system of equations  is $X_r(n,p)$ defined as follows.  We have $n$ variables over GF(2).  Each of the $n\choose r$ $r$-tuples of variables is chosen to form the LHS of an equation with probability $p$.  For each equation, the RHS is chosen uniformly from $\{0,1\}$.  We focus on the case $p=c/n^{r-1}$ where \aas there are a linear number of equations.  We restrict ourselves to the case $r\geq 3$, as the case $r=2$ behaves very differently, and is already well-understood (see eg.\cite{bbckw}).

It is not hard to see that this is equivalent to choosing an instance of $r$-XORSAT  on $n$ variables by choosing each $r$-tuple to form a clause with probability $p$, and then signing the variables within each clause uniformly at random.  An assignment to the variables is satisfying if every clause contains an odd number of true literals.

{An common alternate model is to choose $M=dn$ $r$-tuples of variables, and then form an equation for each $r$-tuple (with a uniformly random RHS).  Standard techniques allow us to translate our results to that model with $d:=c/r!$.}

The following definitions come from~\cite{amxor} and are equivalent to the cluster definition in~\cite{ikkm}.

\begin{definition}
The {\em underlying hypergraph} of a system of linear equations over GF(2) is defined as follows: the vertices are the variables and for each equation, the set of variables appearing in
that equation form a hyperedge.  The {\em 2-core} of a system of linear equations is the subset of equations corresponding to the hyperedges that are in the 2-core of the underlying hypergraph.
\end{definition}

Thus the underlying hypergraph of $X_r(n,p)$ is $\calh_r(n,p)$.
We speak about the variables/equations of a system of equations and the vertices/hyperedges of the underlying hypergraph interchangably.  Note that every solution to the 2-core can easily be extended to a solution of the entire system, by setting the other variables in the reverse order that they are removed.

Roughly speaking, the clusters correspond to solutions of the  2-core.
But this is not quite true - we need to account for the effects of short {\em flippable cycles} which we define as follows:

\begin{definition}\label{flippable_def}
A \emph{flippable cycle} in a hypergraph {$H$} is a set of vertices $S=\{v_1,\ldots,v_t\}$
{where the set of edges incident to $S$ can be ordered as $e_1, \ldots e_t$ such that}
each vertex $v_i$ lies in $e_i$ and in $e_{i+1}$ and in no other {edges of $H$}
 (addition mod $t$).
\end{definition}
Thus, the vertices $v_1,\ldots,v_t$ must have degree exactly two in the hypergraph. The remaining vertices in edges $e_1,\ldots,e_t$ can have arbitrary degree and are \emph{not} part of the  flippable cycle.  Note that if we take a solution $\s$ to the entire system, and change the assignment to each variable in a flippable cycle of the underlying hypergraph, we obtain another solution. It is easy to see that
no vertex can lie in  two flippable cycles.

\begin{definition}\label{coreflippable_def}  A \emph{core flippable cycle} in  a hypergraph $H$ is a flippable cycle in the subhypergraph  induced by the 2-core of $H$.
\end{definition}

Thus, in a core flippable cycle, the vertices $v_1,\ldots,v_t$ have degree exactly two \emph{in the 2-core}, but possibly higher degree in $H$. Note also that {$H$ may contain} flippable cycles outside the 2-core.  In Section~\ref{scon}, we show that very few vertices lie on core flippable cycles.

\begin{lemma}\lab{lflip} For $r\geq 3$, $0<\d<\hf$ and $c=c_{\dd,2}+n^{-\d}$, \aas
the total sizes of all core flippable cycles in $\calh_{\dd}(n,p=c/n^{r-1})$ is at most $O(n^{\d/2}\log n)$.
\end{lemma}

\begin{definition}
Two solutions are \emph{cycle-equivalent} if on the 2-core they differ only on variables in core flippable cycles ({while} they may differ arbitrarily on variables not in the 2-core).
\end{definition}

\begin{definition}\label{cluster_def}
The {\em solution clusters} of {$X_r(n,p={c/n^{r-1}})$} are the cycle-equivalence classes, i.e., two solutions are in the same cluster iff they are cycle-equivalent.
\end{definition}

In other words: Let $\s$ be any solution to the subsystem induced by the 2-core.  It is easy to see that $\s$ can be extended to a solution to the entire system of equations.  All such extensions, along with all extensions of any 2-core solutions obtained by altering $\s$ only on core flippable cycles, form a cluster.
By symmetry, all clusters are isomorphic.
Note that, if the 2-core is empty, then our definitions imply that all solutions are in the same cluster. {So the clustering threshold, $c_{2,r}$, is the density above which there are many (in fact, exponentially many) clusters rather than one.}

\begin{definition}
Two solutions $\sigma,\tau$  are \emph{$d$-connected} if there exists a sequence of solutions $\s,\s',\ldots,\t$ such that every pair of successive solutions differs on at most $d$ variables.
\end{definition}

\begin{theorem}\lab{tc1}  For $\dd\geq 3$,  there exists $\k=\k(r), Z=Z(r)>0$  such that: For $\d>0$ and $c=c_{\dd,2}+n^{-\d}$,  in $X_{\dd}(n,p=c/n^{\dd-1})$ \aas :
\begin{enumerate}
\item[(a)] every pair of solutions $\s,\t$ in the same cluster is $n^{\k\d}$-connected;
\item[(b)] every pair of solutions $\s,\t$ in different clusters is not $Zn^{1-r\d}$-connected.
\end{enumerate}
\end{theorem}

Thus, we can travel throughout any cluster by changing only $n^{\k\d}$ variables in each step.
But to travel from one cluster to another cluster requires changing  at least $n^{1-r\d}$ variables in one of the steps.
Theorem~\ref{tc1}(a) is best possible, up to the value of $\k$:

\begin{theorem}\lab{tc2}  For $\dd\geq 3$, $0<\d<\hf$ and $c=c_{\dd,2}+n^{-\d}$,  \aas there exists a pair of solutions $\s,\t$ in the same cluster of $X_{\dd}(n,p=c/n^{\dd-1})$ such that $\s,\t$ is not $n^{\d/20}$-connected.
\end{theorem}

It is not clear whether Theorem~\ref{tc1}(b) is tight. Our analysis does not exclude the possibility that different clusters are, in fact, linearly separated; i.e. that there is some constant $\a>0$ such that \aas every pair of solutions $\s,\t$ in different clusters is not $\a n$-connected.

Note that Theorem~\ref{tc1} holds trivially when $\d\geq\inv{\k}$ and so it only says something interesting about clustering for sufficiently small $\d$ (see Section~\ref{sfd}).

\subsection{Related Work}\lab{sirw}
The solution clusters for random $r$-XORSAT were analyzed independently in~\cite{cdmm,mez};  much of the work was rigorous, but some key steps were missing. The description of clusters given there was essentially equivalent to what was eventually proven in~\cite{ikkm,amxor}, the main difference being the (relatively minor) effects of short cycles.
As described above, these papers establish the clustering threshold, along with a very detailed description of the clusters when $c$ exceeds the clustering threshold by a constant.

Early rigorous bounds on the satisfiability threshold for random $r$-XORSAT appeared in~\cite{nc,vk1,vk2,cdxor}.  The threshold was established exactly for $r=3$~\cite{dub}, and for $r\geq 4$~\cite{cuc, ps}.

Clustering for other random CSP's is much more complicated, and less has been established rigorously.  There is very extensive non-rigorous work (see, eg.~\cite{mmbook,sp, kmrsz, zk, mpwz}).  Some key rigorous results include: \cite{mmz,daud} proved that $r$-SAT has well-separated clusters.  \cite{aco} determined that such clusters arise in $r$-SAT, $r$-COL and $r$-NAESAT by the time the density is, asymptotically in $r$, no higher than the hypothesized location of the clustering threshold.  Notably, that density coincides (asymptotically in $r$) with the long-observed `algorithmic barrier', a density above which we know of no algorithms that are proven to find solutions. \cite{mrt} extends this analysis to a large family of CSP models. \cite{coind} shows a similar result for independent sets.

\cite{art} proved that at high densities, all $r$-SAT clusters contain frozen variables; i.e. variables that take the same value for all solutions in the cluster.  \cite{aco} determined the asymptotic (in $r$) location of the freezing threshold for $r$-COL, $r$-NAESAT and $r$-SAT; i.e. the density above which almost all clusters have frozen variables. \cite{mmfreeze} determined the exact value of the freezing threshold for $r$-COL;  this was extended to $r$-NAESAT and other CSP's in~\cite{mres}.
\cite{cz} established that {\em condensation} takes place in hypergraph 2-colouring.

We remark that clustering for $r$-XORSAT is much simpler than for the other commonly studied CSP's. All of the $r$-XORSAT clusters are isomorphic, the freezing threshold equals the clustering threshold, and each cluster has the same set of frozen variables.  To study the clusters, it suffices to analyze properties of the random hypergraph rather than random solutions.  This is what allows us to actually look inside the clustering threshold whereas we don't  know the exact clustering threshold for, eg. $r$-SAT, $r$-COL, etc. ($r\geq 3$).   Furthermore, the fact that the solutions can be represented algebraically is what allows us to analyze the connectivity of clusters, which we have been unable to do at all for other CSP's,

{\em A remark on asymptotic notation.}  Sometimes we use, eg. $A=B\pm \bigo(C)$ to emphasize that perhaps $A<B$.  Of course, this is redundant, since standard use of $\bigo(-)$ notation allows for this.

\section{Bounding the stripping number}\lab{smt1}

In this section, we sketch the proof of Theorem~\ref{mt} for the case $c=c_{r,k}+n^{-\d}$. Note that this is the case that is relevant to our study of the clustering threshold for random $r$-XORSAT. We extend the proof to the range $c_{r,k}-n^{-\d}\leq c < c_{r,k}+n^{-\d}$ in Appendix~\ref{scouple}.

We restate the parallel stripping process in a slowed-down version, which will be more convenient to analyze.   Rather than removing all vertices of $S_i$ at once, we remove them one at a time.  When removing a vertex, we slow down further by removing one edge at a time.
To facilitate this, we maintain a queue $\msq$ containing all deletable vertices; i.e. all vertices of degree less than $k$:

\begin{tabbing}
{\bf SLOW-STRIP}\\
{\bf Input:}  A hypergraph $G$.\\
Initialize: $t:=0,G_0:=G$, $\msq:=S_1$ is the set of all vertices of degree less than $k$ in $G$.\\
Whi\=le $\msq\neq\emptyset$:\\
\>Let $v$ be the next vertex in $\msq$.\\
\>Remove one of the hyperedges $e$ containing $v$.\\
\>If \=any vertex of $e$ now has degree 0 then remove that vertex from $G$ and from $\msq$.\\
\>If any other vertex of $e$  has its degree drop to below $k$ then add that vertex to the end of $\msq$.\\
\>$G_{t+1}$ is the resulting hypergraph; $t:=t+1$.\\
\end{tabbing}

At any point, $\msq$ may contain some vertices from $S_i$ and some from $S_{i+1}$.  However,
the vertices of $S_i$ are removed before the vertices of $S_{i+1}$.  Note also that when processing a vertex $v\in\msq$, all edges from $v$ are removed (and hence $v$ is removed) before moving to the next vertex of $\msq$.  So this procedure removes vertices in the same order as the parallel stripping process. In particular, { if $t$ is the iteration when the first vertex of $S_i$ is selected then $G_{t}$ is the graph remaining after $i-1$ iterations of the parallel stripping process and at time $t$, $\msq=S_{i}$.}

\begin{definition}  We use $t(i)$ to denote the iteration of SLOW-STRIP at which iteration $i$ of the parallel stripping process begins; i.e.\ when the first vertex of $S_i$ reaches the front of $\msq$.
\end{definition}

Lemma~\ref{l:B} in the appendix says that for any $\e>0$ there is a $B=B(\e)$ such that after step $B$  of the parallel process, or equivalently, step $t(B)$  of SLOW-STRIP, the size of the remaining graph $G_{t(B)}$ is at most $\e n$ greater than the size of the $k$-core. This implies that various parameters are very close to those of the $k$-core, and so it is often convenient to
focus on $t\geq t(B)$.

Our key parameter is $L_t$, the total degree of the vertices in $\msq$ at iteration $t$ of SLOW-STRIP.

\begin{lemma}\lab{llt1}  There are constants $B,K$ such that: If $c=c_{\dd,k}+n^{-\d}$, then \aas for every $t\geq t(B)$,
\[\ex(L_{t+1}{\mid G_t})\leq L_t-Kn^{-\d/2}.\]
\end{lemma}

{\em Remark:}  The key Lemma 32 of~\cite{amxor} implies that, when $c=c_{r,k}+\e$, we have \aas $\ex(L_{t+1}{\mid G_t})\leq -\z L_t$ for a constant $\z>0$.  That difference is what causes the stripping number and depth to rise from $O(\log n)$ above the clustering threshold to $n^{\Theta(1)}$ inside the clustering threshold.  It is also the cause of most of the difficulties in this paper.  {However, the weaker fact that $\ex(L_{t+1}-L_t{\mid G_t})$  remains bounded below $-Kn^{-\d/2}$,
no matter how small $L_i$ gets,  is still very useful to our analysis.}

Some of our more delicate analysis requires us to be a sublinear distance away from the $k$-core.  Theorem 1.7 of~\cite{jhk} gives a more detailed version of Theorem~\ref{tkim}, from which we can  specify constants $\a=\a(r,k),K_1=K_1(r,k)$ such that when $c=c_{r,k}+n^{-\d}$, \aas the size of the $k$-core is $\a n+K_1 n^{1-\d/2}+o(n^{1-\d/2})$.
See Lemma~\ref{lcoresize} in Appendix~\ref{sbsn} for more detail.

We define $t_0$ to be the first iteration of SLOW-STRIP in which the number of vertices in the remaining graph is  $\a n+3K_1 n^{1-\d/2}$.   Since we delete at most one vertex per iteration, there are at least $K_1n^{1-\d/2}$ iterations following $t_0$.
We will show that, throughout those iterations, $L_t$ is small.

\begin{lemma}\lab{llt2}  If {$c=c_{\dd,k}+n^{-\d}$, then \aas for every $t>t_0$,} $L_{t}\leq O(n^{1-\d}).$
\end{lemma}

This bounds the size of each $S_i$, thus implying that  many rounds of the parallel stripping process are required to remove all the remaining vertices, and so lower bounding  the stripping number.

{\em Proof of Theorem~\ref{mt} for $c=c_{\dd,k}+n^{-\d}$.}
We begin with the lower bound on the stripping number for part (a).   We run SLOW-STRIP and recall that this can be viewed as also running the parallel stripping process slowly. At  iteration $t_0$, $\a n+ 3K_1n^{1-\d/2}$ vertices remain in the graph. Therefore, we need to remove at least $K_1n^{1-\d/2}$ vertices before reaching the $k$-core.  By Lemma~\ref{llt2}, at most $|L_{t(j)}|=O(n^{1-\d})$ of them belong to $S_j$ for each $j$. Therefore, we require at least $K_1n^{1-\d/2}/O(n^{1-\d})=\Omega(n^{\d/2})$ iterations of the parallel stripping process to remove them all.

Now we turn to the upper bound in part (b). {We will focus on the change in $L_{t(i)}$, the sum, over all $v\in S_i$ of the degree of $v$ at the beginning of iteration $i$ of the parallel stripping process}.

Between iterations $t(i)$ and $t(i+1)-1$ of SLOW-STRIP, all hyperedges from all vertices in $S_i$ must be deleted. One hyperedge is removed in each iteration, and it touches at most $\dd$ members of $S_i$.  Thus $t(i+1)-t(i)\geq \inv{\dd} L_{t(i)}$.  So  Lemma~\ref{llt1} yields that  for all $i> B$:
\[\ex(L_{t(i+1)})\leq L_{t(i)}-(t(i+1)-t(i))Kn^{-\d/2}\leq L_{t(i)}-\inv{\dd} L_{t(i)}Kn^{-\d/2}
=L_{t(i)}(1-\frac{K}{\dd}n^{-\d/2}).\]
It follows that
\[\ex(L_{t(i)})\leq n(1-\frac{K}{\dd}n^{-\d/2})^{i-B-1}.\]
Thus, for $i>B+\frac{2\dd}{K} n^{\d/2}\log n$, $\ex(L_{t(i)})=o(1)$ and so \aas $L_{t(i)}=0$.
I.e., \aas the stripping number is at most $B+\frac{2\dd}{K} n^{\d/2}\log n=O( n^{\d/2}\log n)$.
\proofend

\section{Bounding the maximum depth}\lab{smd}

In this section, we sketch the proof of Theorem~\ref{mt2}. It is easy to see that the stripping number provides a lower bound on the depth (see Lemma~\ref{ls2}) and so Theorem~\ref{mt}(a) provides the lower bound of Theorem~\ref{mt2}.  We will focus on the upper bound.

  Recall that $S_i$ is the set of vertices removed during iteration $i$ of the parallel stripping process, and $\hG_i$ is the subhypergraph remaining after $S_1,...,S_{i-1}$ are removed.  We actually run SLOW-STRIP and so we remove the {hyperedges containing vertices of $S_i$} one-at-a-time. This allows us  to create a directed graph (not a directed hypergraph) $\cald$ as follows. {When we select $v\in \msq$ and remove a hyperedge $e$ containing $v$,
we add a directed edge to $v$ from each of the $r-1$ other vertices in $e$.}
For any vertex $v\in H$, we define $R^+(v)$ to be the set of vertices reachable from $v$ in $\cald$.  It is not hard to see that $R^+(v)$ forms a stripping sequence ending with $v$.  So our upper bound in Theorem~\ref{mt2} follows from:
\begin{equation}\lab{emt2}
|R^+(v)|\leq n^{O(\d)} \mbox{ for every non-$k$-core vertex } v.
\end{equation}

For any vertex $v$ in some $S_i$, we analyze $R^+(v)$ by a graph search starting at $v$  and continuing to explore along outedges into levels $S_{i-1}, S_{i-2},...,S_1$ in decreasing  sequence.  We do so by taking advantage of the randomness of those edges.  However, when stripping down to level $S_i$ in order to discover  that $v\in S_i$, we expose all of those edges and hence remove the randomness.   So we must use a very careful exposure.  First we expose the vertices in each level $S_1,S_2,...$, without exposing the actual edges.  After all the levels are exposed, we expose the edges.

We describe this exposure for the case $r=2$, i.e.\ for a graph.  The general description for hypergraphs is the natural generalization, but one has to be careful about the fact that a hyperedge can have multiple vertices in one level $S_j$; see Appendix~\ref{s2}.

Let $\calc_k$ denote the $k$-core.
For a vertex $v\in S_i$, we let $d_i(v), d^+(v)$, and $d^{++}(v)$ denote the number of neighbours that $v$ has in $S_i, S_{i+1}$ and $\calc_k\cup\left(\cup_{j\geq i+2} S_j\right)$.  These are all the edges  in $\hG_i$ incident with $v$ and so $d_i(v)+d^+(v)+d^{++}(v)\leq k-1$.

\vspace{1ex}

{\bf EXPOSURE:}
\begin{enumerate}
\item Expose the vertices in $S_1,S_2,...$.
\item For each $i\geq 1$ and each vertex $v\in S_i$, expose  $d_i(v), d^+(v), d^{++}(v)$.
\end{enumerate}

\vspace{1ex}

{\bf EDGE-SELECTION:}
\begin{enumerate}
\item For each $i\geq 1$, we expose the edges within $S_i$ by selecting a uniformly random graph on the given degree sequence; i.e. where each $v\in S_i$ has degree $d_i(v)$.
\item For each $i\geq 1$, we expose the edges between $S_i,S_{i+1}$ by selecting a uniformly random bipartite graph subject to:
\begin{enumerate}
\item every $v\in S_i$ lies in exactly $d^+(v)$ edges;
\item every $u\in S_{i+1}$ lies in at least one edge.
\end{enumerate}
\item For each $i\geq 1$ we expose the remaining edges incident with $S_i$ in $\hG_i$ as follows:  For each $v\in S_i$, we choose $d^{++}(v)$ uniformly random neighbours in $\calc_k\cup\left(\cup_{j\geq i+2} S_j\right)$.
\end{enumerate}

It is not hard to verify that the edges chosen by EDGE-SELECTION are chosen with the correct distribution (see appendix~\ref{s2}).

This exposes all of the hypergraph except the edges of the $k$-core. Those edges can be chosen by exposing the degree sequence  and then taking a random graph on that degree sequence.

Now our proof of~(\ref{emt2}) goes as follows.  First, we run EXPOSURE and prove that the sets $S_i$ and degrees satisfy certain conditions (Lemma~\ref{lsi} below). Then we pick some $v\in S_i$ and for $j=i-1,...,1$ we reveal $R^+(v)\cap S_j$.  As we do so, we expose the relevant edges via EDGE-SELECTION.  The randomness of those edges allows us to set up a recursive equation bounding $|R^+(v)\cap S_j|$ which we solve to bound $|R(v)|$. For the details, see Appendices~\ref{s2},~\ref{slsi},~\ref{srec}.

\section{Inside the clustering threshold}\lab{sct}
\subsection{Theorem~\ref{tc1}}
The proof of Theorem~\ref{tc1}(b)  is nearly identical to that of Theorem 2 in \cite{amxor}. We omit the repetitive details; instead, in Appendix~\ref{scon}, we  note the differences in the two proofs. This section will be focussed on Theorems~\ref{tc1}(a) and~\ref{tc2}.  {The argument
that our upper bound on the maximum depth implies that the clusters are well-connected is the same as that used in~\cite{amxor,ikkm}.  We include it here for exposition, and because it is needed to understand the proof of Theorem~\ref{tc2}.}

We first describe how to use Gaussian elimination to express every solution in a particular cluster
in terms of  {\em free variables}:

In each core flippable cycle, $C$, we choose an arbitrary variable (vertex) $v_C$.  We eliminate each of the equations (hyperedges) containing two variables of $C$ except for one of the two containing $v_C$;
the result is that each of the $|C|-1$ other variables of $C$ is expressed in terms of $v_C$  and
some of the other 2-core variables not in any flippable cycles.

Next, we  process non-2-core equations in the reverse order of their deletion during SLOW-STRIP.  When we remove a variable (vertex) $v$ of degree 1 in the remaining subsystem (subhypergraph), along with the only equation (hyperedge) containing $v$, we express $v$ in terms of the other $r-1$ variables of that equation.  If the equation contains any other variables of degree 1, then they become free variables; i.e. they are not expressed in terms of any other variables.  Note that such variables have indegree zero in $\cald$.   The non-free variables in the equation already have their own expressions, which  propogate to the expression for $v$.

\begin{definition} The {\em free variables} are: (i) the non-2-core vertices with indegree zero in $\cald$; (ii) one specified vertex, $v_C$ on each core flippable cycle $C$.
\end{definition}

Each cluster is specified by an assignment to all of the 2-core variables not in any core flippable cycles.  If we fix such an assignment,  every other variable $v$ is now expressed as
\[v=z_v+\sum_{u\in\chi (v)} u,\]
where $\chi(v)$ is a subset of the free variables, and $z_v\in\{0,1\}$ is determined by the assigment to the 2-core variables not in any core flippable cycles; hence $z_v$ is fixed within each cluster.  Note that $\chi(v)$ is the same for every cluster.  So each cluster can be specified by fixing the
values of $z_v$, and then taking the set of solutions to the diagonal system formed by
the equations $v=z_v+\sum_{u\in\chi (v)} u$. If $\calf$ is the set of free variables, then clearly, there is one such solution for each of the $2^{|\calf|}$ assignments $\calf\rightarrow\{0,1\}$. Thus, each cluster has size exactly $2^{|\calf|}$.

Theorem~\ref{tc1}(a) now follows easily from our proof of Theorem~\ref{mt2}, by observing that if $u\in \chi(v)$ then $v\in R^+(u)$:

{\em Proof of Theorem~\ref{tc1}(a):}  Within a cluster, we can travel from any solution to any other solution by changing the free variables one-at-a-time.  We will argue that changing one free variable results in a change of at most $n^{O(\d)}$ variables.

Every time we remove a variable $v$, we express it in terms of the $r-1$ other vertices in its equation.  We also add edges in $\cald$ to $v$ from each of those vertices.  It follows that every non-2-core variable $v$ can be reached, in $\cald$, from
every variable in $\chi(v)$.  Therefore, for each non-2-core free variable, $u$, every $v$ with $u\in\chi(v)$ is in $R^+(u)$. In our proof of  Theorem~\ref{mt2}, we showed that \aas $|R^+(u)|\leq n^{O(\d)}$ for all $u$.  Therefore, changing $u$ results in a change of at most $|\chi^{-1}(u)|\leq|R^+(u)|\leq n^{O(\d)}$  variables.

Now consider changing a free variable $v_C$ in some core flippable cycle $C$.  This will change the variables of $C$.  Every non-2-core vertex $v$ with $v_C\in \chi(v)$ must be in $R^+(u)$ for some $u$ which shares a hyperedge with some vertex of $C$.  By Lemma~\ref{lflip}, \aas $|C|<n^{\d/2}\log n$.  A.a.s.\ each vertex of $C$ shares a hyperedge with fewer than $\log n$ non-2-core vertices, since the maximum degree in $\calh_{\dd}(n,p=c/n^{k-1})$ is \aas less than $\log n$.  Since \aas $|R^+(u)|\leq n^{O(\d)}$ for all $u$, changing $v_C$ results in a change of at most $|C|\times\log n\times  n^{O(\d)}= n^{O(\d)}$ variables.
\proofend

\subsection{Theorem~\ref{tc2}}

Theorem~\ref{mt2} says that there is at least one variable $u$ with $|R^+(u)|\geq n^{\Theta(\d)}$.  However, this does not immediately imply that each solution cluster is
not $n^{\Theta(\d)}$-connected.  For one thing, we require that there is a {\em free} variable $u$ with such a large  $R^+(u)$.  In fact, we actually require $\chi^{-1}(u)\subseteq R^+(u)$ to be that large. But even that would not suffice, since it would only imply that a step where $u$ is the only free variable changed would require changing $n^{\Theta(\d)}$ other variables.  It still leaves open the possibility that one could change $u$ using a step that also changes other free variables $w_1,...,w_t$ where $\cup\chi(w_i)$ intersects $\chi(u)$ and so not every variable in $\chi(u)$ is changed.

To prove Theorem~\ref{tc2} we prove that there is a free variable $u$, along with $n^{\d/20}$ other variables $v_1,...v_{n^{\d/20}}$ such that for each $i$, $\chi(v_i)=\{u\}$.  In order to move through every solution in a cluster, eventually there must be a step where $u$ is changed.  At that step, each of $v_1,...v_{n^{\d/20}}$ are changed as well, regardless of which other free variables are changed.  So the cluster is not $n^{\d/20}$-connected.

To do this,
we let $u^*$ be a non-core free vertex that is removed latest in the parallel stripping process.
Specifically, let $I^*$ be the largest value of $i$ such that $S_i$ contains a vertex with indegree zero in $\cald$, and let $u^*$ be some such vertex.

Our first step is to prove that there are no other free vertices within $n^{\d/20}$ levels of $u^*$:

\begin{lemma}\lab{lgap} A.a.s\ $u^*$ is the only free vertex in
$\cup_{i\geq I^*-n^{\d/20}} S_i$.
\end{lemma}

Given a non-core vertex $w$, we define $T(w)$ to be the set of vertices $v$ that can reach $w$ in $\cald$; i.e. the set of vertices $v$ such that $w\in R^+(v)$. For $u\in T(w)$, define $T(w,u)$ be the subgraph of $T(w)$ containing all vertices reachable from $u$; i.e. vertices on walks from $u$ to $w$.  We will prove:

\begin{lemma}\lab{lTw}  A.a.s. there is some $w\in S_{I^*-n^{\d/20}}$ such that
(a) $u^*\in T(w)$;
(b) no vertex of any core flippable cycle is in $T(w)$;
(c) the subgraph of $\cald$ induced by the vertices of $T(w,u^*)$ is a path.
\end{lemma}

{\em Proof of Theorem~\ref{tc2}(a)}  Consider $w$ from Lemma~\ref{lTw} and  any vertex $v\in T(w,u^*)$.
Since $T(w,u^*)$ induces a path in $\cald$, there is exactly one path from $u^*$ to $v$, and so $u^*\in \chi(v)$.

Since $T(v)\subseteq T(w)$, Lemma~\ref{lTw}(b) implies that no vertex of any core flippable cycle is in $T(v)$.  By Lemma~\ref{lgap}, $u^*$ is the only vertex in $\cup_{i\geq I^*-n^{\d/20}} S_i$
with indegree zero in $\cald$.
Therefore, $u^*$ is the only free variable in $T(v)$ and so $\chi(v)=\{u^*\}$.

Therefore, if any two solutions  in the same cluster differ on $u^*$ then they differ on all of the
$n^{\d/20}+1$ variables on the path from $u^*$ to $w$.  Since every cluster contains one solution for each setting of the free variables, this implies that every cluster is not $n^{\d/20}$-connected.
\proofend

The proofs of these Lemmas appear in Appendix~\ref{scon}.

\section{Future directions}\lab{sfd}
We have examined the solution clusters of $X_r(n,p=c/n^{r-1})$ when $c$ is within the clustering threshold; specifically $c=c_{r,2}+n^{-\d}$.  We showed that \aas the connectivity parameter of each cluster is $n^{\Theta(\d)}$, as opposed to $O(\log n)$ when $c$ is a constant greater than the threshold $c_{r,2}$.  We have also shown that \aas any two solutions in different clusters are not $n^{1-r\d}$-connected.   So if $\d$ is sufficiently small, solutions in the same cluster are $n^{\e}$-connected, while solutions in different clusters are not even $n^{1-\e}$-connected where $\e$ is small.

It is possible that the clusters are even more separated than we have shown. { We would like to know whether there is \aas a pair of solutions in different clusters that is $n^{1-z\d}$-connected for some constant $z>0$ .}

More importantly, we would like to see what happens for larger $\d$.  In particular, we would like to find out whether clusters arise as soon as the first non-empty 2-core appears in the underlying hypergraph. Is there always some $f(n)=o(g(n))$ such that if there is a 2-core then, under the cluster definitions from Section~\ref{ssc}, any two solutions in the same cluster are $f(n)$-connected, while any two solutions in different clusters are not $g(n)$-connected?  Or do we need to change the definition of clusters?  Or perhaps our notion of clustering falls apart if we are sufficiently close to the point when the 2-core is formed.

{Of course, instances of $r$-XORSAT can be solved in polynomial time, using global algorithms such as Gaussian elimination.  However, when $c$ is above the clustering threshold, random $r$-XORSAT seems to be very difficult for generic CSP solvers and local algorithms such as WalkSat\cite{young}.  It is natural to wonder whether such difficulties arise precisely when the 2-core appears.  Answering this question, and understanding exactly what these difficulties are, could provide insights into how it is that  clustering creates algorithmic difficulties for other random CSP's.  Resolving the issues discussed in the previous paragraph would be very helpful for this question.}

\newpage

\begin{center}
{\bf Appendix}
\end{center}

\section{The case $k=r=2$}\lab{sa1}

 The 2-core of the random graph $G(n,p=c/n)$ is well-studied\cite{tlcomp,tlcomp2,rw,DKLP,dklp2,frmix}. For $0<\d<\inv{3}$ and $p=1+n^{-\d}$ a.a.s.\ the largest component consists of a large 2-core $C_2$, with a Poisson Galton-Watson tree of branching parameter $1-n^{-\d}$ rooted at each vertex.  In addition, there are $\Theta(n)$ smaller components distributed essentially like Poisson Galton-Watson trees of branching parameter $1-n^{-\d}$, except that some of them contain a single cycle.  We can define the  stripping number to be the stripping number of the largest component, or to be the maximum of the stripping number of all components. Either way, it is of the order of the height of the tallest of those trees which is $\tdT(n^{\d})$.  Similarly, the maximum depth is of the order of the size of the largest tree, which is  $\tdT(n^{2\d})$.   If $p=1+n^{-\d}$ then we only have the smaller components, but again we have a stripping number of $\tdT(n^{\d})$ and a depth of  $\tdT(n^{2\d})$.

\section{The stripping number for supercritical $c$}\lab{sbsn}

This appendix repeats much of the material from Section~\ref{smt1}, but more formally. It also adds many more details.

{
In this section, we present the proof of Theorem~\ref{mt} for the case $c=c_{r,k}+n^{-\d}$. Note that this is the case that is relevant to our study of the clustering window for random $r$-XORSAT. To extend the proof to the range $c_{r,k}-n^{-\d}\leq c < c_{r,k}+n^{-\d}$, and thus complete the proof of, Theorem~\ref{mt}, we couple our random hypergraph $H=\calh_{\dd}(n,p=c/n^{\dd-1})$ to a random hypergraph $H'=\calh_{\dd}(n,p=(c_{r,k}+n^{-\d})/n^{\dd-1})$
and argue that the stripping number of $H$ is not much lower than the stripping number of $H'$; the details are in Section~\ref{scouple}.
}

We begin with a closer look at the size of the $k$-core when the density is  $c_{\dd,k}+o(1)$.  We define:
\begin{eqnarray*}
f_t(\la)&=&e^{-\la}\sum_{i\geq t}\frac{\la^i}{i!};\\
h(\mu)=h_{\dd,k}(\mu)&=&\frac{\mu}{f_k(\mu)^{\dd-1}}.
\end{eqnarray*}
Note that $f_t(\la)$ is the probability that a Poisson variable with mean $\la$ is at least $t$.
Thus $c_{\dd,k}=\inf_{\mu>0} (\dd-1)! h(\mu)$. Now for any $\dd,k\geq2, (\dd,k)\neq (2,2)$, we define
$\mu_{\dd,k}$ to be the value of $\mu$ that minimizes $h(\mu)$; i.e. the (unique) solution to:
\begin{equation}
c_{\dd,k}= (\dd-1)! h(\mu_{\dd,k}).\lab{murk}
\end{equation}
We define
\begin{eqnarray*}
\a=\a_{\dd,k}&=&f_k(\mu_{\dd,k})\\
\b=\b_{\dd,k}&=&\frac{1}{\dd}\mu_{\dd,k} f_{k-1}(\mu_{\dd,k})
\end{eqnarray*}

Theorem 1.7 of \cite{jhk}, along with Lemma~\ref{l:diff} below, yields the following bound on the size of the $k$-core:

\begin{lemma}\lab{lcoresize}
For  $\dd,k\geq2, (\dd,k)\neq (2,2)$, there exist {two positive constants} $K_1=K_1(r,k),K_2=K_2(r,k)$ such that: Consider any $0<\d<\hf$ and $c=c_{\dd,k}+ n^{-\d}$. Then $\aas$ the $k$-core of $\calh_{\dd}(n,p=c/n^{\dd-1})$ has
\begin{enumerate}
\item[(a)] $\a n +K_1n^{1-\d/2}+O(n^{3/4})$ vertices and
\item[(b)] $\b n +K_2n^{1-\d/2}+O(n^{{3/4}})$ edges.
\end{enumerate}
\end{lemma}

Next, we turn our attention to the parallel stripping process.
At any point during the process, a vertex is said to be {\em light} if it has degree less than $k$, and {\em heavy} otherwise.
So in each iteration, we delete all light vertices.

We restate the parallel stripping process in a slowed-down version, which will be more convenient to analyze.   Rather than removing all vertices of $S_i$ at once, we remove them one at a time.  When removing a vertex, we slow down further by removing one edge at a time.
To facilitate this, we maintain a queue $\msq$ containing all light vertices:

\begin{tabbing}
{\bf SLOW-STRIP}\\
\\
{\bf Input:}  A hypergraph $G$.\\
\\
Initialize: $t:=0,G_0:=G$, $\msq:=S_1$ is the set of all vertices of degree less than $k$ in $G$.\\
Whi\=le $\msq\neq\emptyset$:\\
\>Let $v$ be the next vertex in $\msq$.\\
\>Remove one of the hyperedges $e$ containing $v$.\\
\>If \=any vertex of $e$ now has degree 0\\
\>\>then remove that vertex from $G$ and from $\msq$.\\
\>If any other vertex of $e$ has its degree drop to below $k$\\
\>\>then add that vertex to the end of $\msq$.\\
\>$G_{t+1}$ is the resulting hypergraph; $t:=t+1$.\\
\end{tabbing}

Note that the vertices of $S_i$ are removed before the vertices of $S_{i+1}$.  Note also that when processing a vertex $v\in\msq$, all edges from $v$ are removed (and hence $v$ is removed) before moving to the next vertex of $\msq$.  So this procedure removes vertices in the same order as the parallel stripping process; in particular, if $t$ is the iteration when the first vertex of $S_i$ is selected then $G_{t}$ is the graph remaining after $i-1$ iterations of the parallel stripping process and at time $t$, $\msq=S_{i}$. It will be convenient to define $t(i)$ to be that iteration; i.e.:

\begin{definition}  We use $t(i)$ to denote the iteration of SLOW-STRIP at which iteration $i$ of the parallel stripping process begins; i.e. the iteration at which the first vertex of $S_i$ reaches the front of $\msq$.
\end{definition}

 Our key parameter is:
\[ L_t:= \mbox{ the total degree of the vertices in $\msq$ at iteration $t$ of SLOW-STRIP.}\]

We focus most of our analysis on a period where the size of the remaining graph is very close to $\a n$.  Lemma~\ref{l:B} below says that for any $\e>0$ there is a $B=B(\e)$ such that after step $B$  of the parallel process, or equivalently, step $t(B)$  of SLOW-STRIP, fewer than $\a n + \e n$ vertices remain.  This allows us to prove things such as:

\newtheorem*{llt1}{Lemma~\ref{llt1}}
\begin{llt1}  There are constants $B,K$ such that: If $c=c_{\dd,k}+n^{-\d}$, then \aas for every $t\geq t(B)$,
\[\ex(L_{t+1}{\mid G_t})\leq L_t-Kn^{-\d/2}.\]
\end{llt1}

{\em Remark:}  The key Lemma 32 of~\cite{amxor} implies that, when $c=c_{r,k}+\e$, we have \aas $\ex(L_{t+1}{\mid G_t})\leq (1-\z)L_t$ for a constant $\z>0$.  That difference is what causes the stripping number and depth to rise from $O(\log n)$ above the clustering threshold to $n^{\Theta(1)}$ inside the clustering threshold.  It is also the cause of most of the difficulties in this paper.

Some of our more delicate analysis requires us to be a sublinear distance away from the $k$-core.  Specifically, we often focus on Phase 2 of SLOW-STRIP, defined as:
\begin{definition} \lab{def:t0}
 $\g=3K_1n^{1-\d/2}$, and $t_0$ is the first iteration of SLOW-STRIP in which the number of vertices in the remaining graph is exactly $\a n+\g$.     We refer to iterations $t=0,....,t_0-1$ as {\em Phase 1} and the remaining iterations as {\em Phase 2}.
\end{definition}

Note that we can lower bound the number of iterations in Phase 2:

\begin{lemma}\lab{lp0}  If $c=c_{\dd,k} + n^{-\d}$  then \aas Phase 2 contains at least
$\g/3=K_1n^{1-\d/2}$ iterations.
\end{lemma}

\proofstart We delete at most one vertex per iteration. So, by Lemma~\ref{lcoresize}(a), \aas the $k$-core has fewer than $\a n + 2K_1n^{1-\d/2}$ vertices and so it
takes at least $K_1n^{1-\d/2}$ iterations to reach it.
\proofend

We will show that, during  Phase 2, $L_t$ is small.

\newtheorem*{llt2}{Lemma~\ref{llt2}}
\begin{llt2} If $c=c_{\dd,k}+ n^{-\d}$, then \aas for every $t\geq  t_0$:
\[L_{t}\leq O(n^{1-\d}).\]
\end{llt2}

We present the proofs of Lemmas~\ref{llt1} and~\ref{llt2} in Section~\ref{slt}.  We close this
section by showing how they imply our bounds on the stripping number:

{\em Proof of Theorem~\ref{mt} for $c=c_{\dd,k}+n^{-\d}$}
We begin with the lower bound on the stripping number for part (a).   We run SLOW-STRIP and recall that this can be viewed as also running the parallel stripping process slowly. At the beginning of Phase 2, exactly $\a n+\g=\a n+ 3K_1n^{1-\d/2}$ vertices remain in the graph. Therefore, by Lemma~\ref{lcoresize}(a), \aas we need to remove at least $K_1n^{1-\d/2}$ vertices before reaching the $k$-core.  By Lemma~\ref{llt2}, at most $|L_{t(j)}|=O(n^{1-\d})$ of them belong to $S_j$ for each $j$. Therefore, we require at least $K_1n^{1-\d/2}/O(n^{1-\d})=\Omega(n^{\d/2})$ iterations of the parallel stripping process to remove them all.

Now we turn to the upper bound in part (b). We will focus on the change in $L_{t(i)}$, the sum, over all $v\in S_i$ of the degree of $v$ in the graph remaining at the beginning of iteration $i$ of the parallel stripping process.

Between iterations $t(i)$ and $t(i+1)$ of SLOW-STRIP, all hyperedges from all vertices in $S_i$ must be deleted. One hyperedge is removed in each iteration, and it touches at most $\dd$ members of $S_i$.  It follows that $t(i+1)-t(i)\geq \inv{\dd} L_{t(i)}$.  So  Lemma~\ref{llt1} yields that there are constants $B,K$ such that for all $i> B$,
\[\ex(L_{t(i+1)})\leq L_{t(i)}-(t(i+1)-t(i))Kn^{-\d/2}\leq L_{t(i)}-\inv{\dd} L_{t(i)}Kn^{-\d/2}
=L_{t(i)}(1-\frac{K}{\dd}n^{-\d/2}).\]
It follows that
\[\ex(L_{t(i)})\leq n(1-\frac{K}{\dd}n^{-\d/2})^{i-B-1}.\]
Thus, for $i>B+\frac{2\dd}{K} n^{\d/2}\log n$, $\ex(L_{t(i)})=o(1)$ and so \aas $L_{t(i)}=0$.
I.e., \aas the stripping number is at most $B+\frac{2\dd}{K} n^{\d/2}\log n=O( n^{\d/2}\log n)$.
\proofend

\section{Bounding the maximum depth}\lab{s2} This appendix repeats much of the material from Section~\ref{smd}, but more formally. It also adds many more details.

 We start by noting:
\begin{lemma} \label{ls2} For any vertex $v\in S_i$, the depth of $v$ is at least $i$.
\end{lemma}

\proofstart  We prove by induction that every vertex $v\in\cup_{j\geq i} S_j$ has depth at least $i$.  This is trivial for $i=1$. Suppose it is true for $i$, and consider $v\in\cup_{j\geq i+1} S_j$.
Since $v\notin S_i$, $v$ has at least $k$ neighbours which are either in the $k$-core or in $\cup_{j\geq i} S_j$.  At least one of those neighbours must be removed before $v$ can be removed, and by our induction hypothesis, each such neighbour has depth at least $i$.  So any stripping sequence which removes $v$ must first include a sequence of length at least $i$ which removes a neighbour of $v$. Thus it must have length at least $i+1$; i.e. the depth of $v$ is at least $i+1$.
\proofend

Therefore, the lower bound in Theorem~\ref{mt}(b) provides the lower bound for Theorem~\ref{mt2}.

As we carry out SLOW-STRIP on an input hypergraph $H$, we create a directed graph $\cald=\cald(H)$ as follows.

\begin{definition}  The vertices of $\cald$ are the vertices of $H$. At each iteration of SLOW-STRIP, we choose the next vertex $v\in \msq$ and remove a hyperedge $e$ containing $v$.  For each of the other $\dd-1$ vertices $u\in e$, we add the directed edge $u\rightarrow v$ to $\cald$.  For any vertex $v\in H$, we define $R^+(v)$ to be the set of vertices reachable from $v$ in $\cald$.
\end{definition}

\begin{observation} For any $v$ that is not in the $k$-core of $H$, the vertices of $R^+(v)$, listed in the order that they are processed in SLOW-STRIP, form a $k$-stripping sequence ending with $v$.
\end{observation}

We will upper bound the depth of $v$ by upperbounding $|R^+(v)|$.
The key observation that allows us to do this is that we can expose the edges outside of the $k$-core as follows. Define
\[\imax \mbox{ to be the number of iterations carried out by the parallel stripping process.}\]

For each $1\leq i\leq \imax$, we are interested in the subhypergraph induced by  $S_i$ and (for $i<\imax$) in  the bipartite  subhypergraph induced by $(S_i,S_{i+1})$. It is possible that some hyperedges of the former are contained in hyperedges of the latter; in order to avoid dependency issues, we remove such edges from the former.  To be specific:

\begin{definition}

\begin{enumerate}
\item[(a)] The vertices of $\cals_i$ are $S_i$. For any hyperedge $e$ that includes at least two vertices of $S_i$, and no vertices of $S_{i+1}$, $e\cap S_i$ is a hyperedge of $\cals_i$.
\item[(b)] The vertices of $\cals_{i,i+1}$ are $S_i\cup S_{i+1}$, and for any hyperedge $e$ that includes at least one vertex of $S_i$ and  at least one vertex of $S_{i+1}$, $e\cap (S_i\cup S_{i+1})$ is a hyperedge of $\cals_{i,i+1}$.
\end{enumerate}
\end{definition}

For each $v\in S_i$, we need to keep track {of the number of the hyperedges containing $v$  that intersect $\cals_i,\cals_{i,i+1}$ in various ways.} So for every $a,b$, and any $v\in S_i$ we define $\la_{a,b}(v)$ to be the number of hyperedges $e$ of $\widehat{H}_i$  (the hypergraph remaining at the beginning of round $i$ of the parallel stripping process) that contain $v$ such that: $e$ contains $a$ vertices of $S_i$ and $b$ vertices of $S_{i+1}$; such a hyperedge has {\em Type $(a,b)$}.   Thus $v$ lies in $\la_{a,0}(v)$ $a$-edges in $\cals_i$ for $2\leq a\leq \dd$, and $\la_{a,b}(v)$ $(a,b)$-edges in $\cals_{i,i+1}$ for $1\leq a,b \leq \dd$ .
\medskip


{\bf EXPOSURE:}
\begin{enumerate}
\item Expose $\imax$.
\item Expose the vertices in $S_1,...,S_{\imax}$.
\item For each $0\leq i\leq \imax$ and each vertex $v\in S_i$, expose $\la_{a,b}(v)$, $1\leq a\leq r, 0\leq b\leq r-a$.
\end{enumerate}

Of course, this also exposes the vertices of the $k$-core $\calc_k=\calc_k(H)$.

\vspace{2ex}

{\bf EDGE-SELECTION:}
\begin{enumerate}
\item For each $1\leq i\leq \imax$, we expose $\cals_i$ by selecting a uniformly random hypergraph with vertex set $S_i$ such that each vertex $v$ lies in exactly $\la_{a,0}$ $a$-edges, $a=2,...,\dd$.
\item For each $1\leq i\leq \imax-1$, we expose $\cals_{i,i+1}$ by selecting a uniformly random bipartite hypergraph with bipartition $(S_i,S_{i+1})$ such that:
\begin{enumerate}
\item every $v\in S_i$ lies in exactly $\la_{a,b}$ $(a,b)$-edges, $a=1,...,\dd-1,b=1,...,\dd-a$;
\item every $u\in S_{i+1}$ lies in at least one hyperedge.
\end{enumerate}
\item For each $1\leq i\leq \imax$ and each $v\in S_i$, we choose the remaining vertices in all hyperedges containing $v$ uniformly from all vertices not removed during iterations $1,...,i+1$ of the parallel stripping process.  {I.e., for each $a+b<r$ and each} Type $(a,b)$ hyperedge $e$ containing $v$, we select $\dd-a-b$ uniformly random vertices from $\calc_k\cup_{j=i+2}^{\imax}S_i$ and add them to $e$.
\end{enumerate}

\begin{claim} The hyperedges selected in this manner are selected with the correct distribution.
\end{claim}

\proofstart  EXPOSURE simply exposes some properties.  Now suppose that EXPOSURE has been completed, and consider two sets of hyperedges $E_1,E_2$ chosen by EDGE-SELECTION. Then $|E_1|=|E_2|$, and both sets are equally likely to be chosen by EDGE-SELECTION.  Let $H_1$ be a hypergraph that is consistent with the choices of EXPOSURE and containing the hyperedges $E_1$.  Let $H_2$ be the graph obtained from $H_1$ by replacing $E_1$ with $E_2$.  Then $H_2$ is also consistent with the choices of EXPOSURE.  Furthermore, $H_1$ and $H_2$ have the same number of edges, and hence are equally likely to be chosen as $\calh_{\dd}(n,p=c/n^{\dd-1})$.  Thus $E_1$ is chosen with the correct distribution.
\proofend

This allows us to bound $\ex(|R^+(v)|)$ as follows.  We first run EXPOSURE. We prove that \aas the sets $S_i$ satisfy certain properties.  For each $u\in S_i$: We use $d_{\cals_i}(u)$ to denote the degree of $u$ in $\cals_i$, i.e. the number of hyperedges in $\cals_i$ containing $u$. We use $d^{(a)}(u)$ to denote the number of $a$-edges in $\cals_i$ containing $u$, and so $d_{\cals_i}(u)=\sum_{a=2}^rd^{(a)}(u)$. We use $d^+(u)$ to denote the number of neighbours that $u$ has in $S_{i+1}$; specifically, the number of vertices $v\in S_{i+1}$ that share an edge of $\cals_{i,i+1}$ with $u$.
Note that these values are determined by the values of $\la_{a,b}(u)$.

We use $\ld_i$ to denote the total degree, in  $\hG_i$, of the vertices of $S_i$. Thus $\ld_i=L_{t(i)}$.

\begin{lemma}\lab{lsi} There exist constants $B,Y_1,Y_2,Z_0,Z_1$, {dependent only on $r,k$,} such that \aas for every $ B\leq i< \imax$ with $\ld_i\ge n^{\d}\log^2 n$:
\begin{enumerate}
\item[(a)] if $\ld_i< n^{1-\d}$ then $(1-Y_1n^{-\d/2})\ld_i\leq \ld_{i+1}\leq  (1-Y_2n^{-\d/2})\ld_i$;
\item[(b)] if $\ld_i\geq n^{1-\d}$ then $(1-Y_1\sqrt{\frac{\ld_i}{n}})\ld_i\leq \ld_{i+1}\leq  (1-Y_2\sqrt{\frac{\ld_i}{n}})\ld_i$;
\item[(c)] $\sum_{j\ge i}\ld_i\le {Z_1}\ld_in^{\d/2}$.
\item[(d)] $\ld_i/(k-1)\le|S_i|\le 4\ld_i$;
\item[(e)] for each $2\leq a \leq r$ and $1\leq t\leq k-1$, the number of vertices $u\in S_i$ with $d^{(a)}(u)=t$ is between $Z_0|S_i|^{(a-1)t+1}/n^{(a-1)t}-\log^2 n$ and $Z_1|S_i|^{(a-1)t+1}/n^{(a-1)t}+\log^2 n$.
\item[(f)] $\sum_{u\in S_{i}} d^+(u) < |S_{i+1}|+Z_1\frac{|S_{i+1}|^2}{n}+\log^2n$.
\end{enumerate}
\end{lemma}

The proof is deferred until Section~\ref{srec}.

{\em Remark:}  For (f), note that  2(b) in EDGE-SELECTION implies that the sum is at least $|S_{i+1}|$; so this condition says  the sum is not much more than that.

Next we pick
any $0\leq i\leq \imax$ and any $v\in S_i$.  We set $R_i:=\{v\}$ and for each $j<i$:
\begin{enumerate}
\item[(a)] We set $R'_j$ to be the set of all vertices $v\in S_j$ that are adjacent to $\cup_{\ell=i}^{j+1} R_{\ell}$.
\item[(b)] We set $R_j$ to be the union of the vertex sets of all components of $\cals_j$ that contain vertices of $R'_j$.
\end{enumerate}
We note that
$R^+(v)\subseteq \cup_{\ell=i}^{1} R_{\ell}$. (Each $R_{\ell}$ may contain some vertices that are not in $R^+(v)$ due to the directions of some of the edges in $\cald$.)

\begin{lemma}\lab{lrec} If Lemma~\ref{lsi}(a,b,c,d,e,f) hold then there are constants $B=B(r,k),Z_2=Z_2(r,k)>0$ such that for all $j\geq B$. { If $\ld_i\geq n^{\d}\log^2 n$ then: }

\[\ex(|R_j|{\mid} R_i,...,R_{j-1})\leq |R_{j+1}|+Z_2\frac{|S_j|}{n}\sum_{\ell=i}^{j+1}|R_{\ell}| +\log^2n.\]
\end{lemma}

\proofstart We take $B=B(r,k)$ to be at least as big as required for Lemma~\ref{lsi}. Suppose that we have exposed levels $R_i,...,R_{j+1}$ of $R^+(v)$ and now we are about to expose $R_j$.
{ Note that Lemma~\ref{lsi}(a,b,d) implies:
\[|S_j|/5k<|S_{j+1}|<5k|S_j|.\]}

Consider any vertex $u\in R'_j$.  We generate $\cals_i$ using the Configuration Model (see Section~\ref{scm}) and expose $C_j(u)$, the component of $u$ in $\cals_j$ using a graph search.  In Section~\ref{srec},  we will use Lemma~\ref{lsi}(e), and a standard branching process analysis to prove:

\begin{equation}\label{erec}
\ex(|C_j(u)|)\leq 1+2r\deg_{\cals_j}(u).
\end{equation}

{\em Remark:}  Note that when $k=2$ (the only case that is relevant to the clustering results in this paper), every vertex lies in at most one edge of $\cals_j$ and so trivially $|C_j(u)|\leq r$. Thus~(\ref{erec}) is not required for this case, and the remainder of the proof of this lemma can be simplified.

Thus, to bound $\ex(|R_j|)=\ex(\sum_{u\in R'_j}|C_j(u)|)$, we will focus on
$\ex\left(\sum_{u\in R'_j}1+2r\deg_{\cals_j}(u)\right)$.

Let $X_{a,t}$ denote the number of vertices $u\in S_j$ with $d^{(a)}(u)=t$.

 For all $u\in\cals_j$, we have $d_{\cals_j}(u),d^+(u)\leq k-1$, by the definition of $\cals_j$.
Since $|\ld_i|\geq n^{\d}\log^2 n$, Lemma~\ref{lsi}(e)   implies
\[
\sum_{u\in S_{j}}  d_{\cals_j}(u)=\sum_{u\in S_{j}} \sum_{a=2}^r d^{a}(u)
=\sum_{a=2}^r\sum_{t=1}^{k-1} tX_{a,t}<rkZ_1\frac{|S_{j}|^2}{n}<25rk^3Z_1\frac{|S_{j+1}|^2}{n}.
\]

This and Lemma~\ref{lsi}(f) imply:

\begin{eqnarray*}
\sum_{u\in S_{j}} d^+(u) (2rd_{\cals_j}(u)+1)&\leq &
\sum_{u\in S_{j}} d^+(u) + 2r\sum_{u\in S_{j}} d^+(u) d_{\cals_j}(u)\\
&\leq &|S_{j+1}|+Z_1\frac{|S_{j+1}|^2}{n} +\log^2n + 2r(k-1) \sum_{u\in S_{j}}  d_{\cals_j}(u)\\
&\leq& |S_{j+1}|+50r^2k^4Z_1\frac{|S_{j+1}|^2}{n} +\log^2n.
\end{eqnarray*}

Next we bound the sum of the $\cals_j$-degrees of the vertices that are adjacent to $\cup_{\ell=i}^{j+1} R_{\ell}$. Consider any $u\in \cals_j$. Let $\calr$ be a bipartite hypergraph on $(\cals_j,\cals_{j+1})$ selected during Step 2 of EDGE-SELECTION, and let $\calr'$ be obtained by permuting the vertices of $\cals_{j+1}$ on $\calr$.  Note that $\calr'$ is just as likely to be chosen in Step 2 as $\calr$ is, and so carrying out a uniformly random such permutation does not affect the distribution of the bipartite graph. It follows that each of the vertices of $S_{j+1}$ is equally likely to be one of the $d^+(u)$ neighbours of $u$ in $\cals_{j+1}$. Therefore, using $|R_{j+1}|\leq|S_{j+1}|$,

\bean
\ex\left(\sum_{u\in N(R_{j+1})\cap\cals_j}(2rd_{\cals_j}(u)+1)\right)
&=& \frac{|R_{j+1}|}{|S_{j+1}|}\sum_{u\in S_{j}} d^+(u) (2rd_{\cals_j}(u)+1)\\
&\leq& |R_{j+1}|\left(1+50r^2k^4Z_1\frac{|S_{j+1}|}{n}\right) +\log^2n.
\eean

Consider any $u\in S_j$. The sum over all $a,b$ of the number of Type $(a,b)$ hyperedges containing $u$ is at most $k-1$.  For each such hyperedge, $u$ selects {$r-a-b\leq r-1$} vertices from $\calc_k\cup_{\ell=\imax}^{j+2} R_{\ell}$, during Step 3 of EXPOSURE.
Therefore, noting $|S_j|<n$, we have:
\begin{eqnarray*}
\ex\left(\sum_{u\in N(\cup_{\ell=i}^{j+2} R_{\ell})\cap\cals_j}(2rd_{\cals_j}(u)+1)\right)
&\leq&\left(\sum_{u\in \cals_j}(2rd_{\cals_j}(u)+1)\right)
\times(k-1)(r-1) \frac{|\cup_{\ell=i}^{j+2} R_{\ell}|}{|\calc_k\cup_{\ell=\imax}^{j+2} R_{\ell}|}\\
&<&\left(|S_{j+1}|+50r^2k^4Z_1\frac{|S_{j+1}|^2}{n}\right)\frac{(k-1)(r-1)}{\a n}\sum_{\ell=i}^{j+2}|R_{\ell}|\\
&<&\frac{(50r^3k^5Z_1+1)|S_j|}{\a n}\sum_{\ell=i}^{j+2}|R_{\ell}|,
\end{eqnarray*}
This gives:

\[\ex(\sum_{u\in R'_j}(2r\deg_{\cals_j}(u)+1))
<|R_{j+1}|\left(1+50r^2k^4Z_1\frac{|S_{j+1}|}{n}\right)+\log^2n
+\frac{(50r^3k^5Z_1+1)|S_j|}{\a n}\sum_{\ell=i}^{j+2}|R_{\ell}|.\]

Applying ~(\ref{erec}) yields:
\begin{eqnarray*}
\ex(|R_j|)&\leq& \ex(\sum_{u\in R'_j}|C_j(u)|)\\
&\leq&\ex\left(\sum_{u\in R'_j}1+2r\deg_{\cals_j}(u)\right)\\
&\leq&|R_{j+1}|\left(1+50r^2k^4Z_1\frac{|S_{j+1}|}{n}\right) +\log^2n
+\frac{(50r^3k^5Z_1+1)|S_j|}{\a n}\sum_{\ell=i}^{j+2}|R_{\ell}|\\
&\leq& |R_{j+1}|+Z_2\frac{|S_j|}{n}\sum_{\ell=i}^{j+1}|R_{\ell}| +\log^2n,
\end{eqnarray*}
for some constant $Z_2=Z_2(r,k)$.
\proofend

For $v\in S_i$, we set $R(v)=\cup_{j=i}^B R_j$, where $B=B(r,k)$ is the constant from Lemma~\ref{lsi}.
In Section~\ref{srec} we augment Lemma~\ref{lrec} with a concentration argument and analyze the resultant recursive equation to prove:
\begin{lemma}\lab{lrec2}
There is a constant $X=X(r,k)$ such that for any $1\leq i\leq \imax$ and any $v\in S_i$:
\[\pr(|R(v)|>{n^{X\d}})<\inv{n^2}.\]
\end{lemma}

This yields our bound on the maximum depth as follows:

{\em Proof of Theorem~\ref{mt2}:} Lemma~\ref{ls2} and Theorem~\ref{mt}(a) imply that the maximum depth is at least $\Omega(n^{\d/2})$.
By Lemma~\ref{lrec2}, the expected number of vertices $v$ with $R(v)>n^{X\d}$ is $o(1)$, and so \aas there are no such vertices. Any vertices of $R^+(v)$ that are not contained in $R(v)$, must have been removed during the first $B$ rounds of the parallel stripping process, and must be within distance $B$ of some member of $R(v)$. A simple argument (eg. Propostion 33 and Lemma 34 of \cite{amxor}) bounds the number of vertices within distance $O(1)$ of a set $R(v)$ and yields that \aas for all $v$: $|R^+(v)|\leq \frac{Y}{2}(|R(v)| +\log n)<n^{2X\d}$, for some $Y=Y(B)=Y(r,k)$. This completes the proof since the depth of $v$ is at most $|R^+(v)|$.
\proofend

\section{The configuration model}\lab{scm}
We analyze random hypergraphs with a given degree sequence using the configuration model of  Bollob\'as~\cite{bb}. We are given the degree $d(v)$ of each vertex $v$; note that $\sum d(v)=kE$, where $E$ is the number of hyperedges. We represent each vertex $v$ as a bin containing $d(v)$ vertex-copies.  We take a uniform partition of these $kE$ vertex-copies such that each part has size exactly $\dd$.  We then contract the bins into vertices, and each part of the partition becomes a hyperedge.  Of course, this hypergraph might not be simple, so we require some conditions on our degree sequence.

\begin{definition}\lab{def:niceDeg}
A {degree sequence $\mathcal{S}$} is \emph{nice} if $\sum_v d(v)=\Theta(n)$, $\sum_v d(v)^2=O(n)$
and $d(v)=o(n^{1/24})$ for all $v$.
\end{definition}

The following  propostion is of a standard type; the hypergraph version used here is from~\cite{cc}:

\begin{proposition} If  $\mathcal{S}$ is {a nice degree sequence,} then there exists $\e>0$ such that the probability that
a random $\dd$-uniform hypergraph with degree sequence $\mathcal{S}$ drawn from the configuration model is simple is at least $\e$.
\end{proposition}

This immediately yields:
\begin{corollary}\lab{ccon}
If  $\mathcal{S}$ is a nice degree sequence then: If property $Q$ holds \whp\ for  hypergraphs with degree sequence $\mathcal{S}$ {drawn from} the configuration model, then $Q$ holds \whp\ for uniformly random simple hypergraphs with degree sequence $\mathcal{S}$.
\end{corollary}

The degree sequences considered in this paper are all easily seen to be nice.  So Corollary~\ref{ccon} allows us to prove our theorems by proving that they hold for the configuration model.

\section{Bounds on $L_i$.}\lab{slt}
In this section, we prove Lemmas~\ref{llt1} and~\ref{llt2}.  So we focus on the case
$c=c_{r,k}+n^{-\d}$.

\subsection{Proof of Lemma~\ref{llt2}}\lab{sec:llt2}

We start by looking more closely at the size of the $k$-core, expanding some of the definitions from Section~\ref{sbsn}. For ease of notation, we drop most of the $\dd,k$ subscripts.

Recall:
\begin{eqnarray*}
f_t(\mu)&=&f_t(\mu)=e^{-\mu}\sum_{i\geq t}\frac{\mu^i}{i!},\quad h(\mu)=\frac{\mu}{f_k(\mu)^{\dd-1}},
\end{eqnarray*}
and
\begin{eqnarray*}
\a=\a_{\dd,k}&=&f_k(\mu_{\dd,k}),\quad \b=\b_{\dd,k}=\frac{1}{\dd}\mu_{\dd,k} f_{k-1}(\mu_{\dd,k}).
\end{eqnarray*}

For any $c$, we define $\mu(c)$ to be the largest solution to
\[c=(\dd-1)! h(\mu).\]
And so, note that $\mu_{\dd,k}=\mu(c_{\dd,k})$, as defined in Section~\ref{sbsn}.
Define
\begin{eqnarray*}
\a(c)&=&f_k(\mu(c)),\quad \b(c)=\frac{1}{r}\mu(c) f_{k-1}(\mu(c)).
\end{eqnarray*}

Theorem 1.7 of \cite{jhk} yields the size of the $k$-core,:

\begin{lemma}\lab{lcoresize2}
For $\dd,k\geq2, (\dd,k)\neq (2,2)$,  consider any $0<\d<\hf$ and $c\geq c_{\dd,k}+ n^{-\delta}$. Then $\aas$ the $k$-core of $\calh_{\dd}(n,p=c/n^{\dd-1})$ has
\begin{itemize}
\item $\a(c) n + O(n^{3/4})$ vertices and
\item $\b(c) n +O(n^{3/4})$ edges.
\end{itemize}
\end{lemma}

\begin{lemma}\lab{l:diff} For $\dd,k\geq2, (\dd,k)\neq (2,2)$, there exist positive constants $K_1$, $K_2$ and $K_3$ such that for any $0<\d<\hf$ and $c= c_{\dd,k}+n^{\d}$,
\begin{eqnarray*}
\mu(c)-\mu_{\dd,k}&=&K_1 n^{-\d/2}+O(n^{-\d}),\\
{\a(c)}-\alpha&=&K_2n^{-\d/2}+O(n^{-\d}), \\
{\b(c)}-\beta&=&K_3n^{-\d/2}+O(n^{-\d}).
\end{eqnarray*}
\end{lemma}
\begin{proof} First we bound $y:=\mu(c)-\mu_{r,k}$. Recall that $\mu(c)$ is the larger root of $h(\mu)=c$ and $\mu_{\dd,k}$ is
the unique root of $h(\mu)=c_{\dd,k}$. As $h(x)$ is convex over $x\in (0,+\infty)$ and has derivative $0$ at $x=\mu_{\dd,k}$.
The Taylor expansion at $\mu=\mu_{\dd,k}$ gives
$$
c=h(\mu_{\dd,k})+\frac{h''(\mu_{\dd,k})}{2}y^2+O(y^3)=c_{\dd,k}+\frac{h''(\mu_{\dd,k})}{2}y^2(1+O(y)).
$$
Thus, $y=\sqrt{2/h''(\mu_{\dd,k})}n^{-\d/2}+O(n^{-\d})=K_1n^{-\d/2}+O(n^{-\d})$, by letting $K_1=\sqrt{2/h''(\mu_{\dd,k})}$.
Next, we bound
$$
f_k(\mu(c))-\alpha_{\dd,k}=f_k(\mu(c))-f_k(\mu_{\dd,k})=f_k'(\mu_{\dd,k})y+O(y^2).
$$
Recall that
$$
f_k(x)=e^{-x}\sum_{i\ge k} x^i/i!.
$$
Thus, $f'_k(x)=e^{-x}x^{k-1}/(k-1)!$. Hence, $f'_k(\mu_{\dd,k})>0$. It follows then that there is a constant $K_2>0$, such that
$$
f_k(\mu(c))-\alpha=K_2n^{-\d/2}+O(n^{-\d}).
$$
Similarly, there is a $K_3>0$ such that
\[
\frac{1}{\dd}\mu(c)f_{k-1}(\mu(c))-\beta=K_3n^{-\d/2}+O(n^{-\d}).\qedhere
\]
\end{proof}

Let $H=\calh_{\dd}(n,p=c/n^{\dd-1})$ and run the SLOW-STRIP algorithm on $H$.  Recall from Definition~\ref{def:t0} that
\[\g=3K_1n^{1-\d/2}.\]
and $t_0$ is the first iteration of SLOW-STRIP in which the number of vertices in the remaining graph is exactly $\a n+\g$. The Second Phase of SLOW-STRIP consists of iterations $t_0$ and greater. Lemma~\ref{lp0} enables us to assume that $\t\geq t_0$, which we often do, and this implies that the algorithm enters the second phase.

Recall that for every $t\ge 0$, $G_t$ is the hypergraph remaining at the beginning of iteration $t$ of SLOW-STRIP. The light vertices in $G_t$ are defined to be the vertices with degree less than $k$ and $L_t$ denotes the total degree of the light vertices in $G_t$. We define:
\[\tau \mbox{ is the iteration in which SLOW-STRIP halts.}\]

The following proposition follows immediately from Lemmas~\ref{lcoresize} and~\ref{lp0} and the fact that each light vertex takes at most $k-1$ steps to be removed in the SLOW-STRIP algorithm.

\begin{proposition}\lab{p:tau}
For $c=c_{r,k}+n^{-\d}$, \aas $t_0+\gamma/3\le\tau\le { t_0+k\gamma}$.
\end{proposition}

 It will be convenient to define
\[\wtau=\min\{\tau,t_0+k\gamma\},\]
and so Proposition~\ref{p:tau} implies that \aas $\t=\wtau$.

The following proposition follows from the definition of $t_0$ and $\wtau$.

\begin{proposition}\lab{p:Gt} For $c=c_{r,k}+n^{-\d}$: Assume $\tau>t_0$.
A.a.s.\ for all $t_0\le t\le \wtau$, the number of vertices in $G_t$ is $\alpha n+O(\gamma)$ and the number of hyperedges in $G_t$ is $\beta n+O(\gamma)$.
\end{proposition}

{
\begin{proof} By definition, the number of vertices in $G_{t_0}$ is $\a n+\g$. Each time we delete a vertex, we delete at most $k$ hyperedges, and by Lemma~\ref{lcoresize}, \aas we reach the $k$-core within $O(\g)$ steps and it has $\b n+O(\g)$ hyperedges. So \aas the number of hyperedges in $G_{t_0}$ is $\b n+O(\g)$. We remove one hyperedge and at most one vertex  per step, and by the definition of $\wtau$, we take $O(\g)$ steps.
\end{proof}
}

Since every vertex of $L_t$ is not in the $k$-core, Proposition~\ref{p:Gt} and Lemma~\ref{lcoresize} yield a weaker bound than Lemma~\ref{llt2}:
\begin{equation}\label{elt2w}
\mbox{A.a.s.\ for every $t_0\le t\le \wtau$: }\qquad |L_t|\leq O(\g).
\end{equation}

Assume $\tau>t_0$. For any $t_0\le t\le \tau$,
let $N_t$ and $D_t$ denote the number of heavy vertices and the total degree of heavy vertices in $G_t$ respectively. By Proposition~\ref{p:Gt},
\begin{equation}
N_t=\alpha n+O(\gamma),\ \ D_t=\dd\beta n+O(\gamma). \lab{vertex-edge}
\end{equation}

Recall  the definition of a nice degree sequence from Definition~\ref{def:niceDeg}, and that it is easy to verify that \aas the degree sequence of $G_{t_0}$ is nice.  Given a (nice) degree sequence ${\bf d}$, let $\Gcal_{{\bf d}}$ denote a uniformly random $\dd$-uniform hypergraph with degree sequence ${\bf d}$. Then, conditional on the degree sequence of $G_{t_0}$ being ${\bf d}$, $G_{t_0}$ is distributed as $\Gcal_{{\bf d}}$. As discussed in Section~\ref{scm}, we may generate $\Gcal_{{\bf d}}$ using the configuration model. By Corollary~\ref{ccon}, in order to prove that $G_{t_0}$ satisfies some property $Q$ \aas, it is sufficient to prove that \aas the degree sequence of $G_{t_0}$ is in a set of degree sequences $\ds$ and \aas for every ${\bf d}\in \ds$, the random hypergraph drawn from the configuration model satisfies property $Q$. In what follows, we first characterise $\ds$. We need the following technical lemma for the monotonicity of the function $x f_{k-1}(x)/f_k(x)$.

\begin{lemma}\lab{l:gk}
Let $g_k(x)=xf_{k-1}(x)/f_k(x)$. Then for any $x>0$,
$g'_k(x)>0$.
\end{lemma}
\begin{proof}
Set $f(x)=e^xf_k(x), h(x)=xe^xf_{k-1}(x)$.  So we wish to show that $f'(x)h(x)<f(x)h'(x)$.

\[f(x)=\sum_{i\geq k}\frac{x^i}{i!},\qquad f'(x)=\sum_{i\geq k}\frac{x^{i-1}}{(i-1)!}.\]

\[h(x)=\sum_{i\geq k}\frac{x^i}{(i-1)!},\qquad h'(x)=\sum_{i\geq k}\frac{ix^{i-1}}{(i-1)!}.\]

\[f(x)h'(x)=\sum_{i,j\geq k}\frac{ix^{i+j-1}}{(i-1)!j!}, \qquad
f'(x)h(x)=\sum_{i,j\geq k}\frac{x^{i+j-1}}{(i-1)!(j-1)!}.\]

When $i=j$, the contribution to each sum is the same:  $\frac{x^{2i-1}}{(i-1)!(i-1)!}$. When $i\neq j$, consider the contribution of $(i,j)$ plus the contribution of $(j,i)$.   The sum of these contributions to $f(x)h'(x)$ and $f'(x)h(x)$ is
\[x^{i+j-1}\left(\frac{i}{(i-1)!j!}+\frac{j}{i!(j-1)!}\right), \qquad x^{i+j-1}\frac{2}{(i-1)!(j-1)!}.\]
The contribution to $f(x)h'(x)$ is larger since $\frac{i}{j}+\frac{j}{i}>2$.
\end{proof}

Given positive integers $n$, $m$ and $k\ge 0$ such that $m\ge kn$, define $Multi(n,m,k)$, the {\em truncated multinomial distribution}, to be the probability space consisting of integer vectors ${\bf X}=(X_1,\ldots,X_n)$ with domain ${\mathcal I}_k:=\{{\bf d}=(d_1,\ldots,d_n):\ \sum_{i=1}^n d_i=m,\ d_i\ge k,\ \forall i\in[n]\}$, such that for any ${\bf d}\in {\mathcal I}_k$,
$$
\pr({\bf X}={\bf d})=\frac{m!}{n^m \Psi}\prod_{i\in [n]}\frac{1}{d_i! }=\frac{\prod_{i\in[n]}1/d_i!}{\sum_{{\bf d}\in{\mathcal I}_k}\prod_{i\in[n]}1/d_i!},
$$
where
$$
\Psi=\sum_{{\bf d}\in {\mathcal I}_k}\frac{m!}{n^m}\prod_{i\in [n]}\frac{1}{ d_i!}.
$$

It is well known that the degree distribution of the heavy vertices of $G_t$
is exactly $Multi(N_t,D_t,k)$, conditional on the values of $N_t$ and $D_t$. (See~\cite{CW} for details.) It was proved in~\cite[Lemma 1]{CW} that the truncated multinomial variables can be well approximated by the truncated Poisson random variables. 
The result can be stated as follows.
\begin{proposition}\lab{p:Poisson}
Given integers $k$, $N$ and $D$ with $D> kN$, assume ${\bf X}\sim Multi(N,D,k)$. For any $j\ge k$, let $\rho_j$ denote the proportion of $X$ that equals $j$. Then, for any $\eps>0$, a.a.s.\
\begin{equation}
\rho_j=e^{-\la}\frac{\la^j}{f_k(\la)j!}+O(n^{-1/2+\eps}), \lab{PoissonApprox}
\end{equation}
where $\la$ satisfies {$g_k(\la)=D/N$}.
\end{proposition}

 By Lemma~\ref{l:gk}, {$g_k(x)=x f_{k-1}(x)/f_k(x)$} is an increasing function on $x>0$. It is easy to show that $\lim_{x\to 0} g_k(x)=k$. Hence, for any $D> kN$, there is a unique $\la$ that satisfies $\la f_{k-1}(\la)/f_k(\la)=D/N$ in the above proposition. {It is also easy to check from the definition of $\mu_{\dd,k}$ above~\eqn{murk} that
\begin{equation}\label{egmrk}
g_k(\mu_{\dd,k})=r\beta/\alpha.
\end{equation}
}
Define
\begin{equation}
p_{\dd,k}(j)=e^{-\mu_{\dd,k}}\frac{\mu_{\dd,k}^j}{f_k(\mu_{\dd,k})j!},\lab{pdk}
\end{equation}
for all integers $j\ge k$, and $p_{\dd,k}(j)=0$ for all $0\le j<k$.

For any $t_0\le t\le\wtau$ and any $j\ge 0$, define
\[\rho_t(j) \mbox{ to be the proportion of vertices with degree $j$ in $G_t$.}\]
Then we have the following lemma.

\begin{lemma}\lab{l:err} For $c=c_{r,k}+n^{-\d}$: Assume $\tau>t_0$.
For all $t_0\le t\le \wtau$ and any {fixed} integer $j\ge 0$, a.a.s.\
$ \rho_t(j)-p_{\dd,k}(j)=O(\gamma/n)$.
\end{lemma}

\begin{proof} For $0\leq j<k$: By (\ref{elt2w}), the number of light vertices is $O(\g)$ and so $\rho_t(j)=O(\g/n)=p_{\dd,k}(j)+O(\gamma/n)$, since $p_{\dd,k}(j)=0$.

For $j\geq k$:  First we show that the proportion of {\em heavy} vertices that have degree $j$ is close to $p_{\dd,k}(j)$.
The fact that this also holds for the proportion of {\em all} vertices will follow since only $O(\g)$ vertices are  light.

As stated above, the heavy degrees are distributed as $Multi(N_t,D_t,k)$.  To apply Proposition~\ref{p:Poisson}, we define $\la$ to be the  unique root of $g_k(\la)=D_t/N_t$.  Applying~(\ref{vertex-edge}) and~(\ref{egmrk}), we have:
\[g_k(\la)=\frac{\a n+O(\g)}{r\b n+O(\g)}=\frac{\a}{r\b}+O(\g/n)=g_k(\mu_{r,k})+O(\g/n).\]
It follows from Lemma~\ref{l:gk} that $g'_k(\mu_{r,k})>0$. Therefore $\la=\mu_{r,k}+O(\g/n)$.  Setting
$\r_j(x)=e^{-x}\frac{x^j}{f_k(x)j!}$ from~(\ref{PoissonApprox}),  recalling~(\ref{pdk}), and noting that $\r'_j(\mu_{\dd,k})=O(1)$, we have that the proportion of heavy vertices
of degree $j$ is
\[\r_j(\la)+O(n^{-1/2+\eps})=\r_j(\mu_{r,k})+O(\g/n)=p_{\dd,k}(j)+O(\g/n).\]
By~(\ref{elt2w}), the total number of vertices is $N_t+O(\g)$ and so the proportion of all vertices with degree $j$ is the proportion of heavy vertices with degree $j$ plus $O(\g/n)$.  This yields the lemma.
\end{proof}

\begin{lemma}\lab{l:degreeK}
For every $\dd,k\ge 2$ $(\dd,k)\neq (2,2)$,
$$
\frac{kp_{\dd,k}(k)\cdot\alpha}{\dd\beta}=\frac{1}{(\dd-1)(k-1)}.
$$
\end{lemma}
\begin{proof}
 By the definitions of $\mu_{\dd,k}$ and $h(\mu)$ from~\eqn{murk}, $h'(\mu_{\dd,k})=0$. Since
$$
h'(x)=\frac{f_{k-1}(x)^{\dd-1}-x(\dd-1)f_{k-1}(x)^{\dd-2}f'_{k-1}(x)}{f_{k-1}(x)^{2(\dd-1)}},
$$
and $f'_k(x)=f_{k-1}(x)-f_k(x)$ for all $k\ge 1$, we have
$$
f_{k-1}(\mu_{\dd,k})=\mu_{\dd,k}(\dd-1)(f_{k-2}(\mu_{\dd,k})-f_{k-1}(\mu_{\dd,k})),
$$
i.e.,
$$
\frac{f_{k-2}(\mu_{\dd,k})}{f_{k-1}(\mu_{\dd,k})}=\frac{1+\mu_{\dd,k}(\dd-1)}{\mu_{\dd,k}(\dd-1)}=1+\frac{1}{\mu_{\dd,k}(\dd-1)}.
$$
On the other hand,
$$
f_{k-2}(\mu_{\dd,k})=f_{k-1}(\mu_{\dd,k})+e^{-\mu_{\dd,k}}\frac{\mu_{\dd,k}^{k-2}}{(k-2)!},
$$
it follows immediately that
$$
e^{-\mu_{\dd,k}}\frac{\mu_{\dd,k}^{k-2}}{(k-2)!f_{k-1}(\mu_{\dd,k})}=\frac{1}{\mu_{\dd,k}(\dd-1)}.
$$
By the definition of $p_{\dd,k}(k)$, $\alpha$ and $\beta$, the assertion follows thereby.
\end{proof}

{By Lemma~\ref{l:err}, \aas the degree sequence of $G_{t_0}$ satisfies $\rho_{t_0}(k)=p_{\dd,k}(k)+O(\gamma/n)$. Furthermore, \aas for all $t_0\le t\le \wtau$, the number of vertices in $G_{t}$ is $\Theta(n)$ and the number of vertices with degree $k$ in $G_t$ differs from that in $G_{t_0}$ by at most $\wtau- t_0=O(\gamma)$. Therefore \aas :
\begin{equation}
\rho_{t}(k)=p_{\dd,k}(k)+O(\gamma/n), \mbox{ for } t\geq t_0.\lab{degCond}
 \end{equation}
}

\begin{corollary}\lab{c:degreeK} { For $c=c_{\dd,k}+n^{-\d}$}, \aas at each iteration $t$ for any $t_0\le t\le \wtau$, for each vertex that enters $\Q$, the probability that it had degree $k$ (in $G_{t}$) is $1/(\dd-1)(k-1)+O(\gamma/n)$.
\end{corollary}
\begin{proof}  {We use} the configuration model. Assume $v$ is the first available vertex in $\Q$ at iteration $t$. Then the algorithm removes one vertex-copy in $v$ and another $\dd-1$ vertex-copies uniformly at random (u.a.r.) chosen from all remaining vertex-copies. By Proposition~\ref{p:Gt}, the probability that the each of these vertex-copies lies in a vertex containing exactly $k$ points equals
$$
\frac{k\rho_t(k)\cdot (\alpha n+O(\gamma))}{\dd\beta n+O(\gamma)}.
$$
 Then the assertion follows immediately from~\eqn{degCond} and Lemma~\ref{l:degreeK}.
 \end{proof}

Recall that $L_t$ denotes the total degree of light vertices in $G_t$. {Lemma~\ref{lcoresize} immediately yields a weaker version of Lemma~\ref{llt2}:}
\begin{observation}\lab{o:Lt}
For $c=c_{\dd,k}+n^{-\d}$, \aas for all $t_0\le t\le \wtau$, $L_t=O(n^{1-\d/2})=O(\gamma)$.
\end{observation}

This allows us to prove that a weaker form of Lemma~\ref{llt1} holds for $t\geq t_0$:

\begin{lemma}\label{llt3} If $c=c_{\dd,k}+n^{-\d}$, then \aas for every $t\geq t_0$,

\[\ex(L_{t+1}{\mid G_t})= L_t\pm O(n^{-\d/2}).\]

\end{lemma}

\begin{proof} Let $v$ be the vertex taken from $\Q$ at iteration $t$. Then, the algorithm removes an edge incident with $v$.  Consider the configuration model. Let $u_1,\ldots,u_{r-1}$ denote the other vertex-copies (except the one contained in $v$) contained in the removed edge. The removal of the vertex-copy in $v$ contributes $-1$ to $\Delta L_t:=L_{t+1}-L_t$ always. We consider the contribution to $\Delta L_t$ from the removal of $u_i$. Let $v(u_i)$ denote the vertex that contains $u_i$. There are three cases.

{\em Case 1}: $v(u_i)$ was in $\Q$ in $G_t$ (i.e.\ after iteration $t-1$). In this case, the contribution of the removal of $u_i$ to $\Delta L_t$ is $-1$. Since $u_i$ is chosen by the algorithm u.a.r.\ from all remaining points, the probability of this event is $O(L_t/(D_t+L_t))=O(\gamma/n)$ by~\eqn{vertex-edge} and Observation~\ref{o:Lt}.

{\em Case 2}: $v(u_i)$ enters $\Q$ at iteration $t$. In this case, the degree of $v(u_i)$ is $k$ in $G_t$ and the contribution to $\Delta L_t$ from the removal of $u_i$ is $k-1$. The probability of this event is $1/(d-1)(k-1)+O(\gamma/n)$ by Corollary~\ref{c:degreeK}.

{\em Case 3}: $v(u_i)$ was not in $\Q$ and does not enter $\Q$ at iteration $t$. In this case, the degree of $v(u_i)$ is more than $k$ in $G_t$ and the contribution to $\Delta L_t$ from the removal of $u_i$ is $0$. The probability of this event is $1-1/(k-1)+O(\gamma/n)$.

By the linearity of expectation, summing the contributions of $u_1,\ldots,u_{r-1}$, we have
\[
\ex(L_{t+1}-L_t\mid G_t)=-1+(\dd-1)\left((-1)\cdot O(\gamma/n)+(k-1)\left(\frac{1}{(\dd-1)(k-1)}+O(\gamma/n)\right)\right)=O(\gamma/n).
\]
The lemma follows by noting that $\gamma=\Theta(n^{1-\d/2})$ by Definition~\ref{def:t0}.
\end{proof}

\begin{lemma}\lab{l:azuma} Let $a_n$ and $c_n\ge 0$ be real numbers and $(X_{n,i})_{i\ge 0}$ be random variables with respect to a random process $(G_{n,i})_{i\ge 0}$ such that
$$
\ex(X_{n,i+1}\mid G_{n,i})\le X_{n,i}+a_n,
$$
and $|X_{n,i+1}-X_{n,i}|\le c_n$,
for every $i\ge 0$ and all (sufficiently large) $n$. Then, for any real number $j\ge 0$,
$$
\pr(X_{n,t}-X_{n,0}\ge ta_n+j)\le \exp\left(-\frac{j^2}{2t(c_n+|a_n|)^2}\right).
$$
\end{lemma}
\begin{proof} Let $Y_{n,i}=X_{n,i}-ia_n$. Then,
$$
\ex(Y_{n,i+1}\mid Y_{n,i})=\ex(X_{n,i+1}\mid X_{n,i})-(i+1)a_n\le  X_{n,i}-ia_n = Y_{n,i}.
$$
Thus, $(Y_{n,i})_{0\le i\le t}$ is a supermartingale. Moreover, $|Y_{n,i+1}-Y_{n,i}|\le c_n+|a_n|$. By Hoeffding-Azuma's inequality,
$$
\pr(Y_{n,t}-Y_{n,0}\ge j)\le \exp\left(-\frac{j^2}{2t(c_n+|a_n|)^2}\right).
$$
This completes the proof of the lemma.
\end{proof}

We now recall the statement of Lemma~\ref{llt2}:

\newtheorem*{llt2b}{Lemma~\ref{llt2}}
\begin{llt2b} If $c=c_{\dd,k}+ n^{-\d}$, then \aas for every $t\geq  t_0$:
\[L_{t}\leq O(n^{1-\d}).\]
\end{llt2b}

\begin{proof} By Lemma~\ref{llt3}, there exists a nonnegative sequence $(a_n)_{n\ge 1}$ such that $a_n=O(n^{-\d/2})=O(\gamma/n)$ and
$$
\ex(L_{t+1}\mid G_t)\le L_t+a_n,\ \ \ \ex(-L_{t+1}\mid G_t)\le -L_t+a_n.
$$
By Lemma~\ref{l:azuma} with $j=\gamma^{1/2}\log n$, and by noting that $t-t_0\le \wtau-t_0\le k\gamma$, we have that
$$
\pr(L_t\ge L_{t_0}+a_n(t-t_0)+\gamma^{1/2}\log n)=o(n^{-1}),\quad \pr(-L_t\ge -L_{t_0}+a_n(t-t_0)+\gamma^{1/2}\log n)=o(n^{-1}).
$$

We apply the union bound over all $t_0\le t\le \wtau$, along with the asymptotics $a_n=O(\gamma/n)$, $t-t_0\leq\wtau-t_0<k\g$, and $\gamma^{1/2}\log n =o(\gamma^2/n)$
(since $\g=3K_1n^{1-\d/2}$ and  $\d<1/2$), to obtain that \aas for all $t_0\le t\le \wtau$:
\begin{eqnarray}
L_{t}&\ge& L_{t_0}-O(\g^2 n)\lab{L0a}\\
L_{t}&\le& L_{t_0}+O(\g^2 n) \lab{Lt}
\end{eqnarray}
By Proposition~\ref{p:tau}, a.a.s.\ $L_{\wtau}=0$.  Therefore, (\ref{L0a}) with $t=\wtau$ yields
\begin{equation}\lab{L0b}
L_{t_0}\le O(\gamma^2/n).
\end{equation}
Substituting that into (\ref{Lt}) yields that for all $t_0\le t\le\wtau$,
$L_t=O(\gamma^2/n)=O(n^{1-\d})$, thus establishing the lemma.
\end{proof}

\subsection{Proof of Lemma~\ref{llt1}}\lab{sec:lp2}


We begin this section by proving that, after a constant number of rounds of the parallel stripping process, we can get within any linear distance of the $k$-core.  Recall that $\hG_i$ is the hypergraph remaining after $i-1$ rounds of the parallel stripping process.

\begin{lemma}\lab{l:B}  For $c=c_{\dd,k}+n^{-\d}$, and $H=\calh_{\dd}(n,c/n^{\dd-1})$:
For every $\eps>0$, there exists a constant $B=B(r,k,\e)>0$, such that a.a.s.\ $|\hG_B\setminus \C_k(H)|\le \eps n$.
\end{lemma}

\begin{proof} Fix a small $\eps'>0$ and let $c'=c+\eps'$. Choose $H=\calh_{\dd}(n,c/n^{\dd-1})$ and $H'=\calh_{\dd}(n,c'/n^{\dd-1})$, on the same set of vertices and coupled so that $H\subset H'$.  Run the parallel $k$-stripping process on $H'$, using $\hG'_i$ to denote the subgraph remaining after $i$ iterations. Several other papers (see eg. Proposition 31 of~\cite{amxor}), show that for any $\s>0$ there exists a constant $B>0$ such that $|\hG'_B\setminus \C_k(H')|\le \sigma n$.

Note that every vertex removed during the first $B$ iterations of the parallel stripping process applied to $H'$ would also have been removed  during the first $B$ iterations of the parallel stripping process applied to $H$.  Thus $\hG_B\subseteq \hG'_B$.  Also, $\calc_k(H)\subseteq\calc_k(H')$ and, by Lemma~\ref{lcoresize2}, $|\hG'_B|=|\hG_B|+(\a(c')-\a_{r,k})n +o(n)$. Therefore:
\[|\hG_B\setminus \C_k(H)|\leq |\hG'_B\setminus \C_k(H')|+(\a(c')-\a(c))n +o(n)
<2(\s +\a(c+\e')-\a(c))n<\eps n,\]
for sufficiently small $\e',\s$.
\end{proof}

Now let $H=\calh_{\dd}(n,c/n^{\dd-1})$, where $c=c_{\dd,k}+n^{-\delta}$ and let $B$ be a sufficiently large constant whose value is to be determined later. Recall that $G_t$ is the hypergraph remaining after $t$ iterations of SLOW-STRIP and so  $G_{t(B)}=\hG_B$.

Recalling that $L_t$ is the {total degree} of the light vertices in $G_t$ (i.e. those of degree less than $k$), our goal in this section is to upper bound $\ex (L_{t+1}-L_t\mid G_t)$.

Let $N_t$, $D_t$ and $D_{t,k}$ denote the number of heavy vertices, the total degree of heavy vertices and the total degree of vertices with degree $k$ in $G_t$. Define
\[\zeta_t=D_t/N_t \qquad \mbox{ to be the average degree of the heavy vertices.}\]
  Recalling that, by Lemma~\ref{lcoresize}, the $k$-core \aas has $\a n +o(n)$ vertices and $\b n+o(n)$ edges, we define:

\be
\zeta=\zeta_{\dd,k}=\dd\beta/\alpha .\lab{zeta}
\ee
Therefore the $k$-core \aas has average degree $\z+o(1)$ and so $\z_t$ approaches $\z$.

We are interested in the proportion of the total degree of the heavy vertices coming from vertices of degree $k$; i.e. in:
\[\bar p_t=D_{t,k}/D_t,\]
since this tells the proportion of the time that a heavy vertex in a deleted hyperedge becomes a light vertex.  In the $k$-core, {a.a.s.}\ the proportion of the total degree that comes from vertices of degree $k$ is {approximately} $\psi(\z)$ {by Lemma~\ref{lcoresize} and Proposition~\ref{p:Poisson}}, where the function $\psi(x)=\psi_{k}(x)$ is defined for all $x>k$ as follows:

Let $\la$ be the root of $\la f_{k-1}(\la)=x f_k(\la)$, and set:

\be
\psi(x)=\frac{e^{-\la}\la^{k-1}}{f_{k-1}(\la)(k-1)!}. \lab{h}
\ee

As $\bar p_t$ approaches $\psi(\z)$, we should have $\bar p_t \approx \psi(\z_t)$.
Our next lemma formalizes this approximation:

\begin{lemma}\lab{l:monotone}
A.a.s.\ $\bar p_t= (1+O(n^{-1/2}\log n))\psi(\zeta_t)$.
\end{lemma}
\begin{proof} Let $\la$ be chosen such that $\la f_{k-1}(\la)=\zeta_t f_k(\la)$. Then by Proposition~\ref{p:Poisson},
\bean
\frac{D_{t,k}}{D_t}&=&\frac{ke^{-\la}\la^k N_t}{f_k(\la)(k-1)!D_t}(1+O(n^{-1/2}\log n))=\frac{e^{-\la}\la^k}{f_k(\la)(k-1)!\zeta_t}(1+O(n^{-1/2}\log n))\\
&=&\frac{e^{-\la}\la^{k-1}}{f_{k-1}(\la)(k-1)!}(1+O(n^{-1/2}\log n)).\qedhere
\eean
\end{proof}

We will analyze $\z_t,\bar p_t$ in order to prove Lemma~\ref{llt1}.  We begin with some technical lemmas:

\begin{lemma}\lab{l40} For all $k,\dd\ge 2$ and for any $x\ge \dd(k-1)$,
$$
\frac{e^{-x} x^{k-1}}{f_{k-1}(x)(k-2)!}<\frac{1}{r-1}.
$$
\end{lemma}
\begin{proof} Let
$$
p(x)=\frac{e^{-x} x^{k-1}}{f_{k-1}(x)}.
$$
Then $p(x)$ decreases on $[\dd(k-1),+\infty)$.
Hence, we only need to prove that
\[
\frac{e^{-\dd(k-1)} (\dd(k-1))^{k-1}}{f_{k-1}(\dd(k-1))(k-2)!}<\frac{1}{r-1},
\]
for all $r,k\ge 2$.
Now, let
$$
\phi(x)=\frac{e^{-rx} (rx)^{x}}{f_{x}(rx)(x-1)!},\quad g(x)=\frac{e^{-rx} (rx)^{x}}{(x-1)!} .
$$
Then, we need to prove that
\[
\phi(k-1)<1/(r-1), \forall k,\dd\ge 2.
\]
We have
$f_x(rx)>1/2$ for all $x\ge 1$ and $r\ge 2$ and hence $\phi(x)\le 2g(x)$.
Moreover,
$$
\frac{g(x+1)}{g(x)}=\frac{e^{-r}(rx+r)^{x+1}}{x(rx)^x}\le \frac{e^{-r}(rx+r)}{x}\left(1+\frac{1}{x}\right)^x\le e^{-r+1}r(1+1/x)\le 2e^{-r+1}r,
$$
for all $x\ge 1$ and $r\ge 3$. Hence, it suffices to show that for all $r\ge 2$,
\bean
2g(3)<\frac{1}{r-1},\quad
2g(2)<\frac{1}{r-1},
\eean
both of which can be easily verified.

\end{proof}

\begin{lemma} \lab{l:rho}
Suppose $\dd,k\ge 2$ and $(\dd,k)\neq (2,2)$. Then, $\zeta<\dd(k-1)$.
\end{lemma}
\begin{proof} Recall the definition of the function $p_{\dd,k}(\cdot)$ from~\eqref{pdk}. By Lemma~\ref{l:degreeK},
$$
\frac{kp_{\dd,k}(k)}{\zeta}=\frac{1}{(\dd-1)(k-1)}.
$$
Then,
$\zeta<r(k-1)$ if and only if $kp_{\dd,k}(k)=\zeta/(r-1)(k-1)<1+1/(r-1)=r/(r-1)$.
Since
$$
\zeta=\frac{\dd\beta}{\alpha}=\frac{\mu_{\dd,k}f_{k-1}(\mu_{\dd,k})}{f_{k}(\mu_{\dd,k})},
$$
and
$$
p_{\dd,k}(k)=e^{-\mu_{\dd,k}}\frac{\mu_{\dd,k}^k}{f_k(\mu_{\dd,k})k!},
$$
We have
$$
p_{\dd,k}(k-1)=e^{-\mu_{\dd,k}}\frac{\mu_{\dd,k}^{k-1}}{f_{k-1}(\mu_{\dd,k})(k-1)!}=\frac{1}{(\dd-1)(k-1)}.
$$
Hence,
$$
kp_{\dd,k}(k)=\mu_{\dd,k}p_k(k-1)=\frac{\mu_{\dd,k}}{(\dd-1)(k-1)}.
$$
Hence,
$
kp_{\dd,k}(k)<\dd/(\dd-1)
$
if and only if $\mu_{\dd,k}<\dd(k-1)$. We know that
$\mu_{\dd,k}$ is the root of
$$
e^{-x}\frac{x^{k-1}}{f_{k-1}(x)(k-2)!}=\frac{1}{\dd-1}.
$$
By Lemma~\ref{l40}, the left hand side is strictly less than the right hand side for all $x\ge \dd(k-1)$. Hence, we have
$\mu_{\dd,k}<\dd(k-1)$.
\end{proof}

Recall that $g_k(x)=x f_{k-1}(x)/f_k(x)$ as defined in Lemma~\ref{l:gk} and that $\zeta=\dd \beta/\alpha$. Then, $\mu_{\dd,k}$ is the root of $g_k(x)=\zeta$ by~\eqn{egmrk}. Recall also that
\[
\psi(x)=\psi_{k}(x)=\frac{e^{-\la}\la^{k-1}}{f_{k-1}(\la)(k-1)!},
\]
where $\la$ is the root of $g_k(\la)=x$. Now define
$p^*=\psi(\zeta)$. Then, we have
\[
p^*=\frac{e^{-\mu_{\dd,k}}\mu_{\dd,{k}}^{k-1}}{f_{k-1}(\mu_{\dd,k})(k-1)!}.
\]
By the definition of $\alpha$, $\beta$ below~\eqn{murk} and $p_{\dd,k}(k)$ in~\eqn{pdk}, we have $p^*=kp_{\dd,k}(k)/\zeta$. Then, by Lemma~\ref{l:degreeK}, we have $p^*=1/(\dd-1)(k-1)$. As a summary of the above discussion, we have the following equalities.
\begin{eqnarray}
g_{k}(\mu_{\dd,k})&=&\zeta=\dd\beta/\alpha;\lab{relation1}\\
p^*&=&\psi(\zeta)=\frac{1}{(\dd-1)(k-1)}.\lab{relation2}
\end{eqnarray}

We wish to bound the expected change in $L_t$.  We begin by bounding  $\z_t$ over the next two lemmas.

\begin{lemma}\lab{l:zetaMartingale}
There is a sufficiently large constant $B$ for which: {For any $0<\d<\hf$ and $c=c_{r,k}+n^{-\d}$,} there are constants $\rho_1>0$ and $\rho_2>0$ such that for every $t\ge t(B)$,
$$
-\frac{\rho_1}{n}\le \ex(\zeta_{t+1}-\zeta_t\mid G_t)\le -\frac{\rho_2}{n}.
$$
\end{lemma}
\begin{proof} Let $\eps>0$ be be a small constant whose value is to be determined later. By Lemma~\ref{l:B}, there is a constant $B=B(r,k,\eps)$, such that  a.a.s.\  $|V(G_{t(B)})\setminus \C_k|\le \eps n$.  Consider any $t\ge t(B)$, and let $v$ be the vertex taken from $\Q$ during iteration $T$. Consider the configuration model. SLOW-STRIP removes one vertex-copy from $v$ and another $\dd-1$ vertex-copies $u_1,\ldots,u_{\dd-1}$ uniformly at random chosen from the remaining ones.  We will split the single step into $\dd-1$ substeps $(T_{t,i})_{1\le i\le \dd-1}$, such that $u_i$ is removed in step $T_{t,i}$ for all $1\le i\le \dd-1$ and let $T_{t,0}=t$. We consider the contribution of $u_i$ to $\ex(\zeta_{{t+1}}-\zeta_{t})$.  If $u_{i+1}$ is contained in a light vertex, then its contribution is $0$, if it is contained in a heavy vertex (which occurs with probability at least $1/2$ if $\eps$ is sufficiently small), then its contribution is
{(extending the definition of $D_t,N_t, \bar p_t$ to $D_{T_{t,i}}, N_{T_{t,i}}, \bar p_{T_{t,i}}$ in the obvious manner):}
\be
\frac{D_{T_{t,i}}-1}{N_{T_{t,i}}}(1-\bar p_{T_{t,i}}) + \frac{D_{T_{t,i}}-k}{N_{T_{t,i}}-1} \bar p_{T_{t,i}}-\frac{D_{T_{t,i}}}{N_{T_{t,i}}}.\lab{expect1}
\ee
Ignoring the subscript, the above expression equals
\bea
&&\frac{D}{N}\Big((1-D^{-1})(1-\bar p)+\bar p(1-k/D)(1+1/N+O(n^{-2}))-1\Big)\non\\
&&\quad=\frac{1}{N}\left(-1+\bar p\left(\frac{D}{N}-(k-1)+O(n^{-1})\right)\right).\lab{expect2}
\eea
By Lemma~\ref{lcoresize} and {since $|V(G_{t(B)})\setminus \C_k|\le \eps n$}, we have
a.a.s.\ $D=d\beta n+O(\eps n)$, $N=\alpha n+O(\eps n)$ and so $\zeta_t=\dd\beta /\alpha +O(\eps){=\z +O(\eps)}$. Recall from~\eqn{relation2} that
\[
 p^*=\psi(\zeta)=\frac{1}{(d-1)(k-1).}
 \]
 By Lemma~\ref{l:monotone}, for all $t\ge {t(B)}$ and $0\le i\le d-1$,
\[\bar p_{{T_{t,i}}}=(1+O(n^{-1/2}\log n))\psi(\zeta_t)=\psi(\zeta)+O(\eps)=p^*+O(\eps).
\]
 We {substitute $\bar p=p^*$, $D=r\beta n$, $N=\alpha n$ into (\ref{expect2}), and obtain:}

 \be
 -1+\frac{1}{(\dd-1)(k-1)}\left(\frac{\dd\beta}{\alpha}-(k-1)+O(n^{-1})\right)= -1+\frac{1}{(\dd-1)(k-1)}\left(\zeta-(k-1)+O(n^{-1})\right).\lab{expect3}
 \ee

By Lemma~\ref{l:rho}, there is a $\sigma>0$ such that $\zeta<\dd(k-1)-\sigma$. Hence,
\bea
-1+\frac{1}{(\dd-1)(k-1)}\left(\zeta-(k-1)+O(n^{-1})\right)&<&-1+\frac{1}{(\dd-1)(k-1)}((\dd-1)(k-1)-\sigma/2)\nonumber\\
&<&-\sigma/2(\dd-1)(k-1).\nonumber
\eea
Since in every step, the quantity of variables in~\eqn{expect2} (e.g.\ $\bar p$ and $D/N$) differs from that in~\eqn{expect3} by $O(\eps)$, by choosing $\eps>0$ sufficiently small, there exists {constants $\rho_1',\rho_2'>0$ such that a.a.s.\ }
$$
 {-\frac{\rho_1'}{n}\le }\frac{D_{T_{t,i}}-1}{N_{T_{t,i}}}(1-\bar p_{T_{t,i}}) + \frac{D_{T_{t,i}}-k}{N_{T_{t,i}}-1} \bar p_{T_{t,i}}-\frac{D_{T_{t,i}}}{N_{T_{t,i}}}\le -\frac{\rho_2'}{n},
$$
for all $t\ge 0$ (here we also use that $N=\Theta(n)$).

{So the contribution of each $u_i$ to $\ex(\zeta_{{t+1}}-\zeta_{t})$ lies between  $-\rho_1'/n$ and $-\rho_2'/n$.  Since there are $r-1$ $u_i$'s, the lemma follows with $\r_1=(r-1)\r_1',\r_2=(r-1)\r_2'$.}
\end{proof}

In the next lemma, we obtain a coarse bound on $\zeta_t$. We will refine this bound in the later part of this paper (See Lemma~\ref{l:zetaseq}).

\begin{lemma} \lab{l:zetaT} There exists constant $B$ such that {for $c=c_{r,k}+n^{-\d}$} and for every $\eps>0$,
a.a.s.\ $\zeta_t\ge \zeta_{\tau}-n^{-1/2+\eps}$, for all $t\ge t(B)$.
\end{lemma}
\begin{proof} By Lemma~\ref{l:zetaMartingale}, a.a.s. $\zeta_t$ is a supermartingale. Moreover, $\zeta_{t+1}-\zeta_t$ is at most $O(1/n)$. By Azuma's inequality, we have that for all $t_1>t_2\ge t(B)$ and for any $j>0$,
$$
\pr(\zeta_{t_1}-\zeta_{t_2}\ge j)\le \exp\left(-\Omega\left(\frac{j^2}{(t_1-t_2)n^{-2}}\right)\right).
$$
Since $\tau=O(n)$, for each $t(B)\le t\le \tau$
$$
\pr(\zeta_{\tau}-\zeta_t\ge n^{-1/2+\eps}) \le \exp(-\Omega(n^{2\eps})),
$$
and so the probability that there is $t$ with $\zeta_t<\zeta_{\tau}-n^{-1/2+\eps}$ is at most $O(n\exp(-\Omega(n^{2\eps})))=o(1)$. \end{proof}

And finally, we come to our proof of Lemma~\ref{llt1}, which we restate.

\newtheorem*{llt1b}{Lemma~\ref{llt1}}
\begin{llt1b}  There are constants $B,K$ such that: If $c=c_{\dd,k}+n^{-\d}$, then \aas for every $t\geq t(B)$,
\[\ex(L_{t+1}{\mid G_t})\leq L_t-Kn^{-\d/2}.\]
\end{llt1b}

\begin{proof}
 Consider the configuration model. In each step $t(B)\le t\le\tau$, the algorithm removes a vertex-copy in a light vertex and another $\dd-1$ vertex-copies $u_1,\ldots,u_{\dd-1}$ chosen uniformly from all remaining ones. For each $1\le i\le \dd-1$, let $h_{t,i}$ denote the probability that $u_i$ is contained in a light vertex. Recall that $\bar p_t=D_{t,k}/D_t$, the proportion of contribution to the total degree of heavy vertices from the vertices with degree $k$.  Then
\begin{equation}
\ex(L_{t+1}-L_t\mid G_t)=-1+\sum_{i=1}^{\dd-1} \Big(-h_{t,i}+(1-h_{t,i})(k-1) (\bar p_t+O(n^{-1}))\Big),\lab{barpt}
\end{equation}
where $O(n^{-1})$ accounts for the change of $\bar p_t$ caused by the removal of the first $i-1$ points. The above expression is maximized when $h_{t,i}=0$ for all $1\le i\le \dd-1$. Thus,
\begin{equation}
\ex(L_{t+1}-L_t\mid G_t)\le -1 + (\dd-1)(k-1)\bar p_t+O(n^{-1}).\lab{barpt2}
\end{equation}
Given $\zeta'>k$, let $\lambda(\zeta')$ denote the unique $\lambda$ satisfying $\la f_{k-1}(\la)/f_k(\la)=\zeta'$.
By Lemma~\ref{l:monotone},
$$
\bar p_t=\frac{e^{-\la(\zeta_t)}\la(\zeta_t)^{k-1}}{f_k(\la(\zeta_t))(k-1)!}(1+O(n^{-1/2}\log n)).
$$

By Lemma~\ref{lcoresize}, a.a.s.\
$$
\zeta_{\tau}=\frac{\mu(c)f_{k-1}(\mu(c))}{f_k(\mu(c))}+O(n^{-1/4}).
$$
By Lemma~\ref{l:diff}, there is a constant $K_1>0$ such that $\mu(c)-\mu_{\dd,k}\ge K_1n^{-\d/2}$. Recall that $\zeta=\dd\beta/\alpha=\mu_{\dd,k}f_{k-1}(\mu_{\dd,k})/f_k(\mu_{\dd,k})$. Then by Lemmas~\ref{l:gk} and~\ref{l:zetaT}, there is a constant $K_2$, such that a.a.s.\ for all $t\ge t(B)$, $\zeta_{t}-\zeta\ge\zeta_{\tau}-n^{-1/3}-\zeta \ge K_2n^{-\d/2}$, as $n^{-1/3}=o(n^{-\d/2})$.  By Lemma~\ref{l:degreeK},
$$
\frac{e^{-\la(\zeta)}\la(\zeta)^{k-1}}{f_k(\la(\zeta))(k-1)!}
=\frac{e^{-\mu_{\dd,k}}\mu_{\dd,k}^{k-1}}{f_k(\mu_{\dd,k})(k-1)!}=\frac{1}{(\dd-1)(k-1)}.
$$
Since a.a.s.\ for all $t\ge t(B)$, $\zeta_t-\zeta\ge K_2n^{-\d/2}$, we have a.a.s.\ $\la(\zeta_t)-\la(\zeta)=\la(\zeta_t)-\mu_{\dd,k}\ge K_3 n^{-\d/2}$ for some constant $K_3>0$ by Lemma~\ref{l:gk}. Now by Lemma~\ref{l:monotone},
\begin{equation}
\bar p_t\le (1+O(n^{-1/2}\log n))\frac{e^{-(\mu_{\dd,k}+K_3n^{-\d/2})}(\mu_{\dd,k}+K_3n^{-\d/2})^{k-1}}{f_k(\mu_{\dd,k}+K_3n^{-\d/2})(k-1)!}
<\frac{1}{(\dd-1)(k-1)}-K_4n^{-\d/2}.\lab{br1}
\end{equation}
It follows then that a.a.s.\ there is a $K>0$ such that
\[
\ex(L_{t+1}-L_t\mid G_t)\le -Kn^{-\d/2}.\qedhere
\]
\end{proof}


We close this section by proving a monotone property of $\psi(x)$, which will be useful in Section~\ref{slsi}.

\begin{lemma}\lab{l2:monotone}
$\psi_k(x)$ is a strict decreasing function on $x>k$.
\end{lemma}
\begin{proof} Let $g(\la)=e^{-\la}\la^{k-1}/f_{k-1}(\la)(k-1)!$. We will prove that $g(\la)$ is a decreasing function on $\la>0$.
It then follows that $\psi_k(x)$ is a decreasing function of $x$ by Lemma~\ref{l:gk}.
In order to show that $g'(\la)<0$, it only requires to show that for every $\la>0$,
$$
(k-1)\sum_{j\ge k-1}\frac{\la^j}{j!}-\la\sum_{j\ge k-2}\frac{\la^j}{j!}=\frac{\la^{k-2}}{(k-2)!}\left(-\la+(k-1-\la)\sum_{j\ge k-1}\frac{\la^{j-k+2}}{[j]_{j-k+2}}\right)
$$
is negative. It is trivially true if $\la\ge k-1$. Now assume that $\la<k-1$. Since for every $j\ge k-1$, we have
$$
\frac{\la^{j-k+2}}{[j]_{j-k+2}}\le\left(\frac{\la}{k-1}\right)^{j-k+2},
$$
where the inequality is strict expect for $j=k-1$. Thus,
\[
-\la+(k-1-\la)\sum_{j\ge k-1}\frac{\la^{j-k+2}}{[j]_{j-k+2}}<-\la+(k-1-\la)\sum_{j\ge 0}\left(\frac{\la}{k-1}\right)^{j+1}=-\la+(k-1-\la)\cdot\frac{\la}{k-1}\cdot\frac{1}{1-\frac{\la}{k-1}}=0.
\]
This completes the proof that $g(\la)$ is a decreasing function on $\la>0$.
\end{proof}

\section{Bounding the stripping number below the supercritical case.}\lab{scouple}
In this section, we complete the proof of Theorem~\ref{mt} by showing that the lower bound on
the stripping number holds for all $c_{r,k}-n^{-\d}\leq c<c_{r,k}+n^{-\d}$.

So consider any such $c$ and let  $H=\calh_{\dd}(n,p=c/n^{\dd-1})$.  We also consider a second random hypergraph $H'=\calh_{\dd}(n,p=(c_{r,k}+n^{-\d})/n^{\dd-1})$, and couple the two so that $H\subseteq H'$.  Consider the following stripping procedure to find the $k$-core of $H$:

\begin{enumerate}
\item Run Phase 1 of SLOW-STRIP on $H'$, thus obtaining $G'_{t_0}\subseteq H'$.
\item For each vertex $v$ removed from $H'$ in Step 1, we also remove $v$ from $H$.
We call the remaining hypergraph $G_{t_0}$.
\item Run SLOW-STRIP on $G_{t_0}$.
\end{enumerate}

Note that since $H\subseteq H'$, the $k$-core of $H$ is contained in the $k$-core of $H'$, and so the $k$-core of $H$ is contained in  $G_{t_0}=H'_{t_0}$. In other words, this is a valid way to obtain the $k$-core of $H$.

Note also that this is not equivalent to running SLOW-STRIP on $H$, since doing so could remove a different (larger) set of vertices in Phase 1, and thus yield a different subgraph $G_{t_0}$.  Nevertheless,
we still have $G_{t_0}\subseteq H$ and so the stripping number of $H$ is at least  the stripping number of $G_{t_0}$.

The number of hyperedges in $H'$ but not in $H$ is distributed exactly like the binomial $BIN({n\choose r},(c_{r,k}+n^{-\d}-c)/n^{r-1})$ and so is highly concentrated around its mean which is less than $\frac{n^r}{r!}\times 2n^{1-r-\d}=2n^{1-\d}/r!$. In particular, standard concentration bounds (eg. the Chernoff Bound) imply that \aas :
\begin{equation}\label{ehh'}
\mbox{the number of edges in $H'$ but not $H$ is at most }2 n^{1-\d}.
\end{equation}

We use $G_{i}$ to denote the subgraph of $H$ remaining after $i-t_0$ iterations of SLOW-STRIP on $G_{t_0}$.  We define $L_i, D_i, N_i$ as we did in Sections~\ref{sbsn} and~\ref{sec:llt2}. Similarly we define $L'_{t_0}, D'_{t_0}, N'_{t_0}$ to be the values of the same parameters for $G'_{t_0}$.

 Since $G'_{t_0}$ is the result of carrying out Phase 1 of SLOW-STRIP on $H'=\calh_{\dd}(n,p=(c_{r,k}+n^{-\d})/n^{\dd-1})$, \aas $L'_{t_0}, D'_{t_0}, N'_{t_0}$ satisfy the bounds in (\ref{vertex-edge}) and (\ref{L0b}).

Applying Proposition~\ref{p:tau} to $H'$ tells us that applying SLOW-STRIP to $G'_{t_0}$
would take  at least $\g/3$ iterations.  Each iteration removes one hyperedge, and the removed hyperedge is not in the $k$-core of $G'_{t_0}$ and hence not in the $k$-core of $G_{t_0}$.  By~(\ref{ehh'}), at most $2n^{1-\d}=o(\g)$ of those  hyperedges are not in $G_{t_0}$. Therefore,
SLOW-STRIP takes at least $\g/3-o(\g)>\g/4$ iterations on $G_{t_0}$.  It takes at most $k-1$ iterations to remove a vertex, and so the parallel stripping process, applied to $G_{t_0}$ removes at least $\g/4(k-1)$ vertices.

Since $L_i, D_i, N_i$ represent sums of the degrees of vertices, and since the number of edges in
$G'_{t_0}\bk G_{t_0}$ is at most the number of edges in $H'\bk H$, which by~(\ref{ehh'}) is \aas  less than
$2n^{1-\d}=O(\g^2/n)=o(\g)$, and applying  (\ref{vertex-edge}) and (\ref{L0b}) to $L'_{t_0}, D'_{t_0}, N'_{t_0}$, we have:
\begin{eqnarray*}
N_{t_0}&=&N'_{t_0}+o(\g)=\a n+O(\g)\\
D_{t_0}&=&D'_{t_0}+o(\g)=r\b n+O(\g)\\
L_{t_0}&=&L'_{t_0}+O(\g^2/n)=O(\g^2/n)
\end{eqnarray*}

Thus, the analysis of Section~\ref{sec:llt2} applies to the running of SLOW-STRIP to $G_{t_0}$, and yields the conclusion of Lemma~\ref{llt3}; i.e.  \aas for every $ t_0\leq t\leq t+\g/4$,
\[\ex(L_{t+1}{\mid G_t})= L_t \pm O(n^{-\d/2}).\]
The  same analysis as in the proof of Lemma~\ref{llt2} yields that (\ref{Lt}) holds; i.e. \aas for every $ t_0\leq t\leq t+\g/4$,
\begin{equation}\lab{ecoup1} L_{t}\leq  L_{t_0}+O(\g^2 n)=O(n^{1-\d}).
\end{equation}

The rest of the proof follows like the case $c=c_{r,k}+n^{-\d}$ from Section~\ref{sbsn}. We argued above that
SLOW-STRIP takes at least $\g/3-o(\g)>\g/4$ iterations on $G_{t_0}$.  It takes at most $k-1$ iterations to remove a vertex, and so the parallel stripping process, applied to $G_{t_0}$ removes at least $\g/4(k-1)$ vertices.  By~(\ref{ecoup1}) each iteration $i$ removes  $|L_{t(i)}|=O(n^{1-\d})$ vertices. So there must be at least $\g/O(n^{1-\d})=\Omega(n^{\d/2})$ iterations of the parallel stripping process. I.e. the stripping number of $G_{t_0}$ is at least $\Omega(n^{\d/2})$ and hence so is the stripping number of $H$.
\proofend

\section{Proof of Lemma~\ref{lsi}}\lab{slsi}

Recall that $G_t$ is the hypergraph remaining after $t$ iterations of SLOW-STRIP, and $\hG_i$ is the hypergraph remaining after $i-1$ iterations of the parallel stripping process.  Recall also that $t(i)$ is the iteration of SLOW-STRIP corresponding to the beginning of iteration $i$ of the parallel stripping process.  So $G_{t(i)}=\hG_i$.
Recall that $\tau$ denotes the step when SLOW-STRIP terminates.

Let $K_1>0$ be the constant from Lemma~\ref{lcoresize}, and define
\[\pi(G_t):=|V(G_t)|-\alpha n-\frac{K_1}{2} n^{1-\delta/2}.\]
So $\pi(G_t)$ is approximately the number of non-core vertices in $G_t$ plus $\frac{K_1}{2} n^{1-\delta/2}$.

 Recall the definitions of $L_t$, $D_t$, $N_t$, $D_{t,k}$, $\zeta_t$ and $\bar p_t$ from Section~\ref{sec:lp2} below Lemma~\ref{l:B}. We define the following parameter of $G_t$:
\[\br(G_t)=-1+(\dd-1)(k-1)\bar p_t.\]
By~\eqn{barpt} (since $h_{t,i}=O(L_i/n)$) and~\eqn{barpt2}, we have that a.a.s.\ for every $t(B)\le t\le \tau$,
\begin{equation}
\ex(L_{t+1}-L_t\mid G_t)=\br(G_t)+O(L_t/n),\quad \ex(L_{t+1}-L_t\mid G_t)\le\br(G_t)+O(n^{-1}).\lab{br}
\end{equation}
The second part of (\ref{br}) above is applied when we only require an upper bound on $\ex(L_{t+1}-L_t\mid G_t)$. However, in some cases we need a lower bound as well, and we will use the first part.

By~(\ref{br1}), there is a constant $K>0$ ($K=(r-1)(k-1)K_4$) such that
\begin{equation}
\mbox{a.a.s.\ for every}\ t(B)\le t\le \tau,\ \ \br(G_t)\le -Kn^{-\d/2}. \lab{brupper}
\end{equation}
Then, using Lemma~\ref{lcoresize}, it is easy to check  that
\be
\mbox{a.a.s.}\ \br(G_{\tau}=\C_k(H))=-\Theta(n^{-\d/2}). \lab{brcore}
\ee

In what follows, we will prove some relations between $\br(G_t)$ and $\pi(G_t)$ and between $\br(G_t)$ and $L_t$. We list below a few facts that we will use in our proofs.

Since $c=c_{r,k}+n^{-\d}, \d<\hf$, we can assume that there is a 2-core on a linear number of vertices. At each step of SLOW-STRIP, we remove at most one hyperedge. Thus there is a constant $Q=Q(r,k)$ such that
in every step, the average degree of the heavy vertices is changed by at most $\pm Q/n$; i.e.
\begin{equation}\lab{ezt}
\zeta_{t+1}-\zeta_{t}=O(1/n) \mbox{ uniformly for all } 0\leq t<\tau.
\end{equation}
Therefore,
$|\zeta_t-\zeta_{\tau}|\le \eps/2$ for all $t\ge t(B)$ by choosing sufficiently large $B$ by Lemma~\ref{l:B}. We also know that a.a.s.\ $|\zeta_{\tau}-\zeta|=o(1)$ by Lemma~\ref{lcoresize}. Hence, for all $t\ge t(B)$, $|\zeta_t-\zeta|\le |\zeta_t-\zeta_{\tau}|+|\zeta_{\tau}-\zeta|\le \eps$. This immediately gives (Fa) below.
\begin{description}
\item[(Fa)] For every $\eps>0$, we can choose $B$ sufficiently large such that \aas for all $t\ge t(B)$, $|\zeta_t-\zeta|<\eps$.
\item[(Fb)] For every $t\ge t(B)$, $\bar p_t=(1+O(n^{-1/2}\log n)) \psi(\zeta_t)$, by Lemma~\ref{l:monotone}.
\item[(Fc)] We can choose $\eps>0$ sufficiently small so that there are $c_1,c_2>0$ such that $-c_1<\psi'(x)<-c_2$ for all $x$ such that $|x-\zeta|<\eps$ by Lemma~\ref{l2:monotone} and since $\z>k$.
\end{description}
 By (Fa) and (Fc), we may assume $B$ is chosen so that
\begin{description}
\item[(Fc')] $-c_1<\psi'(\zeta_t)<-c_2$ uniformly for all $t\ge t(B)$.
\end{description}
By~(\ref{ezt}) we have
\begin{description}
\item[(Fd)] $\zeta_i=\zeta_t+O((i-t)/n)$ uniformly for every $0\le t\le i\le \tau$.
\end{description}

In the next lemma, we prove a more precise form of (Fd).

\begin{lemma}\lab{l:zetaseq}
A.a.s.\ for every $t(B)\le t<i\le\tau$,
\begin{enumerate}
\item[(a)] $\zeta_i\le \zeta_t+O(\log n/n)$;
\item[(b)] if $i-t\ge \log n$, then $\zeta_i=\zeta_t-\Theta((i-t)/n)$ uniformly.
\end{enumerate}
\end{lemma}
\begin{proof}
 By Lemma~\ref{l:zetaMartingale}, a.a.s.\ there exist two constants $\rho_1>\rho_2>0$ such that for every $t\ge 0$,
$$
-\frac{\rho_1}{n}\le \ex(\zeta_{t+1}- \zeta_t\mid G_t)\le -\frac{\rho_2}{n}.
$$
Moreover, $|\zeta_{t+1}-\zeta_t|=O(1/n)$.
By Lemma~\ref{l:azuma}, for every $i>t$ and $j\ge 0$,
\bea
\pr(\zeta_{i}\ge \zeta_{t}-(i-t)\rho_2/n+j)&\le& \exp(-\Omega(j^2 n^2/(i-t))), \lab{1} \\
\pr(\zeta_{i}\le \zeta_{t}-(i-t)\rho_1/n-j)&\le& \exp(-\Omega(j^2 n^2/(i-t))). \lab{2}
\eea
Then by the union bound (by taking $j=(i-t)\rho_2/2n$ in~\eqn{1} and taking $j=(i-t)\rho_1/2n$ in~\eqn{2} for each $t$ and each $i\ge t+\log n$), we obtain (b). Part (a) follows by (b) and the fact that for each $i\le t+\log n$, we always have $\zeta_i=\zeta_t+O(\log n/n)$.
\end{proof}



{The next Lemma essentially says that if we can bound the expected change in $L_i$ then we can show $L_i$ is concentrated.  Recall that $\t$ is the stopping time of SLOW-STRIP; i.e.\ the first iteration $t$ for which $L_t=0$.}

\begin{lemma}\lab{lem:Lconcentration} If we have functions ${1>}a=a(n)>b=b(n)\geq \Theta(n^{-\d/2})$ such that:
\begin{enumerate}
\item[(i)] $L_{t}\ge n^{\d}\log^2 n$ and
\item[(ii)] $-a\leq \ex(L_{i+i}-L_i|G_i)\leq -b$ for every $i\ge t$.
\end{enumerate}
 Then, with probability at least $1-o(n^{-1})$,
 \begin{itemize}
 \item[(a)] $L_t-2a(i-t)<L_i<L_t-\hf  b(i-t)$ for all $i\geq t+ n^{\d}\log^{1.5} n$;
 \item[(b)] $L_i<2L_t$ for all $i\geq t$;
 \item[(c)] $t+\inv{2a}L_t<\tau<t+\frac{2}{b}L_t$.
  \end{itemize}
\end{lemma}

\begin{proof}   We start with the upper bound in part (a).
We will apply Lemma~\ref{l:azuma} with $X_{n,\ell}=L_{t+\ell}-L_t$. So we can take $a_n=-b$.  At each step of SLOW-STRIP, we add at most $r-1$ vertices to $\msq$, each of degree at most $k-1$. So $L_i$ increases by at most $(r-1)(k-1)$ and decreases by at most $r$, and so we can take $c_n=(r-1)(k-1)$.

Setting $\ell=i-t\geq n^{\d}\log^{1.5} n$, if $L_i\geq L_t-\hf  b(i-t)$ then $X_{n,\ell}-X_{n,0}=L_i-L_t\geq \ell a_n +\hf\ell b$.  By Lemma~\ref{l:azuma} with $j=\hf\ell b$, the probability of this is at most
\[\exp\left(\frac{(\hf\ell b)^2}{2\ell((r-1)(k-1) + |b|)^2}\right)
\leq\exp(-\Omega(b^2\ell))\leq\exp(-\Omega(\log^{1.5}n))=o(n^{-2}).\]
The lower bound in part (a) is nearly identical, but this time we apply Lemma~\ref{l:azuma} with $X_{n,\ell}=L_t-L_{t+\ell}$,  $a_n=a$ and $j=\ell a$.  Applying the union bound for the at most $n$ choices for $i$ shows that (a) holds with probability at least $1-o(n^{-1})$.

 For part (b): If $i<t+L_t/(r-1)(k-1)$ then the fact that $L_{j+1}<l_j+(r-1)(k-1)$ implies that $L_i<2L_t$.  If $i\geq t+L_t/(r-1)(k-1)>  t+ n^{\d}\log^{1.5} n$, then part (a) implies $L_i<L_t$.

For part (c): Part (a) implies $L_i<0$ for $i\geq t+\inv{2b}L_t>t+n^{\d}\log^{1.5} n$; this yields the upper bound on $\t$.  For the lower bound on $\t$, we apply the same argument used for the lower bound in part (a) with $\ell\leq\inv{2a}L_t$  but with $j=\hf L_t$.  This time we get
\[\pr(L_{t+\ell}<\hf L_t)<\exp(-\Omega(L_t^2/\ell))\leq\exp(-\Omega(aL_t))=o(n^{-2}),\]
thus providing the upper bound on the stopping time.

\end{proof}

In the following lemmas, we will link $\pi(G_t)$ and $L_t$ with $\br(G_t)$ respectively.

\begin{lemma} \lab{lem:pi-br} There are two constants $C_1,C_2>0$ such that
a.a.s.\ $-C_1\pi(G_t)/n\le \br(G_t)\le -C_2\pi(G_t)/n$ for every $t\ge t(B)$.
\end{lemma}
\begin{proof} Recall that $\zeta_t$ denotes the average degree of heavy vertices in $G_t$. Let $t'$ be the maximum integer such that $\pi(G_{t'})\ge K_1n^{1-\d/2}$. Note that $\pi(G_t)$ is a non-increasing function of $t$. Hence, for all $t\le t'$, $\pi(G_t)\ge K_1n^{1-\d/2}$.

By Lemma~\ref{lcoresize}, a.a.s.\ $\pi(G_{\tau})\sim (K_1/2)n^{1-\d/2}$. Since for every $t$, $|\pi(G_t)-\pi(G_{t+1})|\le 1$ as at most one vertex is removed in each step, we have that a.a.s.\ for all $t\le t'$, $\tau-t\ge \pi(G_{t})-\pi(G_{\tau})=\Omega(\pi(G_t))$. In particular, $\tau-t'\ge \pi(G_{t'})-\pi(G_{\tau})\ge (K_1/3) n^{1-\d/2}$. On the other hand, for every $t$, $\tau-t\le k(\pi(G_t)-\pi(G_{\tau}))\le k\pi(G_t)$, since every light vertex in the queue $\Q$ takes less than $k$ steps to be removed. So a.a.s.\ for every $t\le t'$, $\tau-t=\Theta(\pi(G_t))$ uniformly for all $t$ and so
\be
(K_1/3)n^{1-\d/2}\le \tau-t'\le \tau-t=\Theta(\pi(G_t)),\lab{t}
 \ee
uniformly for all $t$.

If $t\leq t'$ then $\pi(G_t)/n\ge K_1n^{-\d/2}$ and so $\t-t\geq\Theta(n^{1-\d/2})>\log n$ by~\eqn{t}.
Thus by Lemma~\ref{l:zetaseq}(b)
we have $\zeta_t=\zeta_{\tau}+\Theta((\tau-t)/n)$. By (Fb) and (Fc'),
a.a.s.\ for each such $t$,
\[\bar p_t-\bar p_{\tau}=\Theta(1)(\psi(\zeta_t)-\psi(\zeta_{\tau}))=-\Theta(1)(\zeta_t-\zeta_{\tau})=-\Theta((\tau-t)/n)=-\Theta(\pi(G_t)/n),
\]
where the constants in the asymptotic notations above are uniform for all $t$. So our lemma holds for all $t\le t'$, by the definition of $\br(G_i)$ and the fact that, by~\eqn{br1},  a.a.s.\ $\br(G_{\tau})=O(n^{-\d/2})=O(\pi(G_t)/n)$.

Now we consider $t>t'$. All constant bounds involved in the asymptotic notations below will be uniform for all $t$. As we proved before, a.a.s.\ $\tau-t'=O(n^{1-\d/2})$. Thus, for all $t'\le t\le\tau$, $\tau-t\le \tau-t'=O(n^{1-\d/2})$. Then by~(\ref{ezt}) we have
\be
\zeta_{\tau}-\zeta_t=O(n^{-\d/2}) \lab{11}
\ee
By (Fb) and (Fc') 
we have {that for all $t(B)\leq t<i\leq\t$:
\be
\br(G_t)-\br(G_i)=\Theta(1)(\bar p_t-\bar p_{i})=\Theta(1)(\psi(\zeta_t)-\psi(\zeta_{i}))=\Theta(1)(\zeta_{i}-\zeta_{t}).\lab{br-zeta}
\ee
}
Then, by Lemma~\ref{l:zetaseq}(a) and~\eqn{11}, there is a constant $C>0$ such that
\[
-Cn^{-\d/2}\le\br(G_t)-\br(G_{\tau})\le C\log n/n.
\]
Since $\br(G_{\tau})=-\Theta(n^{-\d/2})$ a.a.s.\ by~\eqn{brcore}, we have that a.a.s.\ for all $t>t'$, $\br(G_t)=-\Theta(n^{-\d/2})$. The definition of $t'$ and Lemma~\ref{lcoresize} imply that  $\pi(G_t)/n=\Theta(n^{-\d/2})$ for all $t>t'$.  This yields our lemma.
\end{proof}


The following is a key lemma to prove Lemma~\ref{lsi}(a,b).

\begin{lemma}\lab{lem:l-br} There are constants $D_1,D_2,B>0$ such that
a.a.s.\ for all $t\geq t(B)$,
\begin{enumerate}
\item[(a)] for all $t$ such that $L_t\ge n^{1-\d}$, $-D_1\sqrt{L_t/n}\le \br(G_t)\le -D_2\sqrt{L_t/n}$.
\item[(b)] for all $t$ such that $L_t<n^{1-\d}$, $-D_1n^{-\d/2}\le\br(G_t)\le -D_2n^{-\d/2}$.
\end{enumerate}
\end{lemma}

\begin{proof} We take $B$ large enough so that the relevant preceeding results hold.

Let $t_1$ be the smallest $t$ such that  $L_{t+1} <n^{\d}\log^2 n$. We first prove the lemma for all $t\le t_1$.

 We have $\br(G_t)\leq -Kn^{-\d/2}$ by~\eqn{brupper}.  By Lemma~\ref{l:zetaseq}(a), a.a.s.\ for all $t(B)\leq t<i\leq\t$, we have $\zeta_i\le \zeta_t+O(\log n/n)$. Thus by (Fb) and (Fc') and the definition of $\br(G_i)$,  a.a.s.\ for all $t(B)\leq t<i\leq\t$,
\bean
\br(G_i)&=&\br(G_t) -\Theta(1)(\bar p_t-\bar p_i)=\br(G_t)-\Theta(1)(\psi(\z_t)-\psi(\z_i))=\br(G_t)-\Theta(1)(\z_i-\z_t)\\
&\ge& \br(G_t)-\Theta\left(\frac{\log n}{n}\right)\ge 2\br(G_t),
\eean
since $\log n/n=o(\br(G_t))$. By~\eqn{br}, a.a.s.\ for every $t(B)\leq t<i\leq\t$,
\[
 2\br(G_t)-O(L_i/n)\leq \ex(L_{i+1}-L_i\mid G_i)\leq \br(G_i)+O(n^{-1}).
\]
We know that a.a.s.\ for every $t(B)\leq t<i\leq\t$, we have $\br(G_t),\br(G_i)=-\Omega(n^{-\d/2})$  by~\eqn{brupper}. So,  if $t\leq t_1$ then $L_t\geq n^{\d}\log^2 n$ and so:
\be
2\br(G_t)-AL_i/n\le \ex(L_{i+1}-L_i\mid G_i)\le -Kn^{-\d/2},\lab{ebr*}
\ee
for some constants $A,K>0$, uniformly for all  $t(B)\leq t<i\leq t_1$. Applying the union bound to Lemma~\ref{lem:Lconcentration}(b) yields a.a.s.\ for all $t(B)\leq t<i\leq t_1$, $L_i\le 2L_t$. So, the left hand side of~\eqn{ebr*} is at least $2\br(G_t)-2AL_t/n$.

So we can apply Lemma~\ref{lem:Lconcentration}(c) with $a=-2\br(G_t)+2AL_t/n$, to show that with probability $1-o(n^{-1})$, our stopping time $\t\geq t+ 4A L_t/(|\br(G_t)|+L_t/n)$. The union bound shows that this holds for all $t(B)\leq t\leq t_1$.

We delete a vertex at least once every $k-1$ steps of SLOW-STRIP, and so  $|V(G_t)\setminus V(G_{\tau})|\geq (\t-t)/(k-1)$. Therefore,  applying Lemma~\ref{lcoresize}, we have that \aas for all $t(B)\leq t\leq t_1$:
\bea
\pi(G_t)&=&|V(G_t)|-\alpha n-\frac{K_1}{2} n^{1-\delta/2}
\geq|V(G_t)| -|V(G_{\tau})|+\frac{K_1}{3}n^{1-\d/2}\nonumber\\
&\geq&\frac{4A L_t}{(k-1)(|\br(G_t)|+L_t/n)}+\frac{K_1}{3}n^{1-\d/2}. \lab{epgt}
\eea

By Lemma~\ref{lem:pi-br}, there exists a constant $C_2>0$ such that a.a.s.\ for all $t(B)\leq t\leq t_1$,
\[
\br(G_t)\le -C_2\frac{\pi(G_t)}{n}\le-C_2 \left(\frac{4A L_t/n}{(k-1)(|\br(G_t)|+L_t/n)}+\frac{K_1}{3}n^{-\d/2}\right),
\]
and so
\be
|\br(G_t)|\ge C_2 \max\left(\frac{4A L_t/n}{(k-1)(\br(G_t)+L_t/n)},\frac{K_1}{3}n^{-\d/2}\right).\lab{sigma1}
\ee
Taking $B$ large enough that $L_t/n$ is sufficiently small for $t\geq t(B)$ (by Lemma~\ref{l:B}), there is a constant $D_1>0$, such that a.a.s. for all $t(B)\leq t\leq t_1$:
\be\lab{egoal1}
|\br(G_t)|\geq D_1\max\left( \sqrt{\frac{L_t}{n}},n^{-\d/2}\right).
\ee
This yields the upper bounds in our lemma for $t\leq t_1$.
Next we prove the lower bounds; i.e., we wish to prove that for some constant $D_2>0$ for all $t(B)\leq t<i\leq t_1$:
\be\lab{egoal}
|\br(G_t)|\leq D_2\max\left( \sqrt{\frac{L_t}{n}},n^{-\d/2}\right).
\ee

Let $A_2,A_3$ be the implicit constants in (Fd),~\eqn{br-zeta}, and set $A_1=1/(2A_2A_3)$.
Applying~\eqn{br-zeta},~(Fd) we get that for any $t(B)\leq t\leq i\le t+A_1|\br(G_t)|n$:
\[\br(G_i)-\br(G_t)\leq A_3(\z_{t}-\z_i)\leq A_2A_3((i-t)/n)\leq A_1A_2A_3|\br(G_t)|=-\hf\br(G_t),\]
and so $\br(G_i)\le \hf\br(G_t)$. So, by~\eqn{br}:
\be\lab{eggg}
\ex(L_{i+1}\mid G_i)\le L_i+\frac{1}{2}\br(G_t)+O(1/n),\quad \forall t(B)\leq t< i\le t+A_1|\br(G_t)|n.
\ee
We can assume $|\br(G_t)|> \sqrt{(8/A_1)L_t/n}$, as otherwise~(\ref{egoal}) holds with $D_2=\sqrt{8/A_1}$.
Therefore we have (i) $8L_t/|\br(G_t)|<A_1|\br(G_t)|n$ and (ii) the RHS of~(\ref{eggg}) is at most $\frac{1}{4}\br(G_t)$.
We wish we could apply Lemma~\ref{lem:Lconcentration}(c),  with $b=\frac{1}{4}\br(G_t)$ (and with $a=-2\br(G_t)+2AL_t/n$, as argued above), to show that with probability $1-o(n^{-1})$,  the stopping time
\be\lab{esay*}
\tau<t+\frac{2}{b}L_t=t+8L_t/|\br(G_t)|<t+A_1|\br(G_t)|n.
\ee
But we cannot because~\eqn{eggg} does not hold for all $i>t$.  Nevertheless, it is straightforward to adapt the proof of Lemma~\ref{lem:Lconcentration} to obtain~\eqn{esay*};
the main point is that~\eqn{eggg} holds for all $i$ up to the RHS of~\eqn{esay*}.
Taking the union bound shows that \aas this holds for all $t(B)\leq t\leq t_1$ for which $|\br(G_t)|> \sqrt{(8/A_1)L_t/n}$.

Since we remove at most one vertex during each iteration of SLOW-STRIP, it follows that $G_t$ contains at most $8L_t/|\br(G_t)|$ non-core vertices.
Recalling that $\pi(G_t)$ is approximately the number of non-core vertices in $G_t$ plus $\hf K_1 n^{-\d/2}$, this implies that
$\pi(G_t)\le A_4(L_t/|\br(G_t)|+ n^{1-\d/2})$ for some constant $A_4>\max\{8,\hf K_1\}$.
Now, by Lemma~\ref{lem:pi-br}, we have that there is a constant $C_1>0$ such that a.a.s.\ for all $t(B)\leq t\leq t_1$ with $|\br(G_t)|> \sqrt{(8/A_1)L_t/n}$ we have:
\[\br(G_t)\ge -C_1\pi(G_t)/n \ge -C_1A_4\left(\frac{L_t/n}{|\br(G_t)|}+ n^{-\d/2}\right).\]
This implies~(\ref{egoal}) for $t(B)\leq t\leq t_1$ with $D_2=D'_2:=\max\{\sqrt{8/A_1},C_1A_4\}$.

Now we consider $t>t_1$. Since $L_{t_1}\ge n^{\d}\log^2 n$, by~\eqn{esay*}, a.a.s.\ $\tau-t_1<8L_{t_1}/|\br(G_{t_1})|$. By the definition of $t_1$ and Lemma~\ref{lsi}(a), we also have a.a.s.\ $L_{t_1}<2n^{\d}\log^2 n$. (\ref{egoal1}) says that a.a.s.\ $|\br(G_{t_1})|\ge D_1 n^{-\d/2}$. So we obtain: a.a.s.\ for all $t>t_1$,
\[
\tau-t<\tau-t_1<(16/D_1)n^{3\d/2}\log^2n.
\]
This implies that a.a.s.\ for all $t_1<t\le \tau$, $t-t_1<(16/D_1)n^{3\d/2}\log^2n$. Then, by (Fd) and~\eqn{br-zeta}, a.a.s.\ for all $t_1<t\le \tau$,
\[|\br(G_t)-\br(G_{t_1})|\leq A_2A_3(t-t_1)/n\le A_2A_3(16/D_1)n^{3\d/2-1}\log^2n=o(n^{-\d/2})=o(\br(G_{t_1})),
 \]
 as $\d<1/2$. Since~\eqn{egoal1},~\eqn{egoal} hold for $t=t_1$ and for some $D_1=D_1',D_2=D_2'$, they hold for $t>t_1$, and hence for all $t(B)\leq t\leq\t$   by taking $D_1=D'_1/2$ and $D_2=2D'_2$.
\end{proof}

And now we can prove Lemma~\ref{lsi}, which we restate.  Recall that $\ld_i=L_{t(i)}$, and that $\imax$ is the number of iterations carried out by the parallel stripping process.

\newtheorem*{lsi}{Lemma~\ref{lsi}}
\begin{lsi} There exist constants $B,Y_1,Y_2,Z_0,Z_1$ dependent only on $r,k$, such that \aas for every $ B\leq i< \imax$ with $\ld_i\ge n^{\d}\log^2 n$:
\begin{enumerate}
\item[(a)] if $\ld_i< n^{1-\d}$ then $(1-Y_1n^{-\d/2})\ld_i\leq \ld_{i+1}\leq  (1-Y_2n^{-\d/2})\ld_i$;
\item[(b)] if $\ld_i\geq n^{1-\d}$ then $(1-Y_1\sqrt{\frac{\ld_i}{n}})\ld_i\leq \ld_{i+1}\leq  (1-Y_2\sqrt{\frac{\ld_i}{n}})\ld_i$;
\item[(c)] $\sum_{j\ge i}\ld_i\le Z_1\ld_in^{\d/2}$.
\item[(d)] $\ld_i/(k-1)\le|S_i|\le 4\ld_i$;
\item[(e)] for each $2\leq a \leq r$ and $1\leq t\leq k-1$, the number of vertices $u\in S_i$ with $d^{(a)}(u)=t$ is between $Z_0|S_i|^{(a-1)t+1}/n^{(a-1)t}-\log^2 n$ and $Z_1|S_i|^{(a-1)t+1}/n^{(a-1)t}+\log^2 n$.
\item[(f)] $\sum_{u\in S_{i}} d^+(u) < |S_{i+1}|+Z_1\frac{|S_{i+1}|^2}{n}+\log^2n$.
\end{enumerate}
\end{lsi}

{\bf Proof of Lemma~\ref{lsi}:} We take $B$ large enough so that the relevant preceeding results hold. We first prove (a), and so we have  $ n^{\d}\log^2 n\leq \ld_i< n^{1-\d}$. We will apply Lemma~\ref{lem:Lconcentration}(a) with $t:=t(i)$.

Since SLOW-STRIP  removes at least one and at most $r$ vertex-copies from $S_i$ in every iteration $t(i)\le j<t(i+1)$, we have $t(i+1)-t(i)=\Theta(\ld_i)$, uniformly over $i$. So
$t(i+1)>t(i)+n^{\d}\log^{1.5}n$.  Furthermore, by (Fd),
for all $t(i)\le j\le t(i+1)-1$, $\zeta_j-\zeta_{t(i)}=O((j-t(i))/n)=O(\ld_i/n)$, and so by~\eqn{br-zeta}, we have $\br(G_j)=\br(G_{t(i)})+O(\ld_i/n)=\br(G_{t(i)})+O(n^{-\d})$.  Lemma~\ref{lem:l-br}(b) says:
\[-D_1n^{-\d/2}\le\br(G_t)\le -D_2n^{-\d/2}.\]
Therefore, for all $t(i)\le j\le t(i+1)-1$, applying~(\ref{br}) we have
\be\lab{eabab}
-2D_1n^{-\d/2}\le\ex(L_{j+1}-L_j\mid G_j)\le -\hf D_2n^{-\d/2}.
\ee
{The same argument applies for every $i'>i$ and so~\eqn{eabab} holds for every $j\geq t(i)$.}
This allows us to apply  Lemma~\ref{lem:Lconcentration}(a) to prove that with probability $1-o(n^{-1})$,
\[L_{t(j+1)}=L_{t(j)}-\Theta(n^{-\d/2})(t(j+1)-t(j))=L_{t(j)}(1-\Theta(n^{-\d/2})).\]
Taking the union bound over all $i$ yields part (a).

The proof of part (b) is nearly identical, applying Lemma~\ref{lem:l-br}(a)  rather than Lemma~\ref{lem:l-br}(b).

Now we prove part (c). Recalling that $\t$ is the  stopping time of  SLOW-STRIP, we observe that $\sum_{j\ge i}\ld_i\le \dd (\tau-t(i))$, since the total degree in $\msq$ decreases by at most $\dd$ in every iteration of SLOW-STRIP.  As argued in the proofs of parts (a,b), we can apply Lemma~\ref{lem:Lconcentration} - if $\ld_i<n^{1-\d}$ we use $b=\Theta(n^{-\d/2})$;  if $\ld_i\geq n^{1-\d}$ we use $b=\Theta(\sqrt{\ld_i/n})\geq \Theta(n^{-\d/2})$.  Lemma~\ref{lem:Lconcentration}(c) yields that with probability $1-o(n^{-1})$,
$\t-t(i)\leq \frac{2}{b}\ld_i\leq \Theta(n^{\d/2})\ld_i$, and so
\[\sum_{j\ge i}\ld_i\le \dd (\tau-t(i))  \leq \Theta(n^{\d/2})\ld_i.\]
Taking the union bound over all $i$ yields part (a).

Next, we prove parts (d,e,f). We consider running SLOW-STRIP from step $t(i)$ to $t(i+1)-1$; i.e.\ iteration $i$ of the parallel stripping process.  We expose the degree sequence of $G_{t(i)}$ (i.e.\ $\hG_i$) and then use the configuration model.  The maximum degree is at most the maximum degree in $\H_r(n,p=c/n^{r-1})$ which is \aas less than $\log n$.

For part (f):  For any $v\in S_{i+1}$, let $d^-(v)$ denote the number of neighbours $v$ has in $S_{i}$. Then, $\sum_{u\in S_i}d^+(u)=\sum_{v\in S_{i+1}}d^-(v)$. We will use the configuration model in our analysis. In each step of SLOW-STRIP, a vertex-copy is removed from $S_i$, together with another $\dd-1$ vertex-copies chosen uniformly at random. A vertex $u$ in $V(G_{t(i)})\setminus S_i$ is called {\em active} if $u\in S_{i+1}$ and at least two vertex-copies of $u$ were removed from step $t(i)$ to $t(i+1)-1$. We estimate the total number $X$ of vertex-copies removed from active vertices. This upper bounds $\sum_{v\in S_{i+1}}d^-(v)-|S_{i+1}|=\sum_{u\in S_i}d^+(u)-|S_{i+1}|$. Let $X_j$, $j\ge 2$, denote the number of active vertices from which exactly $j$ vertex-copies are removed. Then $X=\sum_{2\le j\le \log n}jX_j$ (since the maximum degree is at most $\log n$). A vertex $u$ can be counted in $X_j$ only if its degree is between $j$ and $k+j-1$ (so that it may be in $S_{i+1}$). A vertex-copy randomly chosen by SLOW-STRIP lies inside $u$ with probability at most $O((k+j-1)/n)=O(j/n)$. There are at most $k\dd|S_i|$ vertex-copies in total that were randomly removed by SLOW-STRIP from step $t(i)$ to $t(i+1)-1$. So, the probability that a vertex $u$ (whose degree is between $j$ and $k+j-1$) is counted in $X_j$ is at most $\binom{k\dd|S_i|}{j}\times O((j/n)^j)=(O(k\dd|S_i|/n))^j$. Hence, for every $2\le j\le \log n$, $X_j$ is stochastically dominated by $\Bin(n,(O(k\dd|S_i|/n))^j)$. (There is some dependency between the events that two vertices $u,u'$ count in $X_j$, but it goes in the right direction for us; i.e.\ conditioning that $u$ counts in $X_j$ decreases the probability that $u'$ counts in $X_j$.)  Applying the standard Chernoff bound (see eg.\ the statement provided in~\cite{mrbook}), it is easy to show that there is a constant $Z>0$ such that a.a.s.\ for every $2\le j\le \log n$, $X_j\le (Zk\dd|S_i|)^j/n^{j-1}+\log n\le (1/2)^{j-2}(Zk\dd|S_i|)^2/n+\log n$, by choosing $B>0$ sufficiently large so that for all $i\ge B$, $Zk\dd|S_i|/n<1/2$. Now, we have a.a.s.\
\bean
\sum_{u\in S_i}d^+(u)-|S_{i+1}|&\le &X=\sum_{2\le j\le \log n}jX_j\le \log^2n+\sum_{2\le j\le \log n}j(1/2)^{j-2}(Zk\dd|S_i|)^2/n\\
&<&8(Zk\dd|S_i|)^2/n+\log^2 n.
\eean
Part (f) follows by choosing $Z_1=8Z^2$.

For part (d). The lower bound is trivial since every vertex in $S_i$ has degree at most $k-1$ in $G_{t(i)}$. The upper bound follows from (f): Every vertex in $S_{i}$ has degree at least $k$ in $G_{t(i-1)}$. So if $\ld_i<\inv{4}|S_i|$, i.e. if the total degree in $G_{t(i-1)}$ of the vertices of $S_i$ is less than $\inv{4}|S_i|$, then the number of edges from $S_i$ to $S_{i-1}$ must be at least $(k-\inv{4})|S_i|\geq \frac{7}{4}|S_i|>|S_i|+Z_1\frac{|S_i|^2}{n}$, if we take $B$ sufficiently large that $|S_i|< \frac{3}{4Z_1}n$ for $i\geq B$.

Finally, we prove part (e): fix $2\le a\le \dd$. When  SLOW-STRIP removes a hyperedge incident with $u\in S_i$, it removes a vertex-copy from $u$ and another $\dd-1$ vertex-copies chosen uniformly at random. This edge contains exactly $a$ vertices in $S_i$ iff exactly $a-1$ of these $\dd-1$ vertex-copies are contained in vertices in $S_i$. This occurs with probability at most $O(|S_i|^{a-1}/n^{a-1})$. Thus, the probability that $u$ is contained in exactly $t$ of those edges is $O(|S_i|^{t(a-1)}/n^{t(a-1)})$, {regardless of the choices of the hyperedges containing the vertices of $S_i$ that preceded $u$ in $\msq$}. Then, the total number of such vertices inside $S_i$ is dominated by $\Bin(k|S_i|,c_1|S_i|^{t(a-1)}/n^{t(a-1)})$ for some constant $c_1>0$. The upper bound in this claim follows by choosing $Z_1$ sufficiently large and by applying the Chernoff bound to $\pr(\Bin(k|S_i|,c_1|S_i|^{t(a-1)}/n^{t(a-1)})>Z_1|S_i|^{t(a-1)+1}/n^{t(a-1)}+\log^2 n)$ and taking the union bound over $i$. Now we prove the lower bound. Assume $Z_0>0$ is sufficiently small. Let $t'$ denote the step when {$\inv{8(r-1)(k-1)}|S_i|$} of the vertices of $S_i$ have been removed from the queue $\Q$. { By part (d), the total number of vertex-copies in $S_i$ at time $t(i)$ is $\ld_i>\inv{4}|S_i|$.  At most $(r-1)(k-1)\times\inv{8(r-1)(k-1)}|S_i|\leq\inv{8}|S_i|$ of those copies are used during steps $t(i)$ to $t'$. So when processing any $u$ from $\Q$ between these steps, at least $\inv{8}|S_i|$ of those vertex-copies remain and so the probability that $u$ is contained in exactly $t$ edges in $S_i$ with size $a$ is $\Omega(|S_i|^{t(a-1)}/n^{t(a-1)})$.} Hence, the number of such vertices in $S_i$ dominates $\Bin(|S_i|/4,c_2|S_i|^{t(a-1)}/n^{t(a-1)})$ for some constant $c_2>0$. Then the lower bound follows by applying the Chernoff bound and taking the union bound.


\section{Bounding $|R_i|$.}\lab{srec}

In this section, we will prove Lemma~\ref{lrec2}, which bounds $|R(v)|$.  We can assume conditions (a,b,c,d,e,f) of Lemma~\ref{lsi}.
We choose any $B\leq i\leq \imax$ and $v\in S_i$.  We recall the definition of $R_j$ from Section~\ref{s2} and that $R^+(v)\cap S_j\subseteq R_j$.
We have $R(v)=\cup_{j=i}^B R_j$.

We begin by completing the proof of Lemma~\ref{lrec} by proving (\ref{erec}), which we first restate.

Recall that $C_j(v)$ is the component of $\cals_j$ containing $v$. (\ref{erec}) states that if $\ld_j\geq n^{\d}\log^2 n$ then for any $v\in\cals_j$:
\[
\ex(|C_j(v)|)\leq 1+2rd_{\cals_j}(v).\]

{\bf Proof of (\ref{erec}):}
We follow the arguments from \cite{mr1} for analyzing the sizes of components in a random graph on a given degree sequence.  We expose $C_j(v)$ using a graph search. We begin by initializing $W=\{v\}$; $W$ is a set of vertices known to be in $C_j(v)$.  We initialize $U$ to be the  $d_{\cals_j}(v)$ hyperedges of $\cals_j$ containing $v$; $U$ is a set of  {\em unexposed hyperedges} - these are hyperedges, containing vertices of $W$, in which we have not yet exposed the other vertices. $Y_s$ counts the size of $U$ after $s$ iterations, so $Y_0=d_{\cals_j}(v)$.

At each iteration, we choose a hyperedge $e\in U$ and expose the other vertices of $e$; i.e. the vertices that are not in $W$.
For  each $w\in e$, if $w\notin W$: (a) we add the other $d_{\cals_j}(w) - 1$ hyperedges of $\cals_j$ containing $w$ to $U$, and (b) we add $w$  to $W$.
We halt when $U=\emptyset$.  When we halt, $W=C_j(v)$.

We wish to bound $\ex(Y_{s+1}-Y_s)$. The edge $e$ selected during iteration $t$ has size $a\leq r$.  We expose $\cals_j$ using the configuration model.  So for any vertex $w\notin W$, the probability that $w\in e$ is at most $(a-1)d^{(a)}(w)/\sum_{x\in S_j\bk W} d^{(a)}(x)$ (it is slightly less than this since some members of $W$ could be in $e$).  Summing over all $w\in S_j\bk W$, and accounting for the fact that $e$ is removed from $U$ yields
\begin{equation}\label{exx}
\ex(Y_{t+1}-Y_t)\leq -1 + \sum_{w\in S_j\bk W}(d_{\cals_j}(w)-1)\frac{(a-1)d^{(a)}(w)}{\sum_{x\in S_j\bk W} d^{(a)}(x)}.
\end{equation}

 As in the proof of Lemma~\ref{lrec}, we use
$X_{a,t}$ to denote the number of vertices $u\in S_j$ with $d^{(a)}(u)=t$.
We set $\la_{a,t}=|S_j|^{(a-1)t+1}/n^{(a-1)t}$.  We assume that condition (e) of Lemma~\ref{lsi} holds and so:
\[Z_0\la_{a,t}-\log^2 n\leq X_{a,t}\leq Z_1\la_{a,t}+\log^2 n.\]

{\em Case 1: $\la_{a,1}< \log^3n$.}  Therefore, $|S_j|\leq n^{\frac{a-1}{a}+o(1)}$. As computed in the proof of Lemma~\ref{lsi}(c),
$\ex(X_{a,t})=O(\la_{a,t})$.  So for all $t\geq 2$,
\[\ex(X_{a,t})=O\left(\left(\frac{|S_i|}{n}\right)^{(a-1)(t-1)}\la_{a,1}\right)
=O(n^{-\frac{1}{a+1}+o(1)}),\qquad\mbox{ since } a\geq 2.\]
Similar simple calculations show that the expected number of vertices $u\in S_j$ with
$d^{(a)}(u)=1$ and $d_{\cals_j}(u)\geq 2$ is
\[\sum_{a'\neq a}\sum_{t=1}^{k-1}O(\la_{a,1}\la_{a',t}/|S_j|)
=O(\la_{a,1}\la_{2,1}/|S_j|)=n^{-\inv{a+1}+o(1)}.\]
By Theorem~\ref{mt}, there are at most $n^{\d/2}\log n$ levels $S_i$. We take $\d<\frac{1}{a+1}$. Therefore, the probability that there is a vertex $w$ in any $S_i$ of size at most $n^{\frac{a-1}{a}+o(1)}$ such that  $d^{(a)}(w)\geq 1$ and $d_{\cals_j}(w)\geq 2$ is
$n^{-\frac{1}{a+1}+o(1)}n^{\d/2}\log n=o(1)$.
So \aas there are no such vertices $w$ and so whenever we are in Case 1, we can assume that
\begin{center}
(A1): for all $w\in S_j$ either $d^{(a)}(w)=0$ or $d_{\cals_j}(w)=1$.
\end{center}
 Therefore~(\ref{exx}) yields:
\[\ex(Y_{s+1}-Y_s)\leq -1+o(1)<-\hf,\]
where the $o(1)$ term accounts for the probability that our assumption (A1) fails (here we also use the fact that $|Y_{s+1}-Y_s|$ is always bounded).

{\em Case 2: $\la_{a,1}\geq \log^3n$}.  Then $X_{a,1}\geq Z_0\la_{a,1}-\log^2n>\hf Z_0\la_{a,1}$.
By Lemma~\ref{l:B}, we take $B=B(r,k)$ to be large enough that $|S_j|<\frac{Z_0}{10rk^3Z_1}n$, for $j\geq B$. Then:
\begin{eqnarray*}
\sum_{t=2}^{k-1}X_{a,t}<\sum_{t=2}^{k-1}(Z_1\la_{a,t}+\log^2n)
&=&Z_1\sum_{t=2}^{k-1}\left(\frac{|S_j|}{n}\right)^{(a-1)(t-1)}\la_{a,1}+(k-2)\log^2n\\
&<&kZ_1\frac{|S_j|}{n}\la_{a,1}+(k-2)\log^2n,\qquad\mbox{ since } a\geq 2\\
&<&\frac{Z_0}{8rk^2}\la_{a,1}\\
&<&\frac{X_{a,1}}{4rk^2}.
\end{eqnarray*}
Since $d_{\cals_j}(w)\leq k-2$ we have:
\[\sum_{w\in S_j}(d_{\cals_j}(w)-1)d^{(a)}(w)
\leq(k-3)\sum_{t=2}^{k-1}tX_{a,t}<k^2\sum_{t=2}^{k-1}X_{a,t}<\frac{X_{a,1}}{4r}
\leq\frac{1}{4r}\sum_{x\in S_j} d^{(a)}(x).\]
Furthermore $\sum_{x\in S_j} d^{(a)}(x)\geq X_{a,1}>\frac{Z_0}{2}\log^3n$, and so if $|W|<r\log^2n$ then 
\[
\sum_{x\in  W} d^{(a)}(x){<kr\log^2 n}<\hf \sum_{x\in S_j} d^{(a)}(x).
\]
 Thus, we will have the inequality
 \[\sum_{w\in S_j\setminus W}(d_{\cals_j}(w)-1)d^{(a)}(w)
<\frac{1}{2r}\sum_{x\in S_j\setminus W} d^{(a)}(x)\]
 holds as long as $|W|<r\log^2n$.
 Therefore, for $0\leq s\leq \log^2 n$ ({in which range $|W|\le r\log^2 n$ is guaranteed}),
we have
\[\ex(Y_{s+1}-Y_s)
< -1+r\frac{\sum_{w\in S_j{\setminus W}}(d_{\cals_j}(w)-1)d^{(a)}(w)}{\hf\sum_{x\in {S_j\setminus W}} d^{(a)}(x)}
\leq -\hf.\]
A.a.s. $Y_0=d_{\cals_j}(v)<\log n$ (as \aas the maximum degree in $H_r(n,p=c/n^{r-1})$ is less than $\log n$). Now define $Y'_s=Y_s$ for all $0\le s\le \log^2 n$ and $Y'_s=Y'_{s-1}-1$ for all $s>\log^2 n$. Note that $|Y_{s+1}-Y_s|$ is bounded always and so is $|Y'_{s+1}-Y'_s|$. By applying Lemma~\ref{l:azuma}, with probability at least $1-o(n^{-2})$, $Y'_t$ becomes $0$ before $t$ reaches $\log^2 n$. Hence, with probability at least $1-o(n^{-2})$, $Y_s=Y'_s$ for all $0\le s\le \tau<\log^2n$, where $\tau$ denotes the first iteration $t$ that $Y_t=0$.  It follows from $\ex(Y'_{s+1} -Y'_s)<-1/2$ for all $s\ge 0$ that the expected number of iterations before $Y'_s=0$ is at most $2Y'_0=2Y_0=2d_{\cals_j}(v)$ (via standard methods, e.g.\ the optional stopping theorem). Since $\tau=O(n)$ always and we have proved that $\pr((Y_{t\wedge \tau})_{t\ge 0}\neq (Y'_{t\wedge \tau})_{t\ge 0})\le\pr(\tau>\log^2 n)=o(n^{-2})$, it follows that $\ex\tau \le 2d_{\cals_j}(v)+o(n^{-1})$. Since each iteration adds at most $r-1<r$ new vertices to $W$, (\ref{erec}) follows.
\proofend

Having completed the proof of Lemma~\ref{lrec}, we now augment it with a concentration argument to prove:

\begin{lemma}\lab{lrconc} If Lemma~\ref{lsi}(a,b,c,d,e,f) hold, then there are constants $B=B(r,k),Z=Z(r,k)>0$ such that with probability at least
$1-n^{-3}$: for all $j\geq B$ with $\ld_j>n^{\d}\log^2n$, we have

\[|R_j(v)|\leq |R_{j+1}(v)|+Z\frac{\ld_j}{n}\sum_{\ell=i}^{j+1}|R_{\ell}(v)|
+ n^{\d}.\]
\end{lemma}

\proofstart  We will take $B$ from Lemma~\ref{lsi}, $Z_2$ from  Lemma~\ref{lrec}, and set $Z=8Z_2$.
By Lemma~\ref{lsi}(d), it suffices to prove that with probability at least
$1-n^{-5}$, we have
\begin{equation}\label{econc}
|R_j(v)|\leq |R_{j+1}(v)|+\frac{Z}{4}\frac{|S_j|}{n}\sum_{\ell=i}^{j+1}|R_{\ell}(v)|
+ n^{\d}.
\end{equation}

By Lemma~\ref{lrec}, we have
\[\ex(|R_j|\mid R_i,...,R_{j-1})\leq |R_{j+1}|+Z_2\frac{|S_j|}{n}\sum_{\ell=i}^{j+1}|R_{\ell}|+\log^2n.\]

We will apply McDiarmid's Inequality, a variation of Talagrand's Inequality to show that $|R_j(v)|$ is concentrated.
We use the version stated in ~\cite{mrbook}:

\noindent{\bf McDiarmid's Inequality}\cite{cm}
{\em Let $X$ be a non-negative random variable determined by
independent trials $T_1,...,T_{m}$ and independent permutations $\Pi_1,...,\Pi_{m'}$.
We call the outcome of one trial $T_{\imath}$, or the mapping of a single element in a permutation $\Pi_{\imath}$, a {\em choice}.
Suppose that for every
set of possible outcomes of the trials and permutations, we have:
\begin{enumerate}
\item[(i)] changing the outcome of any one trial can affect $X$
by at most $\varrho$;
\item[(ii)] interchanging two elements in any one permutation can affect $X$
by at most $\varrho$; and
\item[(iii)] for each $s>0$, if $X\geq s$ then there is a set
of at most $qs$ choices whose outcomes certify that $X\geq s$.
\end{enumerate}
Then for any $t\geq0$, we have
\[\pr(|X-\ex(X)|>t+25\varrho\sqrt{q\ex(X)}+128\varrho^2q)\leq 4e^{-\frac{t^2}{32\varrho^2q(\ex(X)+t)}}.\]
}

Because every vertex in $R_{j+1}$ has at least one neighbour in $S_j$, we always have $|R_j|\geq |R_{j+1}|$.  So we will apply McDiarmid's Inequality to
\[X=|R_j|-|R_{j+1}|.\]
We will show below that we can take $\varrho=\log^4 n$ and $q=1$.
Setting $t=\hf\max\{\ex(X),n^{\d}-\log^2n\}$, we have $25\varrho\sqrt{\ex(X)}+128\varrho^2)< t$.
So McDiarmid's Inequality yields:
\begin{equation}\label{emd}
\pr(X>2t)\leq 4e^{-\frac{t^2}{32\log^4 n(\ex(X)+t)}}<e^{-\d n/50}=o(n^{-5}).
\end{equation}

In step 1 of EDGE-SELECTION, we expose the components of the subhypergraph induced by $\cals_j$.
We take $B$ to be as large as required in the proof of~(\ref{erec}) above.  Then, a straightforward argument, such as that in~\cite{mr1}, shows
the probability that at least one component of $\cals_j$ has size greater than $\log^2 n$ is
$o(n^{-5})$.  So we will assume that every component has size at most $\log^2 n$.

The degree of any vertex in $\calh_{\dd}(n,p=c/n^{\dd-1})$ is distributed like the binomial variable $BIN({n-1\choose r-1},p)$ and so with probability at least $1-o(n^{-5})$, every degree is less than $\log^2 n$.

In step 2 of EDGE-SELECTION, we first choose a uniformly random bipartite hypergraph
$(\cals_j,\cals_{j+1})$ and then take a random permutation of the vertices in $\cals_{j+1}$.
That permutation will be $\Pi_1$ in applying McDiarmid's Inequality (with $m'=1$).

The degree of each vertex is distributed like a truncated Poisson with mean less than 2 (for sufficiently large $B$) and truncated as being at least one.  With probability at least $1-o(n^{-6})$, every degree is less than $\log^2 n$.  So each vertex in $R_{j+1}$ is linked to
at most $\log^2 n$ vertices in $\cals_j$, each of which lies in a component of size at most $\log^2 n$.  Thus each {\em choice} affects $X$ by at most $\varrho=\log^4 n$.

In step 3 of EDGE-SELECTION, we choose a set of $r-a-b$ uniformly random vertices for each Type $(a,b)$ hyperedge containing a vertex of $\cals_j$.  These  are the trials $T_1,...,T_m$ in
applying McDiarmid's Inequality (with $m$ being the total number of Type $(a,b)$ hyperedges over all $a,b$).  Each trial results in at most
one vertex of $\cals_j$ being adjacent to $\cup_{\ell=i}^{j+2} R_{\ell}$ and that vertex lies in a component of $\cals_j$ with size at most $\log^2 n$. So each trial affects $X$ by at most  $\varrho=\log^4 n$.

Finally, we note that if $X\geq s$ then there are at most $s$ {\em choices} whose outcomes
certify $X\geq s$: sets of vertices chosen in step 3 where one of them is in $\cup_{\ell=i}^{j+2} R_{\ell}$, or vertices of $R_{j+1}$ which, in the permutation of step 2, are mapped to a vertex whose neighbours in $\cals_j$ have components in $\cals_j$ of sizes totalling at least 2. This
justifies taking $q=1$, and completes the proof of~(\ref{emd}). So with probability at least $1-n^{-5}$ we have:
\bean
|R_j(v)|&=&|R_{j+1}(v)| +X\leq |R_{j+1}(v)| +\ex(X)+2t\leq |R_{j+1}(v)| + 2\ex(X)+n^{\d}\\
&\leq& |R_{j+1}(v)|+\frac{Z}{4}\frac{|S_j|}{n}\sum_{\ell=i}^{j+1}|R_{\ell}(v)|
+ n^{\d}.
\eean

\proofend

Because Lemma~\ref{lrconc} holds only for $\ld_j>n^{\d}\log^2n$, we define
\[i_0 \mbox{ to be the smallest integer such that }|\ld_{i_0+1}|< n^{\d}\log^2n.\]
We will prove:
\begin{lemma}\lab{lrec3}
There is a constant $X'=X'(r,k)>0$ such that for any $1\leq i\leq i_0$ and any $v\in S_i$:
\[\pr(|R(v)|>n^{X'\d})<\inv{n^2}.\]
\end{lemma}

This yields Lemma~\ref{lrec2} as follows:

\newtheorem*{lrec2}{Lemma~\ref{lrec2}}
\begin{lrec2}
There is a constant $X=X(r,k)>0$ such that for any $1\leq i\leq \imax$ and any $v\in S_i$:
\[\pr(|R(v)|>n^{X\d})<\inv{n^2}.\]
\end{lrec2}

\proofstart Lemma~\ref{lrec3} implies this lemma for all $i\leq i_0$.  So consider some $i>i_0$, $v\in S_i$
and note that $R(v)\subseteq \left(\cup_{\ell\geq i_0+1} S_{\ell}\right) \cup\left(\cup_{u\in S_{i_0}}R(u)\right)$.

By Lemma~\ref{lsi}(c), and the fact that {$\inv{4}|S_{\ell}|\leq \ld_{\ell}\leq (k-1)|S_{\ell}|$ (by Lemma~\ref{lsi}(d))}, we have
\[\sum_{\ell\ge i_0+1}|S_{\ell}|\leq4\sum_{\ell\ge i_0+1}|\ld_{\ell}|\le 4Z_1\ld_{i_0+1}n^{\d/2}<\frac{4Z_1}{k-1}|S_{i_0+1}|n^{\d/2}n^{3\d}{=o(n^{6\d})}.\]
By Lemma~\ref{lsi} (a,b,d), we have $|S_{i_0}|<5|S_{i_0+1}|= O(n^{\d}\log^2n)$. These, and Lemma~\ref{lrec3} yield:
\[|R(v)|\leq  n^{6\d} + O(n^{\d\log^2n})\times n^{X'\d}<n^{X\d},\]
for $X=X'+6$.
\proofend

{\em Proof of Lemma~\ref{lrec3}:}
 Lemma~\ref{lrconc} bounds $|R_j|$ in terms of $|R_{\ell}|$ for values of $\ell$ that are {\em larger} than $j$.  We will perform a change of indices in order to analyze this using a recursive equation where values are bounded by values with {\em smaller} indices.  The value $r_0$ will correspond to $R_i$, and  $r_j$ will correspond to $R_{i-j}$.  Specifically:

\begin{eqnarray*}
r_0&=&1\\
r_{j}&=&r_{j-1}+Z\frac{|\ld_{i-j}|}{n}\sum_{\ell=0}^{j-1}r_{\ell}+n^{\d},\quad \forall j\ge 1.
\end{eqnarray*}
Thus, by Lemma~\ref{lrconc} we have
\[|R_{i-j}(v)|\leq r_i ,\quad \forall \ 0\leq j\leq i-B.\]
It will be convenient to define:
\[t_j=\sum_{\ell=0}^j r_j,\]
and so
\begin{equation}\lab{rec}
r_{j}=r_{j-1}+Z\frac{\ld_{i-j}}{n}t_{j-1}+n^{\d},\quad \forall j\ge 1,
\end{equation}
and
\[|R(v)|<t_{i-B}.\]

We solve the recurrence~\eqn{rec}. Since $r_j=t_j-t_{j-1}$, we have
$$
t_{j}-t_{j-1}= t_{j-1}-t_{j-2}  + Z\frac{\ld_{i-j}}{n}t_{j-1}
+ n^{\d}, \quad \forall j\ge 1,
$$
where $t_0=1$ and $t_{-1}=0$.
We solve this recurrence. It will be helpful to find sequences $(a_j)_{j\ge 0}$ and $(b_j)_{j\ge 1}$ that satisfy
\begin{equation}
t_{j}-a_jt_{j-1}=b_j(t_{j-1}-a_{j-1}t_{j-2})+n^{\d},\quad \forall j\ge 1.\lab{ezzzz}
\end{equation}
Rearranging gives
\[t_{j}-t_{j-1}=(a_j-1+b_j)t_{j-1}+(-a_{j-1}b_j)t_{j-2}+n^{\d},\]
so we require that for all $j\ge 1$,
\bean
a_j-1+b_j&=&1+ Z\frac{\ld_{i-j}}{n}\\
 -a_{j-1}b_j&=&-1,
\eean
with the initial condition $a_{0}=b_{1}=1$.

So we define $a_j,b_j$ recursively as:
\[
b_j=1/a_{j-1},\quad a_j=2-\frac{1}{a_{j-1}}+ Z\frac{\ld_{i-j}}{n},
\]
and~(\ref{ezzzz}) holds. Since $a_j+\inv{a_{j-1}}\ge 2$ for each $j$ and $a_{0}=1$, it follows that $a_j\ge 1$ for every $j\ge 0$.   Next, we show that
there is a constant $D=D(r,k)>0$  such that for every $j\ge 1$,
\begin{equation}
1+Z\frac{\ld_{i-j}}{n}\le a_j\le 1+D\sqrt{\frac{\ld_{i-j}}{n}}.\lab{induction}
\end{equation}
The lower bound follows immediately from $a_{j-1}\geq 1$ and the recursion
$a_j=2-\frac{1}{a_{j-1}}+ Z\frac{\ld_{i-j}}{n}$.  We prove the upper bound
by induction. By taking  $D\geq Z$ we ensure that~\eqn{induction} holds for $j= 1$. Now
 assume that $j\ge 2$ and that~\eqn{induction} holds for $j-1$, and so:
\[
\frac{1}{a_{j-1}}\geq 1-D\sqrt{\frac{\ld_{i-j+1}}{n}}.\]
Since $a_j=2+ Z\frac{\ld_{i-j}}{n}-1/a_{j-1}$, we have
\[ a_j\le 1+Z\frac{\ld_{i-j}}{n}+D\sqrt{\frac{\ld_{i-j+1}}{n}}.
\]
Since $i-j\leq i\leq i_0$, we have $\ld_{i-j}>n^{\d}\log^2 n$. So
  by Lemma~\ref{lsi}(b), we have that for all $0\le j\le i-B$, $\ld_{i-j+1}\le (1-Y_2\sqrt{\ld_{i-j}/n})\ld_{i-j}$ for constant $Y_2=Y_2(r,k)>0$.
Thus,
\bean
1+Z\frac{\ld_{i-j}}{n}+D\sqrt{\frac{\ld_{i-j+1}}{n}}&\le & 1+Z\frac{\ld_{i-j}}{n}+\left(1-\frac{Y_2}{2}\sqrt{\ld_{i-j}/n}\right)D\sqrt{\frac{\ld_{i-j}}{n}}\\
&=&1+D\sqrt{\frac{\ld_{i-j}}{n}}-\left(\frac{DY_2}{2}-Z\right)\frac{\ld_{i-j}}{n}\le 1+D\sqrt{\frac{\ld_{i-j}}{n}}
\eean
where the last inequality holds
by choosing $D>2Z/Y_2$. Thus,~\eqn{induction} holds also for $j$ and thus it holds for every $j\le i-B$.

Now we continue to solve the recurrence~\eqn{rec}. Let $c_j=t_{j}-a_jt_{j-1}$. Then, since $b_j=1/a_{j-1}\le 1$ for every $j\ge 1$, we have
\[
c_{j}=b_j c_{j-1}+n^{\d} \le c_{j-1}+n^{\d}\le c_{0} +j n^{\d}=1+j n^{\d}.
\]
 Since $j=O(n^{\d/2}\log n)<n^{\d}$ by Theorem~\ref{mt}(b), this yields
\be
t_{j}-a_jt_{j-1} \le 1+j n^{\d}\le U:=n^{\d}\log^2n.\lab{rec2}
\ee

Let $\ell_0$ be the maximum integer for which $\ld_{\ell_0}\ge n^{1-\d}$.  Again applying Theorem~\ref{mt}(b), we have $\ell_0 =O(n^{\d/2}\log n)$.

Now is a good time to recall that our goal is to show $t_{i-B}=n^{O(\d)}$.

\eqn{rec2} says $t_{j}\leq 1+ a_jt_{j-1} + U$. Applying this
recursively yields that for every $1\le \ell\le i-B$,
\be
t_{i-B}\le t_{i-B-\ell}\prod_{h=0}^{\ell-1}a_{i-B-h}+U\left(1+\sum_{h_2=0}^{\ell-2}\prod_{h=0}^{h_2} a_{i-B-h}\right).
\lab{rec3}
\ee

Since $a_h\ge 1$ for each $h$, we have $1+\sum_{h_2=0}^{\ell-2}\prod_{h=0}^{h_2} a_{i-B-h}\le \ell\prod_{h=0}^{\ell-2} a_{i-B-h}$. Now taking $\ell=\ell_0-B+2$ in~(\ref{rec3}) and noting that
$a_{i-\ell_0-1}<U\ell_0$ by~\eqn{induction}  yields
\be
t_{i-B}\le   t_{i-\ell_0-2}\left(\prod_{j=0}^{\ell_0-B+1} a_{i-B-j}\right)+U \ell_0 \prod_{j=0}^{\ell_0-B} a_{i-B-j}\le (1+t_{i-\ell_0-2})U \ell_0 \prod_{j=0}^{\ell_0-B} a_{i-B-j}.\lab{ezz}
\ee
Again applying~\eqn{induction}, we have:
\be
\prod_{j=0}^{\ell_0-B} a_{i-B-j}\le\exp\left(D\sum_{j=0}^{\ell_0-B} \sqrt{\frac{\ld_{j+B}}{n}}\right)
=\exp\left(D\sum_{j=B}^{\ell_0} \sqrt{\frac{\ld_j}{n}}\right).
\lab{ezz2}
\ee

Next, we bound $\exp\left(D\sum_{j=B}^{\ell_0} \sqrt{\frac{\ld_j}{n}}\right)$. By Lemma~\ref{lsi}(b), we have for all $j>i_0$,
\[\ld_{j+1}\leq \exp\left(-Y_2\sqrt{\frac{\ld_{j}}{n}}\right)\ld_j,\]
and so
\[\ld_{\ell_0+1}\le\exp\left(-Y_2\sum_{j=B}^{\ell_0}\sqrt{\frac{\ld_j}{n}}\right)\ld_B.
\]
By the definition of $\ell_0$ and Lemma~\ref{lsi}(a), we have $\ld_{\ell_0+1}>\hf n^{1-\delta}$, so
\[
\exp\left(Y_2\sum_{j=B}^{\ell_0}\sqrt{\frac{\ld_j}{n}}\right)\le \frac{\ld_B}{\ld_{\ell_0+1}}< \frac{n}{\ld_{\ell_0+1}}=O(n^{\delta}),
\]
and so
\[
\exp\left(D\sum_{j=B}^{\ell_0}\sqrt{\frac{\ld_j}{n}}\right)=n^{O(\d)},
\]
since both $D$ and $Y_2$ are positive constants.
This, (\ref{ezz}),~(\ref{ezz2}), $\ell_0=O(n^{\d/2}\log n)$ and $U=n^{\d}\log^2n$ yield
\be
t_{i-B}=n^{O(\d)}t_{i-\ell_0-2}.
\lab{rec4}
\ee
It only remains to show that $t_{i-\ell_0-2}=n^{O(\d)}$.
The same recursion that produced~(\ref{rec3})  yields that for every $1\le \ell\le i-\ell_0-2$,
\[
t_{i-\ell_0-2}\le t_{i-\ell_0-2-\ell}\prod_{h=0}^{\ell-1}a_{i-\ell_0-2-h}+U\left(1+\sum_{h_2=0}^{\ell-2}\prod_{h=0}^{h_2} a_{i-\ell_0-2-h}\right).
\]
Arguing as for~\eqn{ezz}, this time taking $\ell=i-\ell_0-2$ yields
\be
t_{i-\ell_0-2}\le (1+t_0)Ui\prod_{j=0}^{i-\ell_0-4} a_{i-\ell_0-2-j}
=2Ui\prod_{j=0}^{i-\ell_0-4} a_{i-\ell_0-2-j}
=n^{O(\d)}\prod_{j=0}^{i-\ell_0-4} a_{i-\ell_0-2-j},
\lab{eq:t}
\ee
and arguing as for~\eqn{ezz2} yields
\be
\prod_{j=0}^{i-\ell_0-4} a_{i-\ell_0-2-j}\leq\exp\left(D\sum_{j=\ell_0+2}^{i-2} \sqrt{\frac{\ld_j}{n}}\right).
\lab{eq:tt}
\ee

By  the definition of $\ell_0$ we have $\ld_{\ell_0+2}< n^{1-\d}$.
Since $i\le i_0$, for every $\ell_0+2\le j\le i-2$, we have $\ld_{j}\ge n^{\d}\log^2 n$, and so we can apply Lemma~\ref{lsi}(a) to show
\[
\ld_{j}\le (1-Y_2n^{-\d/2})^{j-(\ell_0+2)}\ld_{\ell_0+2}
\le (1-Y_2n^{-\d/2})^{j-(\ell_0+2)}n^{1-\d}, \quad \forall\ \ell_0+2\le j\le i-2,
\]
which implies that
\[
\sum_{j=\ell_0+2}^{i-2}\sqrt{\frac{\ld_j}{n}}\le n^{-\d/2}\sum_{j\ge 0} (1-Y_2n^{-\d/2})^j =O(1).
\]
This proves that $t_{i-\ell_0-2}=n^{O(\d)}$ by~\eqn{eq:t} and~\eqn{eq:tt}. So by~\eqn{rec4} it completes our proof that $t_{i-B}=n^{O(\d)}$. Since $|R(v)|< t_{i-B}$, this proves the lemma.
\proofend

\section{Connectivity of the clusters}\lab{scon}
In this section, we complete the proof of Theorem~\ref{tc2} by proving Lemmas~\ref{lflip}, \ref{lgap} and~\ref{lTw}.  We also describe how to adapt the proof of Theorem~2 from \cite{amxor} to prove Theorem~\ref{tc1}(b).  Note that by choosing $\k$ to be sufficently large, we can take $\d$ to be as small a constant as we wish.

First, we prove that not many vertices lie on flippable cycles.

\newtheorem*{lflip}{Lemma~\ref{lflip}}
\begin{lflip}  For $k\geq 3$, $0<\d<\hf$ and $c=c_{\dd,2}+n^{-\d}$, \aas
the total sizes of all flippable cycles in $\calh_{\dd}(n,p=c/n^{k-1})$ is $O(n^{\d/2}\log n)$.
\end{lflip}

\proofstart
We follow very similar analysis to that in \cite{amxor}, but being more careful with the $o(n)$ terms.  We prove that \aas there is no collection of flippable cycles whose total size is $a$ for any $a\geq n^{\d/2}\log n$.  We condition on the degree sequence of the $2$-core and we generate a random $2$-core using the configuration model.

 Recall the definitions of $\mu(c)$, $\alpha(c)$ and $\beta(c)$ from the beginning of Section~\ref{sec:llt2}. By Lemma~\ref{lcoresize2}, \aas the $2$-core contains $Q$ vertices, with $Q=\alpha(c) n+ O(n^{3/4})$, and the total degree of the $2$-core is $\La=\dd\beta(c)n+O(n^{3/4})$.
Conditional on the values of $Q$ and $\La$, the degree distribution of the $2$-core follows the truncated multinomial $Multi(Q,\La,2)$. Let $Q_2$ denote the number of vertices in the 2-core with degree $2$. By Proposition~\ref{p:Poisson}, for any $\e>0$:
\be\lab{er2}
Q_2=\frac{e^{-\mu(c)} \mu(c)^2}{2f_2(\mu(c))}\a(c) n+O(n^{-\hf+\e}).
\ee

The number of choices of $a$ vertices {with degree $2$} is  ${Q_2\choose a}$.  We arrange them into oriented cycles; the number of ways to do this is  less than the number of permutations; i.e. is less than $a!$.

Given such an arrangement, we label the two vertex copies of vertex $i$ as $y_i,z_i$; there are $2^a$ choices.  Now for each pair of vertices $i,j$ where $j$ follows $i$ on an oriented cycle, the pair $\{y_i,z_j\}$ must lie in an $\dd$-tuple of our configuration.

We process those $a$ pairs in sequence, halting if we find that a pair does not lie in a $\dd$-tuple.
When processing $\{y_i,z_j\}$ we expose whether they appear in the same $\dd$-tuple, and we do not
expose any other members of that $\dd$-tuple.  So prior to processing the  $\ell$th pair, $\{y_i,z_j\}$, we have exposed information about exactly $2\ell-2$ copies.   Each of the remaining copies, other than $y_i$, is equally likely to be one of the $\dd-1$ copies in the same $\dd$-tuple as $y_i$ (and for $\dd\geq 3$, the exposed copies also have positive probability of being in that $\dd$-tuple).  So the probability that $y_i,z_j$ are in the same $\dd$-tuple is at most $(\dd-1)/(\La-2\ell+1)$. So the probability that there is some collection of flippable cycles with total size $a$ is at most
\[{Q_2\choose a}a!2^a\prod_{\ell=1}^a\frac{\dd-1}{\La-2\ell+1}
=\prod_{\ell=1}^a\frac{(\dd-1)(2Q_2-2\ell+2)}{\La-2\ell+1}.\]

We will prove below that
\be
\frac{2(r-1)Q_2}{\La}\leq 1-Kn^{-\d/2},\lab{eq2}
\ee
for some constant $K>0$.  Since $2(r-1)\leq 2$, it follows that $(\dd-1)(2Q_2-2\ell+2)/(\La-2\ell+1)<1-K n^{-\d/2}$ for each $\ell$ and so
$\ex(X_a)\leq(1-Kn^{-\d/2})^{2a}=o(1/n)$ for $a>\frac{2}{K}n^{\d/2}\log n$.  Summing over the fewer than $n$ values of  $a>\frac{2}{K}n^{\d/2}\log n$ yields the lemma.

It only remains to prove~(\ref{eq2}). Recalling the definitions of $\a(c),\b(c)$ from Section~\ref{sec:llt2} and substituting into~(\ref{er2}) yields:
\[ \frac{2Q_2}{\La}=\frac{2 \rho_2\alpha(c) }{\dd\beta(c)}+O(n^{-1/2+\eps})=e^{-\mu(c)}\frac{\mu(c)}{f_1(\mu(c))}+O(n^{-1/2+\eps}).\]

 From Lemma~\ref{l:degreeK} (see also~\eqn{pdk}),
\[
e^{-\mu_{\dd,2}}\frac{\mu_{\dd,2}}{f_1(\mu_{\dd,2})}= \frac{2\alpha p_{\dd,2}(2)}{\dd \beta}=\frac{1}{(\dd-1)}.
\]
It is easy to check that the derivative of
$e^{-\mu}\frac{\mu}{f_1(\mu)}$ with respect to $\mu$ is strictly negative in a small neighbourhood of $\mu_{\dd,2}$.   By Lemma~\ref{l:diff}, $\mu(c)= \mu_{r,2}+ K_1n^{-\d/2}+o(n^{-\d/2})$. Therefore,
\[e^{-\mu(c)}\frac{\mu(c)}{f_1(\mu(c))}<e^{-\mu_{\dd,2}}\frac{\mu_{\dd,2}}{f_1(\mu_{\dd,2})}- K_4n^{-\d/2}=\frac{1}{(\dd-1)}- K_4n^{-\d/2},\]
 for some constant $K_4>0$. This implies~(\ref{eq2}).
 \proofend

Recall the context from Section~\ref{sct}, particularly that we are considering the 2-core of $\calh_r(n,p=c/n^{k-1})$ where $c=c_{2,k}+n^{-\d}$ for some sufficiently small constant $\d>0$.

Given a non-2-core vertex $w$, recall that $T(w)$ is the set of vertices $v$ that can reach $w$ in $\cald$; i.e. the set of vertices $v$ such that $w\in R^+(v)$. For $u\in T(w)$, $T(w,u)$ is the subgraph of $T(w)$ containing all vertices reachable from $u$; i.e. vertices on paths from $u$ to $w$.

Recall that $I^*$ is the last iteration of the parallel stripping process during which a free variable is removed, and $u^*$ is a free variable removed during iteration $I^*$.

\newtheorem*{lgap}{Lemma~\ref{lgap}}
\begin{lgap}
 A.a.s\  $u^*$ is the only free vertex in
$\cup_{i\geq I^*-n^{\d/20}} S_i$.
\end{lgap}

\begin{proof}
Choose $\n=\hf - \d/6$. Run the SLOW-STRIP algorithm and let $i_1$ denote the first iteration in which  $L_{t(i_1)}\leq n^{\n}$; set $t_1=t(i_1)$.   Our first step will be to show
that \aas $I^*> i_1+n^{\d/20}$.
Let $i_2= i_1+n^{2\d/5}$ and $i_3= i_1+2n^{2\d/5}$.

 By Lemma~\ref{lsi}(c),  a.a.s.\ $\sum_{j\geq t(i_1)} L_{j}=O(L_{t_1}n^{\d/2})=O(n^{\hf+\d/3})=o(n^{1-\d})$ for small $\d$.
Thus, \aas at any iteration $t\geq t_1$ of SLOW-STRIP, the total degree of vertices in $\msq$ is $o(n^{1-\d})$ and so $L_t<n^{1-\d}$.

Therefore, we can apply Theorem~\ref{lsi}(a) ({recursively}) to obtain that for all $i_1\leq i\leq i_3$:
\be\lab{br2}
\ld_{i}>\left(1-Y_1n^{-\d/2}\right)^{2n^{2\d/5}}\ld_{i_1}=(1-o(1))\ld_{i_1}.
\ee
This is valid since in each recursion $i_1\le i\le i_3$, we have $\ld_i=\Omega(\ld_{i_1})\ge n^{\d}\log^2 n$ and so the assumption of Theorem~\ref{lsi} is satisfied. Theorem~\ref{lsi}(a) also implies that for all ${i_1\le i\le i_3}$, $\ld_i=O(n^{\n})$.
Since a.a.s.\ for every $i$, $\ld_i=\Theta(|S_i|)$ by Theorem~\ref{lsi}(d), we have a.a.s.\ $|S_i|=\Theta(n^{\n})$ for all $i_2\leq i\leq i_3$.

Recall that $S_i$ contains a free variable iff some hyperedge contains two vertices of $S_i$.
The calculations from the proof of Lemma~\ref{lsi}(e) say the probability that such a hyperedge exists is $\Theta(|S_i|^2/n)$.  So the expected number of free variables created between iterations $i_2$ and $i_3$ is $\Omega(n^{2\n-1}\cdot n^{2\d/5})>n^{\d/20}$. Applying the Chernoff bound as in Lemma~\ref{lsi}(e), shows that a.a.s.\ there is a free variable formed between iterations $i_2,i_3$. Hence, a.a.s.\ $I^*\ge i_2>i_1+n^{\d/20}$.

Next we show that \aas there are no two levels  $i_1\leq i<i'\leq i+n^{\d/20}$ such that $S_i,S_{i'}$ both contain a free variable.  By Theorem~\ref{mt}, \aas the total number of iterations in the parallel stripping process is at most $O(n^{\d/2}\log n)$, and so the number of pairs of levels that are  within distance at most $n^{\d / 20}$ of each other is at most $O(n^{11\d/20} \ln n)$. So the expected number of  such pairs $S_i,S_{i'}$, each containing a free variable  is $O(n^{11\d/20} \ln n)\times O(n^{2(2\n-1)})<n^{-\d/10}=o(1)$.  So a.a.s.\ there is no such pair of levels $S_i,S_{i'}$.

Since $I^*>i_1+n^{\d/20}$, if there was another level $S_i$ containing a free variable with $i\geq I^*-n^{\d/20}$, then this would form a pair $(i,i')$ as above.  So \aas there is no such $S_i$, thus proving the lemma.
\end{proof}

\newtheorem*{lTw}{Lemma~\ref{lTw}}
\begin{lTw}  A.a.s. there is some $w\in S_{I^*-n^{\d/20}}$ such that
\begin{enumerate}
\item[(a)] $u^*\in T(w)$;
\item[(b)] no vertex of any core flippable cycle is in $T(w)$;
\item[(c)] the subgraph of $\cald$ induced by the vertices of $T(w,u^*)$ is a path.
\end{enumerate}
\end{lTw}

\begin{proof} As in the previous proof, we let  $\n=\hf - \frac{\d}{6}$. We run  SLOW-STRIP  and let $i_1$ denote the first iteration in which  $L_{t(i_1)}\leq n^{\n}$; set $t_1=t(i_1)$.
Note that, since $k=2$, each iteration of SLOW-STRIP deletes a vertex $v\in\msq$. As usual,
we use the configuration model. $\calc_2$ is the 2-core.

We will prove that \aas for every $w\in \cup_{i\geq i_1}S_i$, {$T(w)\setminus \calc_2$} induces a tree in $\cald$.
This will prove parts (a,c) by letting $w$ be any vertex that $u^*$ can reach in $S_{I^*-n^{\d/20}}$ (such $w$ exists since every $v\in S_i$ can reach at least one vertex in $S_j$ for every $j<i$).

We start by bounding the size of  each $T(w)\setminus\calc_2$.

Fix some $w\in S_i$, with $i\geq i_1$. We maintain a set $\calt(w)$ as follows.

Initially $\calt(w):=\{w\}$. Whenever we delete a  vertex $v\in\msq$ such that $v\in\calt(w)$:
(i) each neighbour of $v$ that has degree at most $2$, and hence is in or will enter $\msq$, is placed in $\calt(w)$; (ii) each neighbour of $v$ that has degree greater than $2$ is coloured Red.
Every time a Red vertex enters $\msq$, it is placed in $\calt(w)$.  Thus, when we finish SLOW-STRIP, $\calt(w)=T(w)\setminus\calc_2$.

We will analyze $\calt(w)$ using a branching process.  {When a vertex $v\in\calt(w)$ is deleted by SLOW-STRIP, we say that we are {\em processing } $v$.  If a vertex $u$ is added to $\calt(w)$ while $v$ is being processed then we consider $u$ to be an {\em offspring} of $v$.
If a vertex $u$ is added to $\calt(w)$ during the deletion of a vertex not in $\calt(w)$ (and so $u$ must be Red), then we consider $u$ to be an {\em offspring} of the most recently processed member of $\calt(w)$.
Note that the offspring of $v$ are not neccessarily adjacent to $v$ in $\cald$. }

In other words: we say that $u$ is an {\em offspring} of a vertex $v\in \calt(w)$, if $u$ entered $\calt(w)$ between the iterations of SLOW-STRIP ranging from the time we remove $v$ up until just before the next iteration where we remove a member of $\calt(w)$.

There are two scenarios
under which $u$ can become an offspring of $v$: (i) at the time we delete $v$, $u$ is a neighbour of $v$ and $u$ has degree at most 2; (ii) on the step after $v$ is deleted, $u$ is Red  and  $u$ enters $\msq$ before the next member of $\calt(w)$ is deleted.  For case (ii) to occur, $u$ must be the neighbour of another vertex $v'\notin\calt(w)$ that is removed from $\msq$, and the degree of $u$ must drop below 2 when $v'$ is removed.

It will be convenient to consider $\calt'(w)\subseteq\calt(w)$, which differs from $\calt(w)$ in that only the first $n^{2\d}$ vertices to be coloured Red can enter $\calt'(w)$.  We will show that, in fact, $\calt'(w)=\calt(w)$.

When removing  $v\in\calt'(w)$,
the analysis used for~\eqn{br} shows that the expected number of offspring created under scenario (i) is $1+\br(G_t)+O(L_t/n)$. It follows from~\eqn{br2} that this is $1-\Theta(n^{-\d/2})$ for all $t\geq t(i)\geq t(i_1)$.

The number of iterations until the next member of $\calt'(w)$ is removed is at most $\t-t_1=O(n^{\n+\d/2})$, by Lemma~\ref{lsi}(c). If $u$ is an offspring of $v$ created under scenario (ii), then $u$ must be one of the first $n^{2\d}$ Red vertices. Furthermore during one of those iterations,  exactly two vertex-copies of $u$ remain and one of them is selected.  The total number of such vertex-copies over all choices of $u$ is  at most $2n^{2\d}$, and since there are a linear number of vertex-copies to choose from, the expected number of offspring of $v$ created under scenario (ii) is at most
\[O(n^{\n+\d/2})\times O(n^{2\d}/n)=O(n^{\n+5\d/2-1})=o(n^{-\d/2}),\]
for sufficiently small $\d$.

Therefore, the total expected number of offspring of $v$, in $\calt'(w)$, is $1-\Theta(n^{-\d/2})+o(n^{-\d/2})=1-z$ for some $z=\Theta(n^{-\d/2})$.  So $\calt'(w)$
grows like a Galton-Watson branching process with branching parameter $1-z$.  The probability that such a branching process has size at least $x$ drops quickly as $x$ exceeds $\Theta(z^{-2})$ (see, eg.~\cite{bbckw}), and in particular, $\Pr(|\calt'(w)|>n^{3\d/2})=o(1/n)$.
So \aas $|\calt'(w)|\leq n^{3\d/2}$ for every $w$.

Note that at most $r-1$ Red vertices are formed each time a member of $\calt'(w)$ is removed.
So \aas the number of Red vertices is at most $(r-1)n^{3\d/2}<n^{2\d}$ and so $\calt(w)=\calt'(w)$ for all $w$. Therefore \aas
\[|T(w)\setminus\calc_2| \leq n^{3\d/2} \mbox{ for every } w\in \cup_{i\geq i_1}S_i.\]

Now we show that \aas each $T(w)\setminus\calc_2$ induces a tree in $\cald$.

{\bf Observation:}  If $T(w)\cap\calc_2$ does not induce a tree in $\cald$, then there must have been an iteration when we deleted some $v\in\calt(w)$ and one of the $r-1$ neighbours of $v$ was in $\calt(w)$.

Each vertex has degree at most 1 after entering $\calt(w)$.
When we delete $v$, we choose $r-1$ uniform vertex-copies as its neighbours. Each vertex of $\calt(w)$ has entered $\msq$ and so has at most 2 copies remaining. So the probability that we choose  a copy of a vertex in $\calt(w)$  is  $|\calt(w)|/\Theta(n)=O(n^{3\d/2 - 1})$.
So the probability that this happens during the deletion of at least one member of $\calt(w)$ is $O(n^{3\d/2})\times O(n^{3\d/2 - 1})=O(n^{3\d - 1})$.

 Multiplying by the $O(n^{\n+\d/2})$ choices for $w\in\cup_{i\geq i_1}S_i$ (by Lemma~\ref{lsi}(c))) we get $O(n^{\n+7\d/2-1})=o(1)$
for $\d$ sufficiently small.
Therefore \aas $T(w)\cap\calc_2$ induces a tree in $\cald$ for every $w\in\cup_{i\geq i_1}S_i$. This proves (a,c) as described above.

To prove (b), we will show that \aas no $w\in \cup_{i\geq i_1}S_i$ has the vertex of any core flippable cycle in $T(w)$.  It will suffice to show that no $w'\in \cup_{i\geq i_1}S_i$ ever selects a vertex of a core flippable cycle as a neighbour during Step 3 of EDGE-EXPOSURE.  There are $O(n^{\n+\d/2})$ choices for $w'$, as stated above.  The total sizes of the core flippable cycles is $O(n^{\d/2}\log n)$ by Lemma~\ref{lflip}. So the probability that some such $w'$ selects a vertex of a core flippable cycle during Step 3 of EDGE-EXPOSURE is $O(n^{\n+\d/2})\times O(n^{\d/2}\log n)/n=O(n^{\n+\d-1}\log^2n)=o(1)$ for small $\d$.
\end{proof}

And finally, we describe the proof that solutions in different clusters are not
$\Theta(n^{1-r\d})$-connected.

{\em Proof of Theorem~\ref{tc1}(b)}
This follows the same argument as the proof of Theorem 2 of~\cite{amxor}.  The only change is to  Lemma~51 of~\cite{amxor}, where instead of proving that \aas there is no non-empty linked set of size less than $\a n$, we prove that \aas there is none of size less than $n^{1-r\d}$.
(Caution:  in~\cite{amxor} the usage of $k,r$ is inverted from that of this paper.)

As in~\cite{amxor}, we use $X_a$ to denote the number of linked sets $S$ with $|\Gamma(S)|=a$ (see~\cite{amxor} for definitions).
Property (iii) at the beginning of the proof of Lemma~51, is equivalent to saying $\frac{2(r-1)Q_2}{\La}\leq 1-\z$, for some $\z>0$, in the notation of this paper.  Instead, we have $\frac{2(r-1)Q_2}{\La}\leq 1-Kn^{-\d/2}$ by(~\ref{eq2}).  This results in replacing $Z_1=\Theta(1)$ from the proof of Lemma~51 of~\cite{amxor} with $Z_1=\Theta(n^{\d/2})$.
In the notation of~\cite{amxor}), this yields
\[\ex(X_a)<\left(\frac{\Theta(an^{r\d})}{n}\right)^{a/2r},\]
and it then follows easily that $\ex\left(\sum_{a=1}^{Zn^{1-r\d}}X_a\right)=o(1)$
 for sufficiently small $Z=Z(r)>0$. This proves the theorem.
\proofend


\begin{thebibliography}{99}

\bibitem{abg} L. Addario-Berry, N. Broutin and C. Goldschmidt. {\em The continuum limit of critical random graphs.}


\bibitem{aco}
D. Achlioptas and A. Coja-Oghlan.
{\em Algorithmic barriers from phase transitions.}
 In 49th Annual IEEE Symposium on Foundations of Computer Science,
               FOCS 2008, October 25-28, 2008, Philadelphia, PA, USA, pages 793--802. IEEE Computer Society, 2008.

\bibitem{aminfede}
D. Achlioptas, A. Coja-Oghlan and F. Ricci-Tersenghi
{\em On the solution-space geometry of random constraint satisfaction problems.}
{Random Struct. Algorithms},
{\bf 38} (3),
251-268 (2011).

\bibitem{amxor} D. Achlioptas and M. Molloy. {\em The solution space geometry of random linear equations.}  Random Structures and Algorithms (to appear).

\bibitem{art}
D. Achlioptas and F. Ricci-Tersenghi
{\em Random formulas have frozen variables.}
SIAM J. Comput.,
{\bf 39} (1),
260-280 (2009).


\bibitem{azuma} K. Azuma. {\em Weighted sums of certain dependent random variables.}
Tokuku Math. J. {\bf 19} (1967), 357~-~367.


\bibitem{bmz} A. Braunstein, M. Mezard and R. Zecchina. {\em Survey propagation: an algorithm for satisfiability}.   Random Structures and Algorithms {\bf 27} (2005), 201~-~226.


\bibitem{bb} B. Bollob\'{a}s. {\em A probabilistic proof of an asymptotic formula for the number of labelled regular graphs.} Europ. J. Combinatorics {\bf 1} 311-316 (1980).

\bibitem{bbcore} B. Bollob\'{a}s. {\em The evolution of sparse graphs.} Graph Theory and Combinatorics. Proc. Cambridge Combinatorial Conf. in honour of Paul Erd\H{̈o}s (B. Bollob ́as, ed.), Academic Press (1984), 35~–~57.

\bibitem{bb3} B. Bollob\'{a}s. {\em The evolution of random graphs.} Transactions of the AMS {\bf 286} (1984), 257~–~274.

\bibitem{bbckw}  B. Bollob\'{a}s, C. Borgs, J. Chayes, J. Kim and D. Wilson. {\em The scaling window of the 2-SAT transition.}  Random Structures and Algorithms {\bf 18} (2001), 201~-~256.

\bibitem{CW} Julie Cain and Nicholas Wormald, Encores on cores,
{\em Electron. J. Combin.}, 13 (2006), no. 1, Research Paper 81, 13 pp.

\bibitem{csw} J. Cain, P. Sanders and N. Wormald. {\em The random graph threshold for $k$-orientiability and a fast algorithm for optimal multiple-choice allocation.} Proceedings of 18th  SODA (2007), 469~-~476.

\bibitem{nc} N. Calkin. {\em Dependent sets of constant weight binary vectors.} Combinatorics, Probability and Computing, {\bf 6} (1997), 263~-~271.

\bibitem{cher} H. Chernoff. {\em A measure of asymptotic efficiency for tests of a
hypothesis based on the sum of observations.} Ann. Math. Statist. {\bf 23} 493~-~509 (1952).


\bibitem{cdmm} S. Cocco, O. Dubois, J. Mandler and R. Monasson.  {\em Rigorous decimation-based construction of ground pure states for spin glass models on random lattices.}
Phys. Rev. Lett. {\bf 90} (2003), 047205.


\bibitem{coalg}  A. Coja-Oghlan. {\em A better algorithm for random {k-SAT}.} SIAM Journal on Computing {\bf 39} (2010), 2823~-~2864.

\bibitem{coind} A. Coja-Oghlan and C. Efthymiou. {\em On independent sets in random graphs.} Proc. 22nd SODA (2011), 136~-~144.

\bibitem{cop} A. Coja-Oghlan, K. Panagiotou. {\em Catching the $k$-NAESAT threshold.} Proc. 44th STOC (2012), 899~-~908.

\bibitem{cop2} A. Coja-Oghlan, K. Panagiotou. {\em Going after the $k$-SAT threshold.} Proc. 45th STOC (2013), 705~-~714.

\bibitem{cov} A. Coja-Oghlan, D. Vilenchik. {\em Chasing the k-colorability threshold.} 	Proc. FOCS (2013).

\bibitem{cz} A. Coja-Oghlan and L. Zdeborova. {\em The condensation transition in random hypergraph 2-coloring.} Proc. 23rd SODA (2012), 241~-~250.

\bibitem{cc} C. Cooper. {\em The cores of random hypergraphs with a given degree sequence.}
 Random Structures Algorithms {\bf 25} (2004), 353~-~375.

\bibitem{cdxor} N. Creignou and H. Daud\'e.  {\em Satisfiability threshold for random XOR-CNF formulas.} Discrete Appl. Math. {\bf 96-97} (1999), 41~-~53.
\bibitem{daud}
H. Daud{\'e}, M. M{\'e}zard, T. Mora, and R. Zecchina.
{\em Pairs of SAT assignments and clustering in random boolean formulae.}
Theoretical Computer Science,
{\bf 393} (1-3),
260-279 (2008).

\bibitem{dmcore} A. Dembo and A. Montanari. {\em Finite size scaling for the core of large random hypergraphs.} Annals of Applied Probability, {\bf 18} (2008), 1993~-~2040.

\bibitem{cuc}
M. Dietzfelbinger, A. Goerdt, M. Mitzenmacher, A. Montanari, R. Pagh and M. Rink.
{\em Tight thresholds for cuckoo hashing via XORSAT.} Arxiv:0912.0287v3.
Also in Proceedings of  Automata, Languages and Programming, 37th International
               Colloquium, ICALP 2010, Bordeaux, France, July 6-10, 2010,
Part I, pages 213~-~225. Springer, 2010.


\bibitem{DKLP} J. Ding,  J.H. Kim,  E. Lubetzky and Y. Peres.
{\em Anatomy of a young giant component in the random graph.}
Random Structures Algorithms {\bf 39} (2011), 139~-~178.


\bibitem{dklp2} J. Ding,  J.H. Kim,  E. Lubetzky and Y. Peres.
{\em Diameters in supercritical random graphs via first passage percolation.} Combinatorics, Probability and Computing {\bf 19} (2010), 729~-~751.


\bibitem{dub}
O. Dubois and J. Mandler.
{\em The 3-XORSAT threshold.}
 In Proceedings of 43rd Symposium on Foundations of Computer Science (FOCS
               2002), 16-19 November 2002, Vancouver, BC, Canada, pages 769-778. IEEE Computer Society, 2002.

\bibitem{fern} D. Fernholz and V. Ramachandran. {\em The $k$-orientability thresholds for $G_{n,p}$.} Proceedings of SODA 2007, 459~–~468.

\bibitem{fkp} N. Fountoulakis,  M. Koshla and K. Panagiotou. {\em The multiple-orientability thresholds of random hypergraphs.}  Proceedings of  SODA 2011, 1222~-~1236.

\bibitem{frmix} N. Fountoulakis and B. Reed. {\em The evolution of the mixing rate of a simple random walk on the giant component of a random graph.} Random Structures and Algorithms {\bf 33} (2008), 68~-~86.

\bibitem{dg} D. Gamarnik and M. Sudan. {\em Limits of local algorithms over sparse random graphs.} arXiv:1304.1831.

\bibitem{pw} P. Gao and N. Wormald. {\em Load balancing and orientability thresholds for random hypergraphs.} Proceedings of STOC (ACM) (2010), 97~-~104.

\bibitem{young}
M. Guidetti and A.P. Young. {\em Complexity of several constraint-satisfaction problems using the heuristic classical algorithm WalkSAT.} Phys. Rev. E, {\bf 84} (1), 011102, July 2011.

\bibitem{hmwc} G. Havas, B.S. Majewski, N.C. Wormald, and Z.J. Czech. {\em Graphs, hypergraphs and hashing.} In 19th International Workshop on Graph-Theoretic Concepts in Computer Science (WG’93), Lecture Notes in Computer Science  {\bf 790} (1993) 153~–~165.

\bibitem{ikkm}
M. Ibrahimi, Y. Kanoria, M. Kraning, and A. Montanari.
{\em The set of solutions of random xorsat formulae.}
Proceedings of the Twenty-Third Annual ACM-SIAM Symposium
               on Discrete Algorithms, SODA 2012, Kyoto, Japan, January
               17-19, 2012, pages 760--779. SIAM, 2012.

\bibitem{jlkp}
S. Janson, T. {\L}uczak, D. Knuth, and B. Pittel. {\em The birth of the giant component.} Random Structures and Algorithms {\bf 3} (1993), 233~–~358.

\bibitem{jlr}
S. Janson, T. {\L}uczak and A. Ruci{\'n}ski.
{Random Graphs.} Wiley, New York (2000).

\bibitem{jhk} J.H.Kim.  {\em Poisson cloning model for random graphs.}  	arXiv:0805.4133v1

\bibitem{vk1} V. Kolchin. {\em Random graphs and systems of linear equations in finite fields.} Random Structures and Algorithms, {\bf 5} (1994), pp. 135–146. In Russian.
\bibitem{vk2} V. Kolchin and V. Khokhlov. {\em A threshold effect for systems of random equations of a special form.} Discrete Mathematics and Applications, {\bf 5} (1995), 425~–~436.


\bibitem{kmrsz} F. Krzakala, A. Montanari, F. Ricci-Tersenghi, G. Semerjian and L. Zdeborova. {\em Gibbs States and the Set of Solutions of Random Constraint Satisfaction Problems.} Proc. Natl. Acad. Sci., (2007).


\bibitem{lmss} M. Luby, M. Mitzenmacher, A. Shokrollahi, and D. Spielman.
{\em Efficient Erasure Correcting Codes.}
IEEE Transactions on Information Theory, {\bf 47} (2001), 569~-~584.

\bibitem{lmss2} M. Luby, M. Mitzenmacher, A. Shokrollahi, and D. Spielman
{\em Improved Low-Density Parity-Check Codes Using Irregular Graphs.}
IEEE Transactions on Information Theory, {\bf 47} (2001), 585~-~598.

\bibitem{tlcomp}  T. {\L}uczak.  {\em Component behaviour near the critical point of the random graph process.}
Rand. Struc. \& Alg. {\bf 1} (1990), 287~-~310.

\bibitem{tlcomp2}  T. {\L}uczak.  {\em Random trees and random graphs.} Rand. Struc. \& Alg. {\bf 13} (1998), 485~–~500.

\bibitem{lpw}  T. {\L}uczak, B. Pittel and J. Weirman. {\em The structure of a random graph at the point of the phase transition.} Trans. Am. Math. Soc. {\bf 341} (1994), 721~-~748.

\bibitem{mmw} 	E. Maneva, E. Mossel and M. J. Wainwright.
{\em A new look at Survey Propagation and its generalizations.}
JACM {\bf 54},  (2007).

\bibitem{cm}  C. McDiarmid. {\em Concentration for independent permutations.}
Combinatorics, Probability and Computing {\bf 11} (2002), 163~-~178.

\bibitem{mmbook}
M. M{\'e}zard and A. Montanari.
\newblock {\em Information, Physics, and Computation}.
\newblock Oxford University Press, Inc., New York, NY, USA, 2009.

\bibitem{mmz}
M. M{\'e}zard, T. Mora,  and R. Zecchina
{\em Clustering of solutions in the random satisfiability problem.}
Phys. Rev. Lett.,
{\bf 94} (19),
197205 (2005).

\bibitem{sp}
M. M{\'e}zard, G. Parisi, and R. Zecchina
{\em Analytic and algorithmic solution of random satisfiability problems.}
Science,
{\bf 297},
812~-~815 (2002).


\bibitem{mez}
M. M{\'e}zard, F. Ricci-Tersenghi, and R. Zecchina. {\em Two solutions to diluted $p$-spin models and XORSAT
problems.} J. Stat. Phys. {\bf 111}, 505~-~533, (2003).


\bibitem{mz}
M. Mezard, R. Zecchina {\em The random K-satisfiability problem: from an analytic solution to an efficient algorithm.} Phys. Rev. E {\bf 66}, (2002).

\bibitem{mmcore} M. Molloy {\em Cores in random hypergraphs and boolean formulas.}
Random Structures and Algorithms {\bf 27}, 124~-~135 (2005).

\bibitem{mmfreeze} M. Molloy. {\em The freezing threshold for $k$-colourings of a random graph.}
Proceedings of STOC (2012).

\bibitem{mr2} M. Molloy and B. Reed. {\em Critical Subgraphs of a Random Graph.} Electronic J. Comb. {\bf 6}  (1999),  R35.

\bibitem{mr1} M. Molloy and B. Reed. {\em A critical point for random graphs with a given degree sequence.}
Random Structures and Algorithms {\bf 6} 161~-~180 (1995).


\bibitem{mrbook}  M. Molloy and B. Reed.  {\em Graph Colouring and the
Probabilistic Method}. Springer (2002).

\bibitem{mres} M. Molloy and R. Restrepo. {\em Frozen variables in random boolean constraint satisfaction problems.} Proceedings of SODA (2013).


\bibitem{ms} M. Molloy and M. Salavatipour. {\em The resolution complexity of random constraint satisfaction problems.} SIAM J. Comp. {\bf 37}, 895~-~922 (2007).

\bibitem{mrt}
A. Montanari, R. Restrepo and P. Tetali. {\em Reconstruction and clustering in random constraint satisfaction problems.} SIAM J. Disc. Math. {\bf 25} (2011), 771~-~808.

\bibitem{mome} T. Mora and M. Mézard. {\em Geometrical organization of solutions to random linear Boolean equations.}
Journal of Statistical Mechanics: Theory and Experiment {\bf 10} (2006), P10007

\bibitem{mpwz} R. Mulet, A. Pagani,  M. Weigt and R. Zecchina.  {\em Coloring random graphs.}  Phys. Rev. Lett. {\bf 89}, 268701 (2002).

\bibitem{pittel} B. Pittel. {\em
On trees census and the giant component  in  sparse random graphs.}
Random Structures and Algorithms {\bf 1} (1990), 311~-~342.

\bibitem{ps} B. Pittel and G. Sorkin. {\em The satisfiability threshold for $k$-XORSAT.} 	arXiv:1212.1905 (2012).


\bibitem{psw} B. Pittel, J. Spencer and N. Wormald.
{\em Sudden emergence of a giant $k$-core in a random graph.}
J. Comb. Th. B {\bf 67}, 111~-~151 (1996).

\bibitem{rw} O. Riordan and N. C. Wormald. {\em The diameter of sparse random graphs.}
Combinatorics, Probability and Computing {\bf 19} (2010), 835~–~926.

\bibitem{mt} M. Talagrand.  {\em Concentration of measure and isoperimetric
inequalities in product spaces.} Institut Des Hautes \'{E}tudes Scientifiques,
Publications Math\'{e}matiques {\bf 81} (1995), 73~-~205.

\bibitem{zk} L. Zdeborov\'a and F. Krzakala. {\em Phase transitions in the colouring of random graphs.}
Phys. Rev. E 76, 031131 (2007).

\end{thebibliography}
\end{document}